\documentclass[pdflatex,bst/sn-mathphys-num]{sn-jnl}
\usepackage{graphicx}%
\usepackage{multirow}%
\usepackage{amsmath,amssymb,amsfonts}%
\usepackage{amsthm}%
\newtheorem{corollary}{Corollary}

\newtheorem{assumption}{Assumption}
\usepackage{mathrsfs}%
\usepackage[title]{appendix}%
\usepackage{xcolor}%
\usepackage{textcomp}%
\usepackage{manyfoot}%
\usepackage{booktabs}%
\usepackage{algorithm}%
\usepackage{algorithmicx}%
\usepackage{algpseudocode}%
\usepackage{listings}%
\usepackage{natbib}
\usepackage{subfig}
\usepackage{mathtools}
\usepackage{float}
\DeclareMathOperator*{\argminA}{arg\,min}

\theoremstyle{thmstyleone}%
\newtheorem{theorem}{Theorem}
\newtheorem{lemma}{Lemma}
\newtheorem{proposition}{Proposition}%
\theoremstyle{thmstyletwo}%
\newtheorem{example}{Example}%
\newtheorem{remark}{Remark}%
\theoremstyle{thmstylethree}%
\begin{document}
\title[2D Stability Selection: Design Jittering to Exclude Correlated Irrelevant Variables]{2D Stability Selection: Design Jittering to Exclude Correlated Irrelevant Variables}
\author*[1]{\fnm{Mahdi} \sur{Nouraie}}\email{soheil.nouraie@mq.edu.au}
\author[1]{\fnm{Houying} \sur{Zhu}}
\author[1,2]{\fnm{Samuel} \sur{Muller}}
\affil[1]{\orgdiv{School of Mathematical and Physical Sciences}, \orgname{Macquarie University}, \orgaddress{\city{Sydney}, \postcode{2109}, \state{NSW}, \country{Australia}}}
\affil[2]{\orgdiv{School of Mathematics and
Statistics}, \orgname{The University of Sydney}, \orgaddress{\city{Sydney}, \postcode{2006}, \state{NSW}, \country{Australia}}}
\abstract{Stability Selection finds the variables that are chosen consistently across random subsamples of the data, but it can consistently choose the wrong ones. An irrelevant variable that is strongly correlated with relevant variables can be selected in many subsamples, because resampling the observations does not change the variables themselves. We call such a variable a confounder; the term refers only to this correlation and not to causal confounding. We propose 2D Stability Selection, which also adds random noise to the design matrix, at increasing noise levels, and keeps the variables that are selected throughout. On average, the added noise acts as a ridge penalty, so each fit is a noisy elastic net. This changes the conditions under which the Lasso selects the correct variables: it can restore them when a confounder would otherwise be chosen, and it can also make selection worse. For uncorrelated relevant variables we give an exact criterion for the noise levels at which the population elastic-net irrepresentable condition holds, recovery guarantees for a single fit and for the aggregated procedure, and, for thresholds above $1/2$, a bound on the expected number of false positives (under i.i.d.\ observations; its closed form also needs exchangeable irrelevant variables). In simulations, the method excludes the confounders that Stability Selection includes. In the simulated design, an elastic net used without any noise, with a fixed ridge weight or along a deterministic path, excludes the confounders as well, which suggests that for the Lasso the ridge effect, not the randomness of the noise, is responsible.}
\keywords{Stability Selection, Variable Selection, Irrepresentable Condition, Confounders, Design Jittering, Measurement Error, High-Dimensional Regression, Elastic Net, Tikhonov Regularisation}
\pacs[MSC 2020 Classification]{Primary 62J07; secondary 62F07, 62J05}
\maketitle
\section{Introduction}\label{s1}
\subsection{Problem and approach}
Variable selection in high-dimensional linear regression asks which of $p$ variables are related to a response, using $n$ observations, where $p$ may be much larger than $n$. The Lasso \citep{tibshirani1996regression} is the standard tool, and the elastic net \citep{zou2005regularization} adds a ridge penalty to it. Selection by the Lasso can be unstable, because small changes in the data change the selected set. Stability Selection \citep{meinshausen2010stability} addresses this by fitting the Lasso to many random half-samples and keeping the variables that are selected in a large fraction of them. Its complementary-pairs version \citep{shah2013variable} bounds the expected number of selected variables with low selection probability, and the expected number of false positives under further conditions.

Resampling changes which observations are used, but not the variables. Suppose an irrelevant variable is strongly correlated with some of the relevant variables. The irrepresentable condition \citep{zhao2006model} states that $|\Sigma_{kS}\Sigma_{SS}^{-1}\operatorname{sign}(\beta_S)|<1$ for every irrelevant variable $k$, where $\Sigma$ is the covariance matrix of the variables, $S$ is the set of relevant variables and $\beta_S$ their coefficients. If this condition fails for variable $k$, the Lasso can select $k$ whatever the sample size. Stability Selection then selects it in about as many half-samples as it selects the weaker relevant variables, and its false-positive bound does not apply to it. We call such a variable a \emph{confounder}. In this paper the word means only an irrelevant variable that is correlated with the relevant variables. It does not refer to confounding in the causal sense. The problem is practical: in genomic data, for example, probes that are correlated with truly relevant probes can be selected again and again, so that the selection is stable but wrong. This paper asks whether a selector can reject confounders without knowing which variables they are.

Our idea is to perturb the variables themselves. Before each fit, we add independent Gaussian noise with variance $\delta^2$ to the entries of the design matrix, which we call \emph{design jittering}. We repeat it over a grid of increasing noise levels $\delta$ and keep the variables that are selected in a large fraction of all fits. We call the procedure \emph{2D Stability Selection} (2D-SS), because stability is required along two axes: over random half-samples and over increasing noise. At $\delta=0$ it is Stability Selection (Figure~\ref{fig:schematic}).

The reason it can work is that noise added to the design acts, on average, as a ridge penalty \citep{bishop1995training,wager2013dropout}. A Lasso fit to the jittered data is therefore a noisy elastic net, and a ridge penalty changes the irrepresentable condition. For a suitable noise level, the condition can hold for a confounder where it failed before, but it can also fail where it held, so the noise does not help in every case.

The paper combines results from two literatures. From machine learning we take the equivalence between training with noise and ridge regularisation \citep{bishop1995training}. From statistics we take the irrepresentable condition, the theory of the elastic net \citep{jia2010model} and Stability Selection. To our knowledge, the first has not been used to study when variable selection is consistent.

The same noise is also related to measurement error. If the variables are observed with error, the error already in the data adds to the noise that we inject, and the procedure does not need an estimate of the error covariance. This is a secondary use, since the main aim is to exclude confounders.

We consider the model $Y=X_0\beta+\varepsilon$, where $X_0\in\mathbb R^{n\times p}$ is the design matrix with centred and scaled columns, $\mathbb E(\varepsilon)=0$ and $\mathrm{Var}(\varepsilon)=\sigma^2I_n$. The analyst observes $X=X_0+W_0$, where $W_0$ is measurement error ($W_0=0$ if there is none). The set of relevant variables is $S=\{j:\beta_j\neq0\}$ and $s=|S|$. The design is treated as fixed in part of the analysis and as Gaussian in another part (Section~\ref{s2}).
\begin{figure}[t]
\centering
\includegraphics[width=\textwidth]{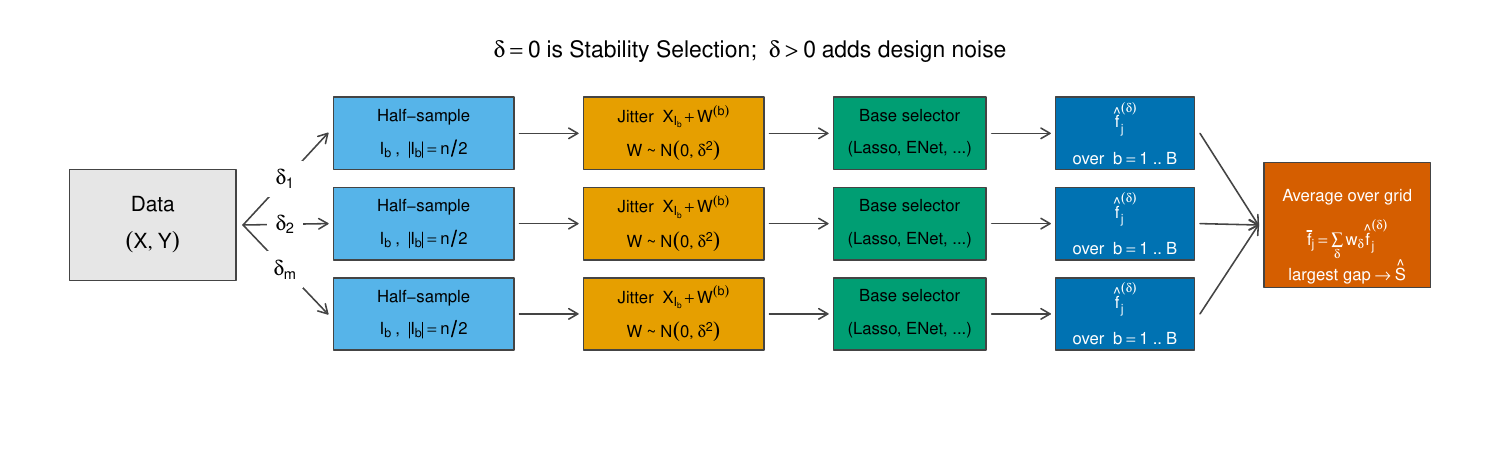}
\caption{2D-SS. For each noise level $\delta$ in the grid, $B$ half-samples are drawn, jittered with noise of variance $\delta^2$ and passed to the base selector. The selection frequencies are averaged over the grid with weights $w_\delta$ and thresholded by the largest-gap rule. The level $\delta=0$ is Stability Selection.}
\label{fig:schematic}
\end{figure}

\subsection{Contributions}
The paper has three contributions.

\emph{(i) Mechanism and repair criterion.} A jittered Lasso fit is, in expectation, an elastic net with ridge weight $\nu^2=\delta_{\mathrm{obs}}^2+\delta^2$, the sum of the variance of the measurement error in the data and of the injected noise (Proposition~\ref{prop:ridge}). Support recovery is then governed by an elastic-net irrepresentable condition, which moves from the Lasso condition at $\nu=0$ to a marginal-correlation condition as $\nu\to\infty$ (Proposition~\ref{prop:repair}). For uncorrelated relevant variables we give the exact noise level that repairs a violated variable and a criterion for repairing all violated variables at once (Theorem~\ref{thm:calibration}). Corollary~\ref{cor:compound} gives the analogous criterion for compound-symmetric correlation when the noise is injected along the all-ones direction of $S$ (an analysis device, not the isotropic noise of 2D-SS), and Example~\ref{ex:corr} shows that isotropic noise can also hurt.

\emph{(ii) A procedure.} 2D-SS applies to any base selector (the variable selection method, such as the Lasso, the elastic net or SCAD, that is fitted to each noisy half-sample), contains Stability Selection as the case $\delta=0$, and needs no estimate of the error covariance (Section~\ref{s2}). For the Lasso, the ridge effect is deterministic, so the random noise can be replaced by an explicit elastic-net path (ENet-SS, Remark~\ref{r:enetss}) or by a fixed elastic net (SS-ENet). In our simulations these exclude confounders as well as 2D-SS (Study~6). The paper is a diagnosis as much as a method: the elastic-net irrepresentable condition identifies when and why the ridge term removes a confounder, and the recovery and error-control guarantees are proved for the jittered procedure.

\emph{(iii) Guarantees.} For a single jittered fit, the elastic-net irrepresentable condition continues to hold under Gaussian noise and subsampling, and the support is recovered exactly at suitable noise levels (Lemma~\ref{l1}, Theorems~\ref{thr:hp_gauss_final}, \ref{thm:single_level}, \ref{thm:gauss_rate} and~\ref{thm:exact}, Corollary~\ref{cor:zeroone}). For the aggregated procedure, single-level guarantees on somewhat more than half of the grid weight separate the averaged frequencies and recover $S$, and the largest-gap rule estimates the model size under a gap condition (Proposition~\ref{prop:bridge}, Theorems~\ref{thm:delta_avg_recovery} and~\ref{thm:gap}). A complementary-pairs bound on the expected number of false positives holds for thresholds above $1/2$ (Theorem~\ref{thm:errctrl}); it assumes i.i.d.\ observations, its closed form also assumes exchangeable irrelevant variables, and it does not cover the largest-gap rule. The Appendix extends the mechanism to other base selectors (Proposition~\ref{prop:other}).

\subsection{Motivating example}
Figure~\ref{fig:ic_violated} shows the effect in a simulated example with $n=100$ and $p=1{,}000$. Five variables are relevant and uncorrelated. Ten irrelevant variables are correlated with them so that the population irrepresentable statistic is $1.153>1$, and the condition fails. The design is described in the Appendix (subsection ``Motivating example: design and computation''). We fitted the Lasso to $B=100$ random half-samples at each of $14$ noise levels $\delta\in[0,3]$.

Without noise ($\delta=0$, which is Stability Selection), the ten confounders are selected in $47\%$ of the half-samples on average. This is close to the $56\%$ of the four weak relevant variables. As $\delta$ increases, the confounders disappear: their mean frequency is $0.09$ at $\delta=0.5$ and at most $0.04$ for $\delta\ge0.75$. The relevant variables stay for longer: the strong one is selected in every half-sample for $\delta\le0.75$, and the weak ones keep a mean frequency of $0.54$--$0.62$ for $\delta\le0.4$. For large $\delta$ the relevant variables are lost as well (the weak ones are at $0.04$ at $\delta=3$). Around $\delta\approx0.4$--$0.75$, the relevant variables can therefore be separated, on average, from the correlated irrelevant ones. This motivates selecting variables by how they behave along the whole noise path, not at a single noise level; Section~\ref{s2:small} explains the mechanism.
\begin{figure}[H]
\centering
\includegraphics[width=0.65\textwidth]{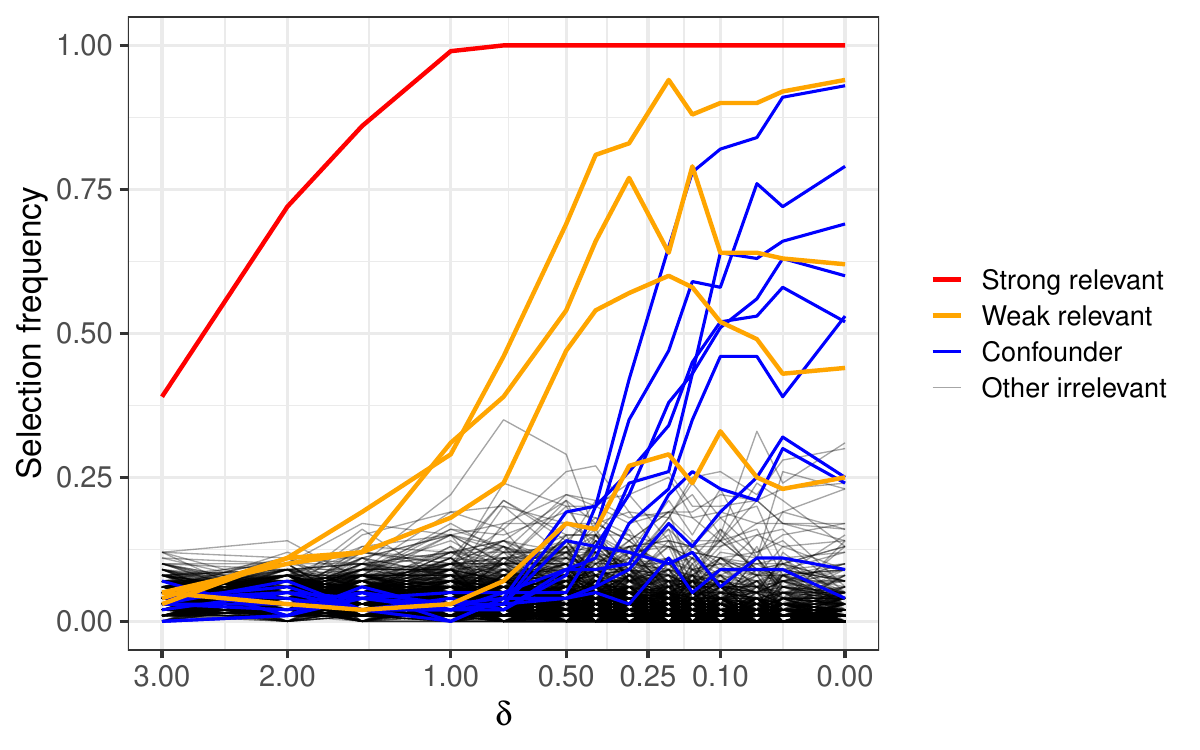}
\caption{Selection frequency of the Lasso from jittered half-samples against the noise level $\delta$ (square-root horizontal axis, reversed: $\delta=3$ at the left, $\delta=0$ at the right). The design violates the irrepresentable condition ($1.153>1$). Shown are the strong relevant variable (red), four weak relevant variables (orange), ten confounders (blue) and $985$ other irrelevant variables (black). $n=100$, $p=1{,}000$, $B=100$ half-samples per $\delta$, $\lambda=\sqrt{\log(p)/n_1}$ fixed.}\label{fig:ic_violated}
\end{figure}

\subsection{Related work}\label{s1:related}
\emph{Stability and perturbation.} Stability Selection \citep{meinshausen2010stability} and the randomised Lasso resample or reweight to stabilise selection, and \citet{shah2013variable} relax its assumptions with complementary pairs and sharpen its error bound under shape restrictions. Integrated path Stability Selection \citep{melikechi2026integrated} integrates selection probabilities along the regularisation path and obtains much tighter bounds on the expected number of false positives. It varies the penalty and leaves the design unchanged, whereas we vary the noise level of the design.

\emph{Noise and regularisation.} Training with noise added to the variables is equivalent, to leading order, to Tikhonov regularisation \citep{bishop1995training}, and in the linear case to ridge regression in expectation \citep{wager2013dropout}. \citet{xu2008robust} show that robust least squares under a worst-case, feature-wise norm-bounded perturbation of the design is equivalent to $\ell_1$-regularised regression. Our perturbation is random and is added on top of a Lasso fit, and the ridge term it produces changes the irrepresentable condition. The condition that results coincides with the elastic irrepresentable condition of \citet{jia2010model} under a change of parametrisation. In variable screening, \citet{10.1111/rssb.12127} (see the corrigendum \citealp{10.1111/rssb.12427}) relate marginal screening and their HOLP estimator as two ends of a ridge path, which is the same pattern as in Section~\ref{s2:small} (Remark~\ref{r:holp}).

\emph{Measurement error.} Errors-in-variables models have a long history \citep{fuller1987measurement,carroll2006measurement}. SIMEX \citep{cook1994simulation} adds simulated noise at increasing levels and extrapolates the estimates back to the noise-free case. Its high-dimensional versions, SIMSELEX \citep{nghiem2019simselex} and boosting with SIMEX or regression calibration \citep{chen2023denoising}, require the error covariance to be known or estimable, and the SIMEX-based ones extrapolate. We use the path to aggregate selection frequencies and do not extrapolate. Other methods aim to recover the clean-design model: the corrected Lasso \citep{loh2012high} and CoCoLasso \citep{datta2017cocolasso} correct the Gram matrix for a known error covariance, and the matrix uncertainty selector \citep{rosenbaum2010sparse} uses a bound on the size of the error. 2D-SS fits the uncorrected design and needs neither. Remark~\ref{r:cocolasso_compare} compares their sign-consistency guarantees with ours, and the three methods are comparators in Study~4.

A note on names: \citet{li2026justaddnoisedebiasing} also use the word ``jittering'', but for categorical variables in tree ensembles. Their estimator and conditions are unrelated to ours, although they also combine the noise with integrated path stability selection.

\subsection{Organisation of the paper}
Section~\ref{s2} defines 2D-SS, relates the jittered Lasso to ridge regression and the elastic net, and gives the theory. Section~\ref{s3} reports simulation and real-data results. Section~\ref{s4} discusses limitations and future work. Proofs, additional results and a table of notation are in the appendices.

\section{Methodology}\label{s2}
This section first defines 2D-SS, together with the effective noise level through which measurement error and injected noise act jointly. Section~\ref{s2:small} then shows how the jittered Lasso is related to ridge regression and the elastic net and identifies the condition that governs its support recovery; Sections~\ref{s2:hp} and~\ref{s2:single} give recovery guarantees for a single fit, {Section~\ref{s2:delta_avg} gives those for the aggregated frequencies,} and Section~\ref{s2:errctrl} gives error control.

\noindent\textbf{Setting.} The response follows $Y=X_0\beta+\varepsilon$ with a deterministic latent design $X_0\in\mathbb R^{n\times p}$ whose entries satisfy $|X_{0,ij}|\le K$, and $\varepsilon\sim\mathcal N(0,\sigma^2I_n)$. The analyst observes $X=X_0+W_0$, where $W_0$ has i.i.d.\ $\mathcal N(0,\delta_{\mathrm{obs}}^2)$ entries independent of $\varepsilon$ (classical additive measurement error; $\delta_{\mathrm{obs}}=0$ is the error-free case, and $\delta_{\mathrm{obs}}$ need not be known). The active set $S=\operatorname{supp}(\beta)=\{j:\beta_j\neq0\}$, $s=|S|$, refers to the latent model, and $\Sigma\coloneq\frac1nX_0^\top X_0$ denotes the latent Gram matrix. Boundedness of the design is used only for Hoeffding-type bounds; sub-Gaussian designs can be handled analogously with different constants. Section~\ref{s2:single} also considers a Gaussian random design, in which case $\Sigma$ denotes the population covariance matrix.

\noindent\textbf{Procedure (subsampled jittering).} Let $n_1\coloneq\lfloor n/2\rfloor$. For every $\delta\in\mathcal G$ (the grid $\mathcal G$ may contain $\delta=0$) and every $b=1,\dots,B$, we draw a random half-sample $I_b\subset\{1,\dots,n\}$ with $|I_b|=n_1$, draw $W^{(b)}\in\mathbb R^{n_1\times p}$ with i.i.d.\ $\mathcal N(0,\delta^2)$ entries, and fit the base selector to the observed data $(X_{I_b}+W^{(b)},Y_{I_b})$ (Algorithm~\ref{alg:delta_averaged_selection}, which also contains the data-driven threshold of Section~\ref{s2:delta_avg}). At $\delta=0$ this is exactly Stability Selection at the fixed penalty $\lambda$ (not the union over a $\lambda$-path or the randomised Lasso). The randomness has two sources: the subsample $I_b$ (resampling of the observed data) and the noise $W^{(b)}$ (design perturbation). subsampling one observed data set probes stability under resampling, as in Stability Selection; it does not by itself guarantee stability across new draws from the population. We call the procedure \emph{2D Stability Selection} (2D-SS); \emph{design jittering} refers only to its noise-injection step.

\noindent\textbf{Effective noise level.} Since $X_I+W=X_{0,I}+(W_{0,I}+W)$ and $\widetilde W\coloneq W_{0,I}+W$ has i.i.d.\ $\mathcal N(0,\nu_\delta^2)$ entries with
\[
\nu_\delta^2\coloneq\delta_{\mathrm{obs}}^2+\delta^2,
\]
a single jittered fit is a Gaussian perturbation of the latent design at the effective level $\nu_\delta$. The noise path therefore cannot start below $\nu_0=\delta_{\mathrm{obs}}$. For fixed $\delta$ we write $\nu=\nu_\delta$, and for a fixed half-sample $I$ we write $X_{0,I}$ for its latent design and $\Sigma_I=\frac1{n_1}X_{0,I}^\top X_{0,I}$. For high-probability statements we use the Gaussianity of $\varepsilon$ and $W_0$.
\begin{algorithm}[t]
\caption{Two-dimensional Stability Selection (2D-SS)}
\label{alg:delta_averaged_selection}
\begin{algorithmic}[1]
\State \textbf{Input:} Observed data $(X,Y)$; grid $\mathcal{G}=\{\delta_1,\dots,\delta_m\}$ (may contain $\delta=0$) with weights $w_\delta>0$, $\sum_\delta w_\delta=1$ (default $w_\delta=1/m$); number of replicates $B$; half-sample size $n_1=\lfloor n/2\rfloor$
\State \textbf{Input:} Base selector with fixed tuning parameters (e.g.\ fixed $\lambda$)
\For{each $\delta \in \mathcal{G}$}
    \For{$b=1$ to $B$}
        \State Draw a half-sample $I_b\subset\{1,\dots,n\}$ uniformly among subsets of size $n_1$ (for error control, $B$ counts pairs: draw a complementary pair $(I_{2b-1},I_{2b})$ and fit both halves, so there are $2B$ half-samples per level; Section~\ref{s2:errctrl})
        \State Generate $W^{(b)}$ with i.i.d.\ $\mathcal{N}(0,\delta^2)$ entries of dimension $n_1\times p$
        \State Form $X^{(b,\delta)} = X_{I_b} + W^{(b)}$
        \State Fit selector to $(X^{(b,\delta)},Y_{I_b})$ to obtain $\hat\beta^{(b,\delta)}$
    \EndFor
    \State Compute frequencies $\hat f_j^{(\delta)} = \frac{1}{B}\sum_{b=1}^B 1\{\hat\beta_j^{(b,\delta)}\neq 0\}$ for all $j$
\EndFor
\State Compute $\delta$-averaged frequencies $\widehat{\bar f}_j = \sum_{\delta\in\mathcal G}w_\delta\hat f_j^{(\delta)}$ for all $j$
\State Sort $\{\widehat{\bar f}_j\}_{j=1}^p$ in decreasing order to obtain $\widehat{\bar f}_{(1)} \ge \cdots \ge \widehat{\bar f}_{(p)}$
\State Compute $\hat s = \arg\max_{1 \le k < p} \big(\widehat{\bar f}_{(k)} - \widehat{\bar f}_{(k+1)}\big)$
\State Set $\hat\tau = \frac{1}{2}\big(\widehat{\bar f}_{(\hat s)} + \widehat{\bar f}_{(\hat s+1)}\big)$
\State Let $\hat S=\{j:\widehat{\bar f}_j\ge \hat\tau\}$
\State \textbf{return} $\hat S$
\end{algorithmic}
\end{algorithm}
\subsection{Relating the Lasso to Ridge and the elastic net: Jittering as Tikhonov Regularisation}\label{s2:small}
We use the Lasso as our running example{; Section~\ref{s2:other} extends the analysis to the elastic net, weighted $\ell_1$ penalties and folded-concave penalties}.

\noindent\textbf{Background and notation.} The Lasso augments the least-squares loss with an $\ell_1$ penalty to promote sparse coefficients. The estimator is
\begin{equation*}
\hat{\beta}(\lambda) = \argminA_{\beta \in \mathbb{R}^{p}} \left(\frac{1}{2n}\|Y - X\beta\|_{2}^{2} + \lambda \|\beta\|_{1}\right),
\end{equation*}
where $\lambda \in \mathbb{R}^{+}$ denotes the Lasso regularisation parameter. {The model has no intercept: we take the latent design $X_0$ to be centred and scaled and the response $Y$ to be centred. Our code re-standardises each jittered half-sample before fitting. Dividing all columns by a constant $c$ is equivalent to multiplying the penalty by $c$. The columns of the jittered design have variance close to $1+\nu^2$, so re-standardising amounts to the level-dependent penalty $\lambda\sqrt{1+\nu^2}$ in place of $\lambda$ (Remark~\ref{r:tuning}). The theorems do not cover the sampling fluctuation of the empirical means and variances in this step.} We focus on the Lasso because of its central role in variable selection and its well-studied selection consistency conditions. We first show that these conditions are robust to sufficiently small perturbations of the design matrix, which provides a theory-grounded basis for studying stability along a path of progressively increasing noise. Throughout, $\|\cdot\|_\infty$ denotes the induced matrix $\ell_\infty$ norm, i.e., the maximum absolute row sum, and the sup-norm on vectors, so that $\|Mv\|_\infty\le\|M\|_\infty\|v\|_\infty$; if a matrix $A$ has $q$ columns and $\max_{j,k}|A_{jk}|\le t$, then $\|A\|_\infty \le q\,t$. The selector uses the Gram matrix $\Sigma^{(\delta)}$ of \eqref{eq:gram} as it is, without bias correction. {The fixed-design analysis below bounds the fluctuation of this matrix entrywise, which is conservative. Matrix concentration inequalities bound it in operator norm, but a sharper recovery condition also needs the structure of the dual vector. Under Gaussian design we use this structure, and the dual-feasibility conditions are then of order $s\log p$ when the ridge-bias term is small (Theorem~\ref{thm:gauss_rate} and Remark~\ref{r:tuning} in the Appendix).} For the Gaussian-design results we use the singular-value bounds for Gaussian matrices \citep[Cor.~5.35]{vershynin2012introduction} and the $\chi^2$ tail bounds of \citet[Lemma~1]{laurent2000adaptive}.
For the Lasso, exact support recovery is typically established under the irrepresentable condition (IC) and a minimum signal condition, together with a choice of $\lambda$ in an appropriate range \citep{zhao2006model}, with neighbourhood stability as the closely related condition in \citet{meinshausen2006high} {and \citet[Sections~2.6.1 and~7.4--7.5]{buhlmann2011statistics}}. The IC can be stated as follows. Let $\Sigma \in \mathbb{R}^{p \times p}$ be a Gram matrix, let $S \subset \{1,\dots,p\}$ be the active set and $S^c$ its complement. Suppose $\Sigma_{SS}$, the sub-matrix of $\Sigma$ indexed by $S$, is invertible and
\[
\|\Sigma_{S^c S}\Sigma_{SS}^{-1}\|_\infty \le 1 - \eta,
\]
for some $\eta \in (0,1)$. This condition constrains the dependence between relevant and irrelevant variables relative to the dependence within the relevant set.

\noindent\textbf{What the jittered fit actually estimates.} The Lasso fitted on $(X_I+W,Y_I)$ does not see the latent Gram matrix $\Sigma$, nor an unbiased version of it. The following identity, the classical statement that training with noise added to the variables is Tikhonov regularisation \citep{bishop1995training}, {which is exact for linear models \citep{wager2013dropout},} shows what it sees instead.
\begin{proposition}[Noise added to the design is Tikhonov regularisation in expectation]\label{prop:ridge}
Let $I$ be a fixed half-sample and let $\widetilde W\in\mathbb R^{n_1\times p}$ have independent entries with mean $0$ and variance $\nu^2$. Then for every $\beta\in\mathbb R^p$,
\[
\begin{aligned}
\mathbb E\!\left[\frac1{2n_1}\|Y_I-(X_{0,I}+\widetilde W)\beta\|_2^2\right]
&=\frac1{2n_1}\|Y_I-X_{0,I}\beta\|_2^2+\frac{\nu^2}{2}\|\beta\|_2^2 .
\end{aligned}
\]
Hence the expected Lasso objective on the perturbed latent design is the elastic-net objective $\frac1{2n_1}\|Y_I-X_{0,I}\beta\|_2^2+\lambda\|\beta\|_1+\frac{\nu^2}{2}\|\beta\|_2^2$ {(the naive elastic net of \citet{zou2005regularization}, whose solution the elastic net rescales by a constant factor, their equation~(12), so that the selected model is the same)}, and each jittered fit is a noisy realisation of an elastic net with ridge weight $\nu^2=\delta_{\mathrm{obs}}^2+\delta^2$. {(The expectation is over the total noise $\widetilde W$, with $X_{0,I}$ and $Y_I$ held fixed; it is an analysis device. In the procedure itself $W_0$ is a single realisation, frozen across replicates, and the theorems below control the fluctuation of the Gram matrix around $\Sigma+\nu^2I_p$ over the joint randomness of $(W_0,W)$; Section~\ref{s2:delta_avg} accounts for the dependence between replicates by conditioning on the data.)}
\end{proposition}
\begin{proof}
Write $Y_I-(X_{0,I}+\widetilde W)\beta=(Y_I-X_{0,I}\beta)-\widetilde W\beta$. Since $\mathbb E(\widetilde W)=0$, the cross term vanishes in expectation, and $\widetilde W\beta\in\mathbb R^{n_1}$ has independent entries with mean $0$ and variance $\nu^2\|\beta\|_2^2$, so $\mathbb E\|\widetilde W\beta\|_2^2=n_1\nu^2\|\beta\|_2^2$. Dividing by $2n_1$ gives the claim.
\end{proof}
\noindent\textbf{The Lasso--ridge--elastic-net picture.} Proposition~\ref{prop:ridge} places the Lasso, ridge regression and the elastic net on a single path indexed by the effective noise level $\nu$. At $\nu=0$ there is no ridge term and the fit is the Lasso. For $0<\nu<\infty$ the expected fit is an elastic net whose ridge weight $\nu^2=\delta_{\mathrm{obs}}^2+\delta^2$ grows with the noise. As $\nu\to\infty$ the ridge term dominates, and, as shown below, the support-recovery condition moves from the Lasso irrepresentable condition to a marginal-correlation condition. This is what the noise path of 2D-SS traces out. For a single fit, the measurement error already present in the data and the injected noise are not distinguished: both enter the law of the fit only through $\nu$ (across replicates, however, the measurement error is frozen while the injected noise is redrawn).

\begin{remark}[A deterministic counterpart: ENet-SS]\label{r:enetss}
Proposition~\ref{prop:ridge} says that, in expectation, noise at level $\delta$ turns the Lasso into an elastic net with ridge weight $\nu^2=\delta_{\mathrm{obs}}^2+\delta^2$. Algorithm~\ref{alg:delta_averaged_selection} can therefore be run with the injected noise replaced by this elastic net: at each level $\delta$, fit on every half-sample the elastic net with ridge weight $\delta^2$ (with respect to the observed half-sample, which already carries the measurement error; and the penalty $\lambda\sqrt{1+\delta^2}$ implied by the re-standardisation of the jittered columns), and aggregate the selection frequencies as before. We call the resulting procedure \emph{ENet-SS}; it equals Stability Selection at $\delta=0$ and has no randomness beyond the subsampling. The two procedures share the expected objective but not the fits: the injected noise also perturbs the Gram matrix at random in each replicate, through the fluctuation $\Sigma^{(\delta)}-\Sigma-\nu^2I_p$ in \eqref{eq:gram}, which ENet-SS lacks. Whether this fluctuation helps is an empirical question, taken up in Study~6; for the Lasso, in the design studied there, it does not help to exclude confounders. Study~6 also runs Stability Selection with a fixed elastic net (SS-ENet) and 2D-SS with that elastic net as base selector (2D-SS-ENet).
\end{remark}

2D-SS deliberately fits the \emph{uncorrected} Gram matrix, so that the ridge term of Proposition~\ref{prop:ridge} is present, in contrast to the bias-corrected methods of \citet{loh2012high,datta2017cocolasso} and to SIMEX \citep{cook1994simulation}, which extrapolates back to the noise-free case (Remark~\ref{r:simex}, Appendix).
Correspondingly, the Gram matrix that the selector sees splits into a fixed part and a mean-zero fluctuation,
\begin{equation}\label{eq:gram}
\begin{split}
\Sigma^{(\delta)}&\coloneq\frac1{n_1}(X_I+W)^\top(X_I+W)=\underbrace{\Sigma+\nu^2 I_p}_{\Sigma_\nu}+\widetilde\Delta,\\
\widetilde\Delta&=\big(\Sigma_I-\Sigma\big)+\frac1{n_1}\big(X_{0,I}^\top \widetilde W+\widetilde W^\top X_{0,I}\big)+\Big(\frac1{n_1}\widetilde W^\top \widetilde W-\nu^2 I_p\Big),
\end{split}
\end{equation}
with $\mathbb E(\widetilde\Delta)=0$: the three terms are due to subsampling, to the cross terms between the latent design and the total noise, and to the noise Gram matrix. The population target of the jittered fit is therefore governed by $\Sigma+\nu^2I_p$, not by $\Sigma$, and in general it is not supported on $S$. The relevant support-recovery condition is thus not the Lasso IC, which is the case $\nu=0$, but an irrepresentable condition of elastic-net type. With $z\coloneq\operatorname{sign}(\beta_S)$, let $v_\nu\coloneq z+\frac{\nu^2}{\lambda}\beta_S$ and define
\[
c_\nu\coloneq\big\|\Sigma_{S^cS}\big(\Sigma_{SS}+\nu^2I_s\big)^{-1}v_\nu\big\|_\infty .
\]
We say that the \emph{elastic-net irrepresentable condition at level $\nu$}, EIC$_\nu$, holds with margin $\eta\in(0,1)$ if $c_\nu\le1-\eta$. Here $\|\cdot\|_\infty$ is the induced $\ell_\infty$ norm on matrices and the sup-norm on vectors. EIC$_\nu$ is the strict dual feasibility condition of the primal--dual witness for the expected objective in Proposition~\ref{prop:ridge}, after also averaging over the half-sample $I$. The subsampling error is the Gram term $\Sigma_I-\Sigma$ of \eqref{eq:gram}. For the jittered Lasso itself, the same computation with $\Sigma^{(\delta)}$ in place of $\Sigma_\nu$ gives, up to a mean-zero remainder, a bound on $\|\Sigma^{(\delta)}_{S^cS}(\Sigma^{(\delta)}_{SS})^{-1}v_\nu\|_\infty$. Theorem~\ref{thm:single_level} makes this precise. As $\nu$ grows, EIC$_\nu$ moves between two known conditions:
\begin{itemize}
\item at $\nu=0$, $v_0=z$ and EIC$_0$ is the (sign-specific) Lasso IC, which is implied by the matrix IC displayed above;
\item as $\nu\to\infty$, $(\Sigma_{SS}+\nu^2I)^{-1}v_\nu\to\beta_S/\lambda$, so $c_\nu\to\|\Sigma_{S^cS}\beta_S\|_\infty/\lambda$, a marginal-correlation (screening) condition.
\end{itemize}
Small effective noise thus leaves the Lasso IC in force, and large noise pushes the Lasso towards marginal screening.

\noindent\textbf{Connection to \citet{jia2010model}.} \citet{jia2010model} study the elastic net $\|Y-X\beta\|_2^2+\lambda_2\|\beta\|_2^2+\lambda_1\|\beta\|_1$ (their eq.~(1.2)) and give the Elastic Irrepresentable Condition (their eq.~(1.4)),
\[
\Big\|C_{21}\Big(C_{11}+\frac{\lambda_2}{n}I\Big)^{-1}\Big(\operatorname{sign}(\beta_{(1)})+\frac{2\lambda_2}{\lambda_1}\beta_{(1)}\Big)\Big\|_\infty\le1-\eta,\qquad C=\tfrac1nX^\top X .
\]
Multiplying the expected objective of Proposition~\ref{prop:ridge} by $2n_1$ gives their parametrisation with $n=n_1$, $\lambda_2=n_1\nu^2$ and $\lambda_1=2n_1\lambda$. Then
\[
\frac{\lambda_2}{n_1}=\nu^2,\qquad \frac{2\lambda_2}{\lambda_1}=\frac{\nu^2}{\lambda},
\]
so EIC$_\nu$ is their condition with the latent Gram matrix $\Sigma$ in place of $C$. Their Lemma~1, proved in their appendix, is a KKT characterisation of exact sign recovery for the elastic net (their eqs.~(2.1)--(2.2)); it is the elastic-net counterpart of the primal--dual witness construction of \citet{wainwright2009sharp}, which we use in the proof of Theorem~\ref{thm:single_level}. Their events $\mathcal M(V)$ and $\mathcal M(U)$ (their eqs.~(2.5)--(2.6)) correspond to our dual feasibility and sign consistency steps, and their constraint (2.7) on $\lambda_2$ to our sign margin $\theta_\nu$ below. Their results concern the elastic net on a fixed design with Gaussian response noise, whereas the jittered fit is a Lasso on a randomly perturbed design, whose Gram matrix and response remainder fluctuate around the elastic-net quantities; Lemma~\ref{l1} and Theorems~\ref{thr:hp_gauss_final}--\ref{thm:single_level} quantify these fluctuations non-asymptotically. \citet{jia2010model} also show three things. If the IC holds, then for any $\lambda_1>0$ there is a $\lambda_2$ (a small one suffices) for which their condition holds. Their condition does not imply the IC (their Proposition~1). And it can hold when the IC fails (their Theorem~2, for one irrelevant variable). The last result offers a possible mechanism for the removal of the confounders at moderate noise in Figure~\ref{fig:ic_violated}.

The next proposition makes this mechanism explicit in our notation. For $k\in S^c$ let $a_k\coloneq\Sigma_{kS}$ (a row vector), and, when $\Sigma_{SS}$ is invertible, let $A_k\coloneq a_k\Sigma_{SS}^{-1}z$ and $B_k\coloneq a_k\Sigma_{SS}^{-1}\beta_S$. Here $A_k$ is the Lasso IC statistic of variable $k$, and $c_\nu^k\coloneq[\Sigma_{S^cS}(\Sigma_{SS}+\nu^2I_s)^{-1}v_\nu]_k$ is its elastic-net counterpart, so that $c_\nu=\max_{k\in S^c}|c_\nu^k|$.
\begin{proposition}[Jittering can repair a violated irrepresentable condition]\label{prop:repair}
Let $\lambda>0$, $\eta\in(0,1)$, let $\Sigma_{SS}$ be invertible and $T_\nu\coloneq\Sigma_{SS}+\nu^2I_s$.
\begin{enumerate}
\item[(a)] For every $\nu\ge0$,
\[
c_\nu^k=A_k+\frac{\nu^2}{\lambda}B_k-\nu^2a_k\Sigma_{SS}^{-1}T_\nu^{-1}v_\nu .
\]
\item[(b)] If $\Sigma_{SS}=I_s$, then
\[
c_\nu^k=\frac1{1+\nu^2}A_k+\frac{\nu^2}{1+\nu^2}\,\frac{B_k}{\lambda},
\]
a convex combination of the Lasso statistic $A_k$ (at $\nu=0$) and the marginal statistic $B_k/\lambda$ (as $\nu\to\infty$). Consequently, with $m\coloneq1-\eta$, we have $|c_\nu^k|\le m$ for some $\nu\ge0$ if and only if $A_k\in[-m,m]$, or $A_k>m$ and $B_k/\lambda<m$, or $A_k<-m$ and $B_k/\lambda>-m$. If $A_k>m$ and $B_k/\lambda<m$, the admissible values of $\nu^2$ form the interval
\[
\frac{A_k-m}{m-B_k/\lambda}\ \le\ \nu^2\ \le\ \frac{A_k+m}{-B_k/\lambda-m},
\]
whose upper endpoint is $+\infty$ unless $B_k/\lambda<-m$.
\item[(c)] If $\Sigma_{SS}=I_s$ and $A_kB_k<0$ (opposite signs), then $c_\nu^k$ is strictly monotone with a unique zero, attained at
\[
\nu_k^{\star2}=-\frac{A_k\lambda}{B_k}\ >\ 0,
\]
where the noise term $\nu^2B_k/\lambda$ exactly cancels $A_k$ in the numerator, so that $c_{\nu_k^\star}^k=0$.
\end{enumerate}
\end{proposition}
Proof is shown in the appendix. Part (a) is exact: the leading effect of jittering on the IC statistic, for $\nu^2$ of order $\lambda$, is governed by $\nu^2B_k/\lambda$ (the remaining term of (a) is then $O(\lambda)$), so a violated IC ($A_k>1$) can be repaired when $B_k$ is negative and large relative to $\lambda$, and, by (b), it cannot be repaired when $\Sigma_{SS}=I_s$ and $B_k/\lambda\ge m$ (with $m=1-\eta$). This is the sign pattern of Jia and Yu's conditions (their eqs.~(3.4)--(3.5)): as $\lambda\to0$, the admissible values of $\nu^2$ in (b) concentrate at $0$ at the rate $\lambda$ if $B_k<0$, and, for $B_k>0$, the interval is empty once $B_k/\lambda\ge m$ (in particular for all sufficiently small $\lambda$). For several irrelevant variables, the admissible sets of $\nu^2$ must intersect. Part (c) makes the single-variable repair exact: whenever a violated variable's marginal statistic $B_k$ has the sign needed to cancel $A_k$, there is a precise, closed-form noise level at which that variable's contribution to the elastic-net statistic vanishes.
The zero-crossing in part (c) is exact for one variable; for several irrelevant variables simultaneously, the relevant object is the intersection of their individual admissible ranges. This intersection question has a sharp, {explicit (in population quantities)} answer when $\Sigma_{SS}=I_s$, because monotonicity turns each admissible range into an interval.
For $k\in S^c$ and $\eta\in(0,1)$, write $\mathcal T_k(\eta)\coloneq\{\nu^2\ge0:|c_\nu^k|\le1-\eta\}$ for the set of noise levels at which variable $k$ alone meets the target margin $\eta$.
\begin{lemma}[Admissible noise levels form an interval]\label{lem:interval}
If $\Sigma_{SS}=I_s$, then for every $k\in S^c$ and $\eta\in(0,1)$, $\mathcal T_k(\eta)$ is a (possibly empty) sub-interval of $[0,\infty)$.
\end{lemma}
\begin{proof}
By Proposition~\ref{prop:repair}(b), $\nu^2\mapsto c_\nu^k$ is monotone (non-increasing or non-decreasing) or constant on $[0,\infty)$, being a fractional-linear (M\"obius) function of $\nu^2$ with derivative of constant sign. The preimage of the interval $[-(1-\eta),1-\eta]$ under a monotone (or constant) function on an interval domain is itself an interval, possibly empty.
\end{proof}
Since $S^c$ is finite, $\bigcap_{k\in S^c}\mathcal T_k(\eta)$ is nonempty if and only if every pair of the $\mathcal T_k(\eta)$ intersects: writing $\mathcal T_k(\eta)=[\ell_k,u_k]$ (with the conventions $\ell_k=+\infty$, $u_k=-\infty$ if $\mathcal T_k(\eta)=\emptyset$), pairwise intersection means $\ell_k\le u_j$ for all $k,j$, which gives $\max_k\ell_k\le\min_j u_j$, so every point of $[\max_k\ell_k,\min_ku_k]$ is common to all the intervals; the converse is immediate. This elementary fact turns the qualitative content of Corollary~\ref{cor:benefit} {(stated below)} into an exact, finite, {pairwise explicit} criterion.
\begin{theorem}[Simultaneous calibration criterion]\label{thm:calibration}
Assume $\Sigma_{SS}=I_s$ and fix $\eta\in(0,1)$. There exists $\nu\ge0$ with $c_\nu=\max_{k\in S^c}|c_\nu^k|\le1-\eta$ (that is, EIC$_\nu$ holds with margin $\eta$) if and only if
\[
\mathcal T_k(\eta)\cap\mathcal T_j(\eta)\neq\emptyset\qquad\text{for every }k,j\in S^c .
\]
Each $\mathcal T_k(\eta)$ is explicit (Proposition~\ref{prop:repair}(b) treats the case $A_k>m$, $B_k/\lambda<m$; the other cases follow in the same way): with $m=1-\eta$, $\mathcal T_k(\eta)=[0,\infty)$ if $A_k\in[-m,m]$ and $B_k/\lambda$ does not force an exit through $\pm m$; a half-line $[0,u_k]$ or $[\ell_k,\infty)$, or a bounded interval $[\ell_k,u_k]$, in the remaining cases; or $\emptyset$ if $A_k>m$ and $B_k/\lambda\ge m$, or $A_k<-m$ and $B_k/\lambda\le-m$ (obtained by solving $|c_\nu^k|\le m$).
\end{theorem}
\begin{proof}
The equivalence $c_\nu\le1-\eta\ \Leftrightarrow\ \nu^2\in\bigcap_{k\in S^c}\mathcal T_k(\eta)$ holds by definition of $c_\nu=\max_k|c_\nu^k|$. By Lemma~\ref{lem:interval}, each $\mathcal T_k(\eta)$ is an interval, so nonemptiness of the (finite) intersection is equivalent to pairwise nonemptiness, as shown above. The explicit form of $\mathcal T_k(\eta)$ follows from Proposition~\ref{prop:repair}(b) by solving $|c_\nu^k|\le m$ directly: writing $c(t)=\mathsf a+(\mathsf b-\mathsf a)t/(1+t)$ with $\mathsf a=A_k$, $\mathsf b=B_k/\lambda$ (an equivalent form of the expression in the proof of Proposition~\ref{prop:repair}), $c(t)\le m$ for all $t\ge0$ iff $\mathsf a\le m$ and $\mathsf b\le m$; $c(t)\le m$ for no $t\ge0$ iff $\mathsf a>m$ and $\mathsf b\ge m$; and otherwise $c(t)\le m$ exactly for $t$ on one side of the single crossing $t=(m-\mathsf a)/(\mathsf b-m)$ (and $t=(-m-\mathsf a)/(\mathsf b+m)$ for $-m$) (the sign of $\mathsf b-\mathsf a$ determining which side). The same case analysis applied to $c(t)\ge-m$ and intersecting the two resulting sets gives $\mathcal T_k(\eta)$ in every case, matching the interval displayed in Proposition~\ref{prop:repair}(b) when $A_k>m,B_k/\lambda<m$, and its mirror image when $A_k<-m,B_k/\lambda>-m$, and $[0,\infty)$ when $A_k\in[-m,m]$ and neither exit condition is triggered.
\end{proof}
Theorem~\ref{thm:calibration} strengthens Corollary~\ref{cor:benefit} (stated below) from an existence statement into a complete characterisation of when some $\nu$ gives EIC$_\nu$. Whether a \emph{single} noise level repairs \emph{every} violated variable at once is decided by comparing explicit intervals (each $\mathcal T_k(\eta)$ and each pair), with no search over $\nu$. The criterion is an oracle one, since it involves $a_k$, $\beta_S$, $\lambda$ and $S$, and recovery additionally needs the conditions of Theorem~\ref{thm:single_level}, and its scope is specific to $\Sigma_{SS}=I_s$. The procedure reaches only $\nu^2\ge\delta_{\mathrm{obs}}^2$, and only on the finite grid. For $\delta_{\mathrm{obs}}>0$ the criterion is therefore applied to the truncated sets $\mathcal T_k(\eta)\cap[\delta_{\mathrm{obs}}^2,\infty)$: each must be nonempty, and they must intersect pairwise; for $\delta_{\mathrm{obs}}=0$, as in our simulations, nothing changes. The criterion says that some $\nu$ exists, not that a given grid contains one.
\begin{remark}[Correlated relevant variables]\label{r:nonmono}
For general (correlated) $\Sigma_{SS}$, $c_\nu^k$ need not be monotone in $\nu^2$: writing $\Sigma_{SS}=UDU^\top$ with $D=\operatorname{diag}(d_1,\dots,d_s)$, one has $c_\nu^k=\sum_{i=1}^s[w_i+(\nu^2/\lambda)u_i]/(d_i+\nu^2)$ for scalars $w_i,u_i$ depending on $(a_k,z,\beta_S,U)$, a sum of fractional-linear terms whose individual monotonicity directions can differ across $i$ when the $d_i$ are unequal (i.e., when relevant variables are correlated). Consequently $\mathcal T_k(\eta)$ need not be an interval and Theorem~\ref{thm:calibration}'s pairwise criterion does not extend directly; characterising the correlated case is left for future work.
\end{remark}

This non-monotonicity is a property of \emph{isotropic} noise, however, not an obstruction intrinsic to correlated designs. Borrowing the strategy of \citet[proof of Theorem~2, Lemma~5]{meinshausen2010stability} for the randomised Lasso---perturb one coordinate (there, one variable's weight) while holding the rest fixed---consider injecting noise along a single eigendirection of $\Sigma_{SS}$ instead of isotropically. Fix $\ell\in\{1,\dots,s\}$ and let $e_\ell$ be the $\ell$th column of $U$ (a unit eigenvector of $\Sigma_{SS}$, with eigenvalue $d_\ell$). For $t\ge0$, let $\Theta_\ell(t)\coloneq te_\ell e_\ell^\top$ be a rank-one noise covariance on the $S$-block, i.e.\ each row of the injected noise on the relevant columns is $\sqrt t\,g_ie_\ell$ for i.i.d.\ $g_i\sim\mathcal N(0,1)$ (the $S^c$ columns may still be perturbed isotropically, as in Algorithm~\ref{alg:delta_averaged_selection}; this does not affect the argument below). By the same expected-loss computation as Proposition~\ref{prop:ridge}, the ridge term becomes $\frac t2(e_\ell^\top\beta_S)^2$ in place of $\frac{\nu^2}2\|\beta_S\|_2^2$, so the elastic-net statistic uses $v_{\Theta_\ell(t)}\coloneq z+\frac t\lambda(e_\ell^\top\beta_S)e_\ell$ in place of $v_\nu$, and the reference matrix becomes $\Sigma_{SS}+\Theta_\ell(t)$ in place of $\Sigma_{SS}+\nu^2I_s$.

\begin{proposition}[Direction-aligned repair for correlated designs]\label{prop:directional}
For $k\in S^c$, define
\[
L_k^{(\ell)}\coloneq A_k+(a_ke_\ell)\Big(\frac{e_\ell^\top\beta_S}\lambda-\frac{e_\ell^\top z}{d_\ell}\Big).
\]
Then, writing $c_k^{(\ell)}(t)\coloneq a_k(\Sigma_{SS}+\Theta_\ell(t))^{-1}v_{\Theta_\ell(t)}$,
\[
c_k^{(\ell)}(t)=\frac{A_kd_\ell+L_k^{(\ell)}t}{d_\ell+t},\qquad t\ge0,
\]
for \emph{every} $k\in S^c$ and every choice of $\ell$. Consequently, after the rescaling $t'=t/d_\ell$, Proposition~\ref{prop:repair}(b)--(c), Lemma~\ref{lem:interval} and Theorem~\ref{thm:calibration} apply verbatim with $B_k/\lambda$ replaced by $L_k^{(\ell)}$: in particular, if $A_kL_k^{(\ell)}<0$, $c_k^{(\ell)}$ has a unique zero at $t=-A_kd_\ell/L_k^{(\ell)}>0$, and a single $t\ge0$ repairs every variable in $S^c$ simultaneously to margin $\eta$ if and only if the resulting family of intervals $\{\mathcal T_k^{(\ell)}(\eta)\}_{k\in S^c}$ is pairwise intersecting.
\end{proposition}
\begin{proof}
{Write $\Sigma_{SS}=\sum_id_ie_ie_i^\top$. Then $\Sigma_{SS}+\Theta_\ell(t)$ has the same eigenvectors, with eigenvalue $d_\ell+t$ in direction $\ell$ and $d_i$ otherwise, and $v_{\Theta_\ell(t)}=z+\frac t\lambda(e_\ell^\top\beta_S)e_\ell$. Hence, with $A_k=\sum_i(a_ke_i)(e_i^\top z)/d_i$,
\[
c_k^{(\ell)}(t)=A_k-\frac{(a_ke_\ell)(e_\ell^\top z)}{d_\ell}+\frac{(a_ke_\ell)\big[e_\ell^\top z+\frac t\lambda e_\ell^\top\beta_S\big]}{d_\ell+t}.
\]
Multiplying by $d_\ell+t$, the terms $(a_ke_\ell)(e_\ell^\top z)$ that are not multiplied by $t$ cancel, and the numerator is $A_kd_\ell+L_k^{(\ell)}t$. With $t'=t/d_\ell$ this is $c=(A_k+L_k^{(\ell)}t')/(1+t')$, the form of Proposition~\ref{prop:repair}(b), and the remaining claims follow as there.}
\end{proof}
Proposition~\ref{prop:directional} shows that the exact repair analysis does not break down for correlated relevant variables. Along any eigendirection of $\Sigma_{SS}$, the elastic-net statistic of a violated variable is again an exact, monotone, closed-form function of a single scalar, as in the uncorrelated case. The same pairwise-intersection criterion then decides whether a single noise level repairs all of them, although for a given direction it need not. Because $s=|S|$ is typically small in sparse problems, one can in principle search over the $s$ eigendirections of $\Sigma_{SS}$ for one that repairs the design. The result comes with an important caveat: it requires noise aligned with the (unknown) eigenbasis of $\Sigma_{SS}$, restricted to the (unknown) active set $S$, whereas 2D-SS as implemented (Algorithm~\ref{alg:delta_averaged_selection}) injects simple entrywise isotropic noise on all $p$ columns. The proposition therefore explains \emph{why} correlated designs are not an intrinsic obstruction. It also suggests a direction-aware variant of jittering, built from a preliminary covariance estimate in the spirit of a whitening transform, that could extend the calibration guarantee of Theorem~\ref{thm:calibration} to correlated designs, which we do not pursue here.

\begin{example}[Isotropic jittering can actively hurt a correlated design]\label{ex:corr}
Take $s=2$ correlated relevant variables with $\Sigma_{SS}=\begin{pmatrix}1&0.5\\0.5&1\end{pmatrix}$, one irrelevant variable with $a=(0.8,1.2)^\top$, $z=(1,1)$, $\beta_S=(2.5,1.5)^\top$ and $\lambda=0.1$, {giving $A=a\Sigma_{SS}^{-1}z^\top=4/3$: the Lasso IC fails (the irrelevant variable's variance must satisfy $\Sigma_{kk}\ge a\Sigma_{SS}^{-1}a^\top\approx1.49$, so it is not standardised).} Because the relevant variables are correlated, Remark~\ref{r:nonmono} allows non-monotonicity. Here isotropic jittering does not merely fail to repair this design: it makes it strictly worse. The statistic $c_\nu^k$ increases monotonically from $4/3$ at $\nu=0$ to $(a\beta_S)/\lambda=38$ as $\nu\to\infty$, so no isotropic noise level brings $|c_\nu^k|$ below $1$. Proposition~\ref{prop:directional} nonetheless repairs it exactly. $\Sigma_{SS}$ has eigenvectors $e_1=(-1,1)/\sqrt2$ (eigenvalue $d_1=1/2$) and $e_2=(1,1)/\sqrt2$ (eigenvalue $d_2=3/2$); along $e_1$, $L_k^{(1)}=-2/3$, so $AL_k^{(1)}<0$ and the unique zero of Proposition~\ref{prop:directional} is at $t=-Ad_1/L_k^{(1)}=1$, giving $c_k^{(1)}(1)=0$ exactly.
\end{example}

{
\begin{corollary}[Compound-symmetric calibration]\label{cor:compound}
Suppose $s\ge2$ and $\Sigma_{SS}=(1-\rho)I_s+\rho J_s$ for some $\rho\in(-1/(s-1),1)$, where
$J_s=\mathbf 1_s\mathbf 1_s^\top$ (equicorrelated relevant variables). Let
$e_1\coloneq s^{-1/2}\mathbf 1_s$. Then $e_1$ is an eigenvector of $\Sigma_{SS}$ with
eigenvalue $d_1=1+(s-1)\rho$, for every $\rho$ in this range. Injecting rank-one
noise $t\,e_1e_1^\top$ on the $S$-block as in Proposition~\ref{prop:directional}
(with $\ell=1$), Theorem~\ref{thm:calibration}'s pairwise-intersection criterion
{applies after rescaling $t'=t/d_1$, with $B_k/\lambda$ replaced by $L_k^{(1)}$ (and $s\ge2$): a single $t\ge0$
brings every $k\in S^c$ within the margin $\eta$ simultaneously if and only if the family
$\mathcal T_k^{(1)}(\eta)\coloneq\{t\ge0:|c_k^{(1)}(t)|\le1-\eta\}$, $k\in S^c$, is pairwise
intersecting (the case $k=j$ encodes non-emptiness).} Unlike the general correlated case of
Remark~\ref{r:nonmono}, no eigen-decomposition of $\Sigma_{SS}$ is needed: $e_1$ is fixed by $s$ alone and $\rho$
entering only through $d_1$. The direction is, however, supported on the unknown set $S$, so this is a statement about the analysis and not a procedure.
\end{corollary}

\begin{proof}
$J_se_1=\mathbf 1_s(\mathbf 1_s^\top e_1)=\mathbf 1_s(s^{-1/2}\cdot s)=s^{1/2}\mathbf
1_s=s\,e_1$, so $\Sigma_{SS}e_1=(1-\rho)e_1+\rho s\,e_1=[1+(s-1)\rho]e_1=d_1e_1$,
confirming $e_1$ is an eigenvector with eigenvalue $d_1$ for every $\rho$ in the
stated range; the direction $e_1$ itself does not depend on $\rho$. Proposition~\ref{prop:directional}, applied with this $(e_1,d_1)$, then gives the stated
formula for $c_k^{(1)}(t)$ for every $k\in S^c$, and Lemma~\ref{lem:interval} and
Theorem~\ref{thm:calibration} apply to it exactly as in the proof of
Proposition~\ref{prop:directional}.
\end{proof}

\begin{remark}
Example~\ref{ex:corr} is itself an instance of Corollary~\ref{cor:compound}: its
$\Sigma_{SS}=\begin{pmatrix}1&0.5\\0.5&1\end{pmatrix}$ {is compound-symmetric with
$s=2$, $\rho=0.5$ (but the all-ones direction of the corollary does not repair it, since the statistic does not decrease along it; the repair uses the orthogonal direction $(-1,1)/\sqrt2$ through Proposition~\ref{prop:directional}), and $e_1=(1,1)/\sqrt2$ is exactly its eigenvector $e_2$ there}
(eigenvalue $d_2=1.5=1+(s-1)\rho$); the repair direction used in that example, $(-1,1)/\sqrt2$, is the orthogonal complement of $e_1$ here, available whenever
$s\ge2$ as the remaining eigenspace of $\Sigma_{SS}$, with common eigenvalue
$1-\rho$. For $s\ge2$, compound symmetry provides a closed-form
direction (the all-ones direction) and an $(s-1)$-dimensional orthogonal complement in which any unit vector may be used; which of these
satisfies the repair condition $A_kL_k^{(\ell)}<0$ for a given violated variable $k$
depends on the specific design and must be checked against both.
\end{remark}
}

\begin{example}\label{ex:jy}
Take $p=6$, $s=5$, $\Sigma_{SS}=I_5$, $a=(\tfrac18,\tfrac14,\tfrac12,\tfrac12,\tfrac12)$ (one irrelevant variable with unit variance), $z=(-1,-1,1,1,1)$ and $\beta_S=(-4,-2,0.5,0.6,0.7)^\top$, which is the first design of Example~3 in \citet{jia2010model}, for which $A=1.125$ and $B=-0.1$. The Lasso IC fails ($A>1$). With $\lambda=0.05$ we get $B/\lambda=-2$ and $c_\nu=(1.125-2\nu^2)/(1+\nu^2)$, so that $|c_\nu|\le0.9$ if and only if $\nu^2\in[9/116,\,81/44]$, that is, $\nu\in[0.279,\,1.356]$ {(endpoints rounded inward, so every $\nu$ in the interval shown satisfies the margin)}; for instance $c_{0.5}=0.5$. Moreover $\min_j|\beta_j|=0.5>\lambda$, so the ridge-shrunken target keeps the signs of $\beta_S$ ($\theta_\nu=(0.5-\lambda)/(1+\nu^2)>0$ for all $\nu$).
\end{example}

The following lemma is the deterministic step. Only the mean-zero fluctuation has to be small; the shift $\nu^2I_p$ is absorbed exactly into the reference matrix $\Sigma_\nu$ and does not count as perturbation.
\begin{lemma}[Perturbation of the elastic-net irrepresentable condition]\label{l1}
Fix $\nu\ge0$ and $\eta\in(0,1)$, assume $\Sigma_{SS}$ is invertible and EIC$_\nu$ holds with margin $\eta$. Let $\Sigma'=\Sigma_\nu+\Delta$ with $\max\{\|\Delta_{SS}\|_\infty,\|\Delta_{S^cS}\|_\infty\}\le\omega$. Put $\kappa_\nu\coloneq\|(\Sigma_{SS}+\nu^2I_s)^{-1}\|_\infty$ and
\[
\begin{aligned}
\omega_0(\nu)\coloneq\min\Bigg\{&\frac1{2\kappa_\nu},\\
&\frac{\eta}{4\kappa_\nu\|v_\nu\|_\infty\big(1+\kappa_\nu\|\Sigma_{S^cS}\|_\infty\big)}\Bigg\}.
\end{aligned}
\]
If $\omega\le\omega_0(\nu)$, then $\Sigma'_{SS}$ is invertible and
\[
\big\|\Sigma'_{S^cS}\big(\Sigma'_{SS}\big)^{-1}v_\nu\big\|_\infty\le1-\frac\eta2 .
\]
Moreover, $\omega_0(\nu)\ge\omega_0^*>0$ for all $\nu\ge0$, where
\[
\begin{gathered}
\omega_0^*\coloneq\min\Big\{\frac1{2\bar\kappa},\ \frac{\eta}{4V(1+\bar\kappa\|\Sigma_{S^cS}\|_\infty)}\Big\},\\
\bar\kappa\coloneq\frac{\sqrt s}{\lambda_{\min}(\Sigma_{SS})}\ \ (\ge\kappa_\nu\ \text{for all }\nu\ge0),\qquad
V\coloneq\sqrt s\Big(\frac1{\lambda_{\min}(\Sigma_{SS})}+\frac{\|\beta_S\|_\infty}\lambda\Big).
\end{gathered}
\]
\end{lemma}
Proof is shown in the appendix. For $\nu=0$ the lemma reduces to the classical statement that the IC is preserved under small perturbations of the Gram matrix. The last claim shows that the tolerance $\omega_0(\nu)$ is bounded below by a constant $\omega_0^*$ that does not depend on $\nu$: the ridge shift makes $\Sigma_{SS}+\nu^2I$ better conditioned ($\kappa_\nu$ stays below $\bar\kappa$), which offsets the growth of $\|v_\nu\|_\infty$. This concerns only the Gram-matrix perturbation, and whether EIC$_\nu$ itself holds at a given $\nu$, and whether the shrunken coefficients keep their signs (Theorem~\ref{thm:single_level}{, stated below}), must be checked separately.
\begin{remark}\label{r1}
Beyond irrepresentability, Lasso support recovery also requires a \emph{beta-min} condition,
\[
\min_{j\in S}|\beta_j|\ge c_0,
\]
for some $c_0>0$. Under perturbations, this introduces an additional restriction on the admissible noise level. In particular, correct support recovery requires the shrunken target to dominate the estimation error, which is made precise by condition (iv) of Theorem~\ref{thm:single_level}. Under jittering the ridge shrinkage acts on the same coefficients: the population target on $S$ is $\beta^\star_\nu\coloneq(\Sigma_{SS}+\nu^2I)^{-1}(\Sigma_{SS}\beta_S-\lambda z)$, so the beta-min condition has to be understood for this shrunken target, whose entries must keep the signs of $\beta_S$ (see $\theta_\nu$ in Theorem~\ref{thm:single_level}{, stated below}).
\end{remark}
\subsection{Gaussian Perturbations and High-Probability Robustness}\label{s2:hp}
We now extend Lemma~\ref{l1} to a probabilistic setting. The fluctuation $\widetilde\Delta$ in \eqref{eq:gram} has a subsampling part and a Gaussian part; Lemma~\ref{l1} only requires bounds on the two blocks $\widetilde\Delta_{SS}$ and $\widetilde\Delta_{S^cS}$, each of which has only $s$ columns.
\begin{theorem}[Gaussian perturbations and subsampling preserve the elastic-net irrepresentable condition with high probability] \label{thr:hp_gauss_final}
Fix $\delta\ge0$, $\nu=\nu_\delta$, and suppose $\Sigma_{SS}$ is invertible and EIC$_\nu$ holds with margin $\eta\in(0,1)$. Let $I$ be a uniformly random half-sample, independent of $(W_0,W)$, so that conditionally on $I$ the perturbation $\widetilde W=W_{0,I}+W$ has i.i.d.\ $\mathcal N(0,\nu^2)$ entries, and let $\Sigma^{(\delta)}$ be as in \eqref{eq:gram}. Let $\alpha,\alpha_I\in(0,1)$, $L\coloneq\log(8ps/\alpha)$ and
\[
\begin{gathered}
t_1\coloneq2K\sqrt{\frac{2\nu^2L}{n_1}},\qquad t_2\coloneq2\nu^2\Big(\sqrt{\frac{L+1}{n_1}}+\frac{L+1}{n_1}\Big),\\
u_\Sigma\coloneq K^2\sqrt{\frac{2\log(2ps/\alpha_I)}{n_1}},\qquad
\omega\coloneq s(t_1+t_2+u_\Sigma).
\end{gathered}
\]
If $\omega\le\omega_0(\nu)$ (in particular, if $\omega\le\omega_0^*$ from Lemma~\ref{l1}), then
\[
\mathbb P\!\left(\big\|\Sigma^{(\delta)}_{S^cS}\big(\Sigma^{(\delta)}_{SS}\big)^{-1}v_\nu\big\|_\infty\le1-\frac\eta2\right)\ge1-\alpha-\alpha_I .
\]
\end{theorem}
Proof is provided in the appendix.

We now turn to larger perturbations. Figure~\ref{fig:ic_violated} suggests that, for a fixed $\lambda$, the selection frequencies of the relevant variables decrease for large noise levels, and Remark~\ref{r2} indicates that for large noise levels the errors-in-variables effect can overwhelm the signal. Aggregating over resamples (bagging) at each noise level stabilises the output with respect to perturbations of the training data \citep{JMLR:v25:23-0536}, and Stability Selection aggregates selection outcomes in this way \citep{meinshausen2010stability}. We add a second layer of aggregation over noise levels, which promotes robustness to measurement error by selecting variables under a more stringent stress test. Theorem~\ref{thr:hp_gauss_final} only covers noise levels for which the perturbation of the Gram matrix is small; how far this extends to the whole grid is discussed in Proposition~\ref{prop:bridge}{ below}.
\subsection{Exact Recovery by a Single Jittered Fit}\label{s2:single}
{This subsection collects the guarantees for one jittered fit at a fixed noise level $\delta$ (with $\nu=\nu_\delta$, the effective level of Section~\ref{s2}): exact recovery of the support, a converse, and sharper statements under Gaussian design. Section~\ref{s2:delta_avg} then uses them to certify the aggregated frequencies.}
For a single noise level, let $\nu=\nu_\delta$, and let $\kappa_\nu$, $\omega_0(\nu)$, $L$, $t_1$, $t_2$, $u_\Sigma$, $\omega$ be as in Lemma~\ref{l1} and Theorem~\ref{thr:hp_gauss_final}, and let $L_p\coloneq\log(4p/\alpha)$. Define the shrunken population target and its sign margin, the response-side remainder bound, and the auxiliary quantities
\[
\begin{gathered}
\beta^\star_\nu\coloneq(\Sigma_{SS}+\nu^2I_s)^{-1}(\Sigma_{SS}\beta_S-\lambda z),\qquad
\theta_\nu\coloneq\min_{j\in S}z_j\beta^\star_{\nu,j},\\
\rho_\nu\coloneq\sigma K\sqrt{\frac{2L_p}{n_1}}+2\sigma\nu\sqrt{\frac{L_p}{n_1}}+\|\beta_S\|_1\Big(\frac{t_1}2+t_2\Big),\\
\Gamma_\nu\coloneq1+2\kappa_\nu\big(\|\Sigma_{S^cS}\|_\infty+\omega\big),\qquad
g_\nu\coloneq\|\Sigma_{SS}\beta_S-\lambda z\|_\infty,\\
\Psi_\nu\coloneq2\kappa_\nu\big(\kappa_\nu g_\nu\,\omega+\omega\|\beta_S\|_\infty+\rho_\nu\big).
\end{gathered}
\]
\begin{theorem}[Exact recovery by a single jittered fit]\label{thm:single_level}
Fix $\delta\ge0$ and $\nu=\nu_\delta$, and assume the model of Section~\ref{s2} ($Y=X_0\beta+\varepsilon$ with $\varepsilon\sim\mathcal N(0,\sigma^2I_n)$ independent of $(I,W_0,W)$, $|X_{0,ij}|\le K$, $\lambda>0$) and the setting and notation of Theorem~\ref{thr:hp_gauss_final}. Suppose that
\begin{enumerate}
\item[(i)] EIC$_\nu$ holds with margin $\eta\in(0,1)$;
\item[(ii)] $\omega\le\omega_0(\nu)$;
\item[(iii)] $\rho_\nu\,\Gamma_\nu<\eta\lambda/2$;
\item[(iv)] $\theta_\nu>\Psi_\nu$ (in particular $\theta_\nu>0$).
\end{enumerate}
Then, with probability at least $1-\bar\alpha$ over $(I,\varepsilon,W_0,W)$, where $\bar\alpha\coloneq\alpha_I+2\alpha+e^{-n_1/8}$, the Lasso fitted to $(X_I+W,Y_I)$ has a unique solution, its support equals $S$, and $\operatorname{sign}(\hat\beta_S)=\operatorname{sign}(\beta_S)$.
\end{theorem}
Proof is provided in the appendix. Conditions (i)--(ii) are the elastic-net irrepresentability and Gram-fluctuation conditions of Theorem~\ref{thr:hp_gauss_final}; (iii) says that the response-side remainder is small relative to the dual margin; and (iv) is a beta-min condition on the ridge-shrunken target $\beta^\star_\nu$. At $\delta=0$ and $\delta_{\mathrm{obs}}=0$ the theorem reduces to the standard primal--dual witness result for the Lasso on a half-sample. Conditions (i)--(iv) involve only the latent quantities $(\Sigma,\beta,\lambda,\sigma,\delta_{\mathrm{obs}})$, the noise level $\delta$, the design bound $K$, the dimensions $(p,s)$, the sample size, and the tolerances $(\alpha,\alpha_I)$, not the observed data, and are used in the analysis only.

Theorem~\ref{thm:single_level} has a complementary failure result under its own margin conditions: if the dual condition is violated by a margin at a noise level, then a single fit at that level cannot recover $S$ with the correct signs.
\begin{lemma}[Failure of a fit under a violated dual condition]\label{lem:failure}
Fix $\delta\ge0$, $\nu=\nu_\delta$ and $\eta_0\in(0,1)$, and assume the model of Section~\ref{s2} ($Y=X_0\beta+\varepsilon$ with $\varepsilon\sim\mathcal N(0,\sigma^2I_n)$ independent of $(I,W_0,W)$, $|X_{0,ij}|\le K$, $\lambda>0$) and the setting and notation of Theorem~\ref{thr:hp_gauss_final} (not its hypothesis EIC$_\nu$). Suppose $c_\nu\ge1+\eta_0$,
\[
\omega\le\min\Big\{\frac1{2\kappa_\nu},\ \frac{\eta_0}{4\kappa_\nu\|v_\nu\|_\infty\big(1+\kappa_\nu\|\Sigma_{S^cS}\|_\infty\big)}\Big\},\qquad \rho_\nu\Gamma_\nu<\frac{\eta_0\lambda}{4}.
\]
Then, with probability at least $1-\bar\alpha$ over $(I,\varepsilon,W_0,W)$, the Lasso fitted to $(X_I+W,Y_I)$ has no solution with support $S$ and $\operatorname{sign}(\hat\beta_S)=\operatorname{sign}(\beta_S)$.
\end{lemma}
\begin{corollary}[Jittering can succeed where the Lasso fails]\label{cor:benefit}
Let $\nu_0=\delta_{\mathrm{obs}}$ be the effective level at $\delta=0$ and let $\delta_1>0$ with $\nu_1=\nu_{\delta_1}$. Suppose that the conditions of Lemma~\ref{lem:failure} hold at $\nu_0$ for some $\eta_0\in(0,1)$, and that conditions (i)--(iv) of Theorem~\ref{thm:single_level} hold at $\nu_1$. Then, with probability at least $1-\bar\alpha$ each (jointly at least $1-2\bar\alpha$), a Lasso fit on a half-sample without jitter ($\delta=0$; for $\delta_{\mathrm{obs}}=0$ this is the plain Lasso) fails to recover $S$ with the correct signs, whereas a jittered fit at level $\delta_1$ recovers $S$ with the correct signs.
\end{corollary}
Proofs are shown in the appendix. \begin{remark}[A parallel ridge-indexed family in screening]\label{r:holp}
Proposition~\ref{prop:repair} is an instance of a broader pattern: indexing an estimator by a ridge-type parameter can trade a condition suited to small perturbations for one suited to large ones, {with different conditions at the two ends.} \citet{10.1111/rssb.12127} {(with the corrected conditions of \citet{10.1111/rssb.12427})} exhibit the same pattern for variable screening, where $X^\top(XX^\top+rI_n)^{-1}Y$ moves from a projection-based estimator at $r\to0$ (no marginal-correlation requirement) to marginal screening $X^\top Y$ at $r\to\infty$ (marginal-correlation condition of \citet{10.1111/j.1467-9868.2008.00674.x}). {Here the direction is the same as for HOLP (the small-parameter end has no marginal-correlation requirement and the large end is marginal screening); what differs is the example, since the small-parameter end ($\nu=0$) fails in Example~\ref{ex:jy} whereas the large end ($r\to\infty$) fails in their designs: EIC$_\nu$ moves} from the Lasso irrepresentable condition at $\nu=0$, which does not directly involve marginal correlations, to a marginal-correlation condition as $\nu\to\infty$ (Section~\ref{s2:small}). Example~\ref{ex:jy} shows a design in which the Lasso IC ($\nu=0$) fails while EIC$_\nu$ holds at intermediate $\nu$, the counterpart of their designs in which marginal screening fails ($r\to\infty$) while $r\to0$ succeeds.
\end{remark}
Example~\ref{ex:jy} is a design in which the hypotheses on $c_\nu$ hold at $\nu_0=0$ ($c_0=1.125\ge1+\eta_0$ with $\eta_0=0.125$) and at every $\nu_1\in[0.279,1.356]$ (margin $\eta=0.1$); the remaining conditions are sample-size and tuning conditions of the type discussed in Remark~\ref{r:range}. In terms of the aggregation below, the corollary says that $\delta=0$ need not belong to the set $\mathcal G_{\mathrm{good}}$ of good noise levels (defined in Section~\ref{s2:delta_avg}), while levels such as $\delta_1$ can. Recovery is then certified by the jittered levels alone, and a grid that consists of $\delta=0$ only (Stability Selection) certifies nothing. The sample-size and tuning constants in these conditions are conservative. In Example~\ref{ex:jy} they require far larger samples than those simulated in Section~\ref{s3}, so the probability $\pi$ of the good levels certified by conditions (i)--(iv) is effectively zero at the simulated $n$. Proposition~\ref{prop:bridge} should therefore be read as a qualitative explanation of the mechanism, not as a quantitative guarantee for those sample sizes.

Under a Gaussian design the sample-size condition can be sharpened to order $s\log p$ (Theorem~\ref{thm:gauss_rate}), the recovery probability of one fit has an exact finite-sample representation (Theorem~\ref{thm:exact}) with a zero--one law as $n_1\to\infty$ (Corollary~\ref{cor:zeroone}), and the guarantee can be compared with that of CoCoLasso (Remark~\ref{r:cocolasso_compare}). These results, and remarks on the tuning parameter, are in the Appendix.

\subsection{$\delta$-Averaged Selection Frequencies and Support Recovery}\label{s2:delta_avg}
Fix a tuning parameter $\lambda$ that does not depend on $\delta$, and let $\hat\beta^{(\delta)}=\hat\beta^{(\delta)}(\lambda)$ denote the coefficient vector fitted by the chosen selector to $(X_I+W,Y_I)$ for a random half-sample $I$ and noise $W$ with i.i.d.\ $\mathcal N(0,\delta^2)$ entries. (The default is a common $\lambda$; the results below use $\lambda$ only through the conditions at the level under consideration and hold with level-specific $\lambda_\delta$; see Remark~\ref{r:tuning}.) Let $\mathcal D\coloneq(\varepsilon,W_0)$ denote the data-generating randomness, so that the observed data $(X,Y)$ are a function of $\mathcal D$. Conditional on the observed data, define for each variable $j$ and noise level $\delta$ the population selection frequency
\[
f_j^{(\delta)}\coloneq\mathbb P\big(\hat\beta^{(\delta)}_j\neq0\ \big|\ \mathcal D\big),
\]
where the probability is over the half-sample $I$ and the noise $W$ (and any algorithmic randomness); $f_j^{(\delta)}$ is a random variable through $\mathcal D$. It is approximated by the empirical frequency $\hat f_j^{(\delta)}\coloneq\frac1B\sum_{b=1}^B1\{\hat\beta_j^{(b,\delta)}\neq0\}$ over the $B$ independent draws $(I_b,W^{(b)})$ of the procedure. Figure~\ref{fig:ic_violated} and Remark~\ref{r2} suggest that, as $\delta$ increases, the selection frequencies of the relevant variables and of the confounders tend to decrease, with a potentially informative intermediate range in which relevant and irrelevant variables are still separable. To exploit robustness across noise magnitudes, we aggregate over a finite grid $\mathcal G=\{\delta_1,\dots,\delta_m\}\subset[0,\infty)$ with weights $w_\delta>0$, $\sum_{\delta\in\mathcal G}w_\delta=1$ (uniform weights $w_\delta=1/m$ by default):
\[
\bar f_j\coloneq\sum_{\delta\in\mathcal G}w_\delta f_j^{(\delta)},\qquad
\widehat{\bar f}_j\coloneq\sum_{\delta\in\mathcal G}w_\delta\hat f_j^{(\delta)}.
\]
This rewards variables that are selected consistently across perturbation magnitudes and reduces sensitivity to the choice of any single $\delta$. With $\delta=0\in\mathcal G$, the $\delta=0$ term of the average is the Stability Selection frequency.
\begin{assumption}[Separation of $\delta$-averaged selection frequencies]\label{as:avg_sep}
There exists a constant $\gamma > 0$ such that $\min_{j \in S} \bar f_j
\;\ge\;
\max_{k \in S^c} \bar f_k + \gamma$.
\end{assumption}
Assumption~\ref{as:avg_sep} is a frequency-domain analogue of classical separability conditions in model selection: every active variable is selected more often, on average across the grid, than every inactive variable, with a gap $\gamma$ acting as a clean score threshold. In applications, $\mathcal G$ is chosen to cover a range of perturbation magnitudes that is large enough to stress-test robustness yet not so large that all frequencies collapse towards zero (see Remark~\ref{r2}).

\noindent\textbf{When does separation hold?} Assumption~\ref{as:avg_sep} is a statement about population frequencies, so on its own a recovery guarantee based on it would be circular. The rest of this section derives it from conditions on the latent model and the noise level. Section~\ref{s2:single} gave conditions under which a single jittered fit recovers $S$ exactly (Theorems~\ref{thm:single_level}, \ref{thm:gauss_rate} and~\ref{thm:exact}). Proposition~\ref{prop:bridge} shows that if these conditions hold on a sufficiently large part of the grid, Assumption~\ref{as:avg_sep} holds with an explicit margin $\gamma$. Theorem~\ref{thm:delta_avg_recovery} controls the Monte Carlo error, and Theorem~\ref{thm:gap} gives a condition under which the largest-gap rule estimates $s$ correctly.

Let $\hat S^{(\delta)}$ denote the support selected by one jittered fit. Under the Gaussian design of Theorem~\ref{thm:gauss_rate}, conditions (i)--(iv) of that theorem may be used in place of those of Theorem~\ref{thm:single_level}, with $\bar\alpha$ replaced by $\alpha$ and $\mathcal D$ by $(X_0,\varepsilon,W_0)$ in what follows. Alternatively, under Gaussian design one may define $\mathcal G_{\mathrm{good}}$ as the set of levels with $1-r_\delta\le\bar\alpha$, where $r_\delta$ is the probability of Theorem~\ref{thm:exact}; since exact recovery implies $\hat S^{(\delta)}=S$, the bound $\mathbb P(\hat S^{(\delta)}\neq S)\le\bar\alpha$ below then holds without any explicit constants. Define the set of \emph{good} noise levels and the share of the weight they carry,
\[
\begin{gathered}
\mathcal G_{\mathrm{good}}\coloneq\big\{\delta\in\mathcal G:\ \text{conditions (i)--(iv) hold at }\nu_\delta\big\},\\
\pi\coloneq\sum_{\delta\in\mathcal G_{\mathrm{good}}}w_\delta ,
\end{gathered}
\]
so that $\mathbb P(\hat S^{(\delta)}\neq S)\le\bar\alpha$ for every $\delta\in\mathcal G_{\mathrm{good}}$ by Theorem~\ref{thm:single_level}, where the probability is over $(I,\mathcal D,W)$. Finally let
\[
\bar\varphi\coloneq\max_{k\in S^c}\sum_{\delta\in\mathcal G\setminus\mathcal G_{\mathrm{good}}}w_\delta f_k^{(\delta)}\ \ (\le 1-\pi)
\]
be the contribution of the remaining noise levels to the averaged frequency of the worst irrelevant variable; like $f_k^{(\delta)}$, it is a random variable through $\mathcal D$.
\begin{proposition}[From local guarantees to $\delta$-averaged separation]\label{prop:bridge}
Let $\kappa\in(0,1)$ with $\bar\alpha<\kappa/2$. With probability at least $1-\kappa$ over the data $\mathcal D$ (that is, $(\varepsilon,W_0)$ for a fixed design and $(X_0,\varepsilon,W_0)$ under the Gaussian design, as explained above), Assumption~\ref{as:avg_sep} holds with margin
\[
\gamma=\pi\Big(1-\frac{2\bar\alpha}{\kappa}\Big)-\bar\varphi
\]
whenever this quantity is positive; here $\gamma$ is $\mathcal D$-measurable through $\bar\varphi$, so ``whenever positive'' is itself an event about $\mathcal D$, which the deterministic margin below avoids. In particular, since $\bar\varphi\le1-\pi$, it holds with $\gamma=2\pi(1-\bar\alpha/\kappa)-1$ as soon as $\pi>1/(2(1-\bar\alpha/\kappa))$, that is, as soon as the weight on good noise levels exceeds $1/(2(1-\bar\alpha/\kappa))$ (a threshold above $1/2$; it is below $1$ exactly when $\bar\alpha<\kappa/2$).
\end{proposition}
\begin{proof}
Let $q_\delta\coloneq\mathbb P(\hat S^{(\delta)}\neq S\mid\mathcal D)$ (a probability, not the vector $q$ of the Gaussian-design subsection of the Appendix; likewise $\kappa$ is the probability level of the proposition), so that $\mathbb E(q_\delta)\le\bar\alpha$ for $\delta\in\mathcal G_{\mathrm{good}}$, and put $Q\coloneq\sum_{\delta\in\mathcal G_{\mathrm{good}}}w_\delta q_\delta$. Then $\mathbb E(Q)\le\pi\bar\alpha$, and by Markov's inequality $Q\le\pi\bar\alpha/\kappa$ with probability at least $1-\kappa$. On this event: for $\delta\in\mathcal G_{\mathrm{good}}$ and $j\in S$, $\{\hat S^{(\delta)}=S\}\subset\{\hat\beta^{(\delta)}_j\neq0\}$, so $f_j^{(\delta)}\ge1-q_\delta$; for $k\in S^c$, $\{\hat\beta^{(\delta)}_k\neq0\}\subset\{\hat S^{(\delta)}\neq S\}$, so $f_k^{(\delta)}\le q_\delta$. For $\delta\notin\mathcal G_{\mathrm{good}}$ we use $f_j^{(\delta)}\ge0$ for $j\in S$. Averaging with the weights gives $\bar f_j\ge\pi-Q$ for $j\in S$ and $\bar f_k\le Q+\bar\varphi$ for $k\in S^c$. Hence $\min_{j\in S}\bar f_j-\max_{k\in S^c}\bar f_k\ge\pi-2Q-\bar\varphi\ge\pi(1-2\bar\alpha/\kappa)-\bar\varphi$. The second claim uses $\bar\varphi\le1-\pi$.
\end{proof}
Proposition~\ref{prop:bridge} shows what aggregation can and cannot do at large noise levels. Noise levels outside the local regime cannot be certified for the relevant variables, but they harm the separation only through $\bar\varphi$, the weighted frequency of irrelevant variables at those levels. Our theory does not show that $\bar\varphi$ is small: in Figure~\ref{fig:ic_violated} the mean frequency of the confounders is at most about $0.04$ for $\delta\ge0.75$ (and the largest single frequency of an irrelevant variable at $\delta=0.75$ is about $0.05$), but this is an empirical observation from one data set and does not bound $\bar\varphi$, which is a maximum over irrelevant variables and also involves the levels $\delta<0.75$ that are not good. In simulations, where $S$ is known, $\pi$ can be computed from conditions (i)--(iv) and $\bar\varphi$ and $\gamma$ from the frequencies, and compared with the observed separation. Two design choices follow directly from the proposition. First, the grid and weights should put somewhat more than half of the total weight on good noise levels, that is, levels at which conditions (i)--(iv) hold; then no assumption on $\bar\varphi$ is needed. Since $\nu_\delta^2=\delta_{\mathrm{obs}}^2+\delta^2$, this requires $\delta_{\mathrm{obs}}$ to lie in the admissible range of Remark~\ref{r:range}; and $\delta=0$ is a good level only if the Lasso conditions hold at $\nu_0=\delta_{\mathrm{obs}}$, which fails, for instance, when the IC is violated (Corollary~\ref{cor:benefit}). Second, weights that decrease with $\delta$ reduce the weight of the levels that are not certified, at the cost of a weaker stress test. Corollary~\ref{cor:benefit} shows that this is not merely formal: in designs where the Lasso IC fails, $\delta=0$ is not a good level while jittered levels can be, so that $\pi$ is carried by the jittered levels alone.
{

}

The next theorem controls the Monte Carlo error: under Assumption~\ref{as:avg_sep}, thresholding the $\delta$-averaged empirical frequencies recovers $S$ with high probability.
\begin{theorem}[Support recovery from $\delta$-averaged empirical frequencies]\label{thm:delta_avg_recovery}
Condition on the data $\mathcal D$ and assume Assumption~\ref{as:avg_sep} holds with margin $\gamma>0$. Let $\alpha_{\mathrm{MC}}\in(0,1)$, let $m = |\mathcal G|$ be the number of perturbation levels, and let $B$ be the number of (subsample, noise) realisations generated at each perturbation level. Define
$\varepsilon_{\mathrm{MC}} \coloneq \sqrt{\frac{\log(2mp/\alpha_{\mathrm{MC}})}{2B}}$. Then, with probability at least $1-\alpha_{\mathrm{MC}}$ over the half-samples and perturbations used to compute
$\{\hat f_j^{(\delta)}\}_{\delta\in\mathcal G,\,1\le j\le p}$,
\[
\min_{j\in S}\widehat{\bar f}_j
\;\ge\;
\max_{k\in S^c}\widehat{\bar f}_k
\;+\;(\gamma-2\varepsilon_{\mathrm{MC}}).
\]
In particular, if $\varepsilon_{\mathrm{MC}}<\gamma/2$, then there exists a threshold $\tau$ such that
\[
\hat S(\tau)\coloneq\{j:\widehat{\bar f}_j\ge \tau\}=S.
\]
\end{theorem}
Proof is shown in the appendix. Theorem~\ref{thm:delta_avg_recovery} is a Monte Carlo statement: the empirical frequencies are close to the population ones. Its substance comes from Assumption~\ref{as:avg_sep}, which is justified by Theorem~\ref{thm:single_level} and Proposition~\ref{prop:bridge}. Combining the three, and writing $\bar\alpha_\kappa\coloneq\bar\alpha/\kappa$ (with $\bar\alpha<\kappa/2$), if $\pi>1/(2(1-\bar\alpha/\kappa))$, $\gamma=2\pi(1-\bar\alpha/\kappa)-1$ and $\varepsilon_{\mathrm{MC}}<\gamma/2$, {then some threshold recovers $S$ exactly (a pre-specified $\tau$ works if $\max_{k\in S^c}\bar f_k+\varepsilon_{\mathrm{MC}}<\tau\le\min_{j\in S}\bar f_j-\varepsilon_{\mathrm{MC}}$, which on the event of Proposition~\ref{prop:bridge} is guaranteed only for $\tau\in(\pi\bar\alpha_\kappa+\bar\varphi+\varepsilon_{\mathrm{MC}},\pi(1-\bar\alpha_\kappa)-\varepsilon_{\mathrm{MC}}]$ with $\bar\varphi$ unknown, and is compatible with $\tau>1/2$ only if $\pi(1-\bar\alpha_\kappa)>1/2+\varepsilon_{\mathrm{MC}}$), jointly} over the data and the Monte Carlo draws, with probability at least $1-\kappa-\alpha_{\mathrm{MC}}$. The Monte Carlo error decreases as $B$ increases.
\subsubsection*{Thresholding by Estimated Model Size}
Since the true model size $s=|S|$ is unknown in practice, we estimate $s$ by locating the largest drop in the sorted $\delta$-averaged empirical frequencies. {This is in the same spirit as automatic threshold-calibration methods for Stability Selection such as that of \citet{Huang17022026}, though their method (EATS) instead selects the threshold via a profile-log-likelihood-based elbow-detection procedure, considerably more elaborate than the simple largest-gap heuristic used here.} We define
\[
\hat s \coloneq \arg\max_{1\le k<p}\big(\widehat{\bar f}_{(k)}-\widehat{\bar f}_{(k+1)}\big),
\qquad
\hat\tau \coloneq \frac{1}{2}\big(\widehat{\bar f}_{(\hat s)}+\widehat{\bar f}_{(\hat s+1)}\big),
\qquad
\hat S \coloneq \{j:\widehat{\bar f}_j\ge \hat\tau\},
\]
that is, we cut at the largest adjacent gap (at the smallest such index if the maximum is attained more than once; if all gaps are zero we take $\hat S=\emptyset$ by convention, and the rule presupposes $1\le s<p$) and use the midpoint of the two neighbouring values as the threshold (Algorithm~\ref{alg:delta_averaged_selection}). Theorem~\ref{thm:delta_avg_recovery} produces a gap of size at least $\gamma-2\varepsilon_{\mathrm{MC}}$ between the two groups. It does not imply that this is the \emph{largest} gap. A large drop \emph{inside} $S$, for instance between a strong signal with frequency close to $1$ and weak signals with frequency around $0.5$, could exceed it and lead to an underestimate of $s$, so the rule needs an additional condition. Let $\bar f_{(1)}\ge\cdots\ge\bar f_{(p)}$ denote the sorted population frequencies; under Assumption~\ref{as:avg_sep} the $s$ largest are those of $S$. Let
\[
g_{\mathrm{in}}\coloneq\max_{k\neq s,\ 1\le k<p}\big(\bar f_{(k)}-\bar f_{(k+1)}\big)
\]
be the largest adjacent gap \emph{within} the two groups.
\begin{theorem}[Largest-gap rule]\label{thm:gap}
Let $s=|S|$ with $1\le s<p$ (and $g_{\mathrm{in}}\coloneq0$ if the maximum is over an empty set). Condition on the data $\mathcal D$, suppose Assumption~\ref{as:avg_sep} holds with margin $\gamma$ and
\[
\gamma>g_{\mathrm{in}}+4\varepsilon_{\mathrm{MC}} .
\]
Then, with $\varepsilon_{\mathrm{MC}}$ and $\alpha_{\mathrm{MC}}$ as in Theorem~\ref{thm:delta_avg_recovery}, on the event $\{\max_{\delta\in\mathcal G}\max_j|\hat f_j^{(\delta)}-f_j^{(\delta)}|\le\varepsilon_{\mathrm{MC}}\}$ used in the proof of that theorem, which has probability at least $1-\alpha_{\mathrm{MC}}$, $\hat s=s$ and $\hat S=S$.
\end{theorem}
Proof is shown in the appendix. The condition says that the boundary gap must dominate every within-group gap by a margin covering the Monte Carlo error. It is a genuine restriction and can fail, for example when $S$ contains both very strong and very weak variables. On the event of Proposition~\ref{prop:bridge}, the population frequencies in $S$ lie in $[\pi(1-\bar\alpha_\kappa),1]$ and those in $S^c$ in $[0,\pi\bar\alpha_\kappa+\bar\varphi]$, so $g_{\mathrm{in}}\le\max\{1-\pi(1-\bar\alpha_\kappa),\ \pi\bar\alpha_\kappa+\bar\varphi\}$, and the theorem applies (with probability at least $1-\kappa-\alpha_{\mathrm{MC}}$) whenever $\pi(2-3\bar\alpha_\kappa)>1+\bar\varphi+4\varepsilon_{\mathrm{MC}}$ and $\pi(1-3\bar\alpha_\kappa)>2\bar\varphi+4\varepsilon_{\mathrm{MC}}$.{ In the worst case $\bar\varphi=1-\pi$ these require $\pi>(2+4\varepsilon_{\mathrm{MC}})/(3(1-\bar\alpha_\kappa))>2/3$, so certifying the largest-gap rule needs more weight on good levels than the recovery of $S$ with a known threshold.} When it is not satisfied, the largest-gap rule can underestimate $s$, and {one should use a threshold $\tau$ that satisfies the interval above (Theorem~\ref{thm:delta_avg_recovery}) or a known model size instead.}
\begin{remark}\label{r:compute}
Our method requires $m\times B$ half-sample fits, compared with $B$ fits for Stability Selection, i.e., it increases computational cost by a factor of $|\mathcal{G}|$ (with $\delta=0\in\mathcal G$, the $\delta=0$ slice costs exactly as much as Stability Selection). For comparisons at matched compute, Stability Selection should be given the same total number $mB$ of half-sample fits. {Each fit uses half of the sample and the $mB$ fits are independent, so they can be run in parallel.}
\end{remark}
\subsection{Error Control by Complementary Pairs}\label{s2:errctrl}
The results above concern recovery of $S$; they do not bound the number of false positives of a thresholded frequency. We follow \citet{shah2013variable} and use complementary pairs. Our first bound below extends their theorem~1(a), for an arbitrary fixed set $L$, to the $\delta$-averaged setting. It specialises to their result when $L$ is a sub-level set $L_\theta=\{k:p_k\le\theta\}$. {Our statistic $\widehat{\bar f}^{\,\mathrm{cp}}_k$ is an average of the two half-samples' indicators: either half selecting $k$ counts, with weight one half each. It is not a product, and so differs from their simultaneous-selection frequency $\tilde\Pi_B(k)$ (their definition~2), which requires both halves to select $k$. The product-based quantity that matches $\tilde\Pi_B(k)$ at a single fixed $\delta$ is the auxiliary indicator $t_{b,\delta}$, used only inside the proof below.} At each level $\delta$ we draw $B$ pairs $(A_{2b-1},A_{2b})$ of disjoint half-samples that partition the data ($n$ even), and fit both halves with independent jitter noise. Write $\hat S(A;\delta)$ for the set selected by the (randomised) base selector on the half-sample $A$ at level $\delta$, and define
\[
\begin{gathered}
\widehat{\bar f}^{\,\mathrm{cp}}_k\coloneq\sum_{\delta\in\mathcal G}w_\delta\,\frac1{2B}\sum_{b=1}^B\Big(1\{k\in\hat S(A_{2b-1};\delta)\}+1\{k\in\hat S(A_{2b};\delta)\}\Big),\\
\hat S^{\mathrm{cp}}_\tau\coloneq\{k:\widehat{\bar f}^{\,\mathrm{cp}}_k\ge\tau\}.
\end{gathered}
\]
Each half is uniformly distributed, so $\widehat{\bar f}^{\,\mathrm{cp}}_k$ estimates the same population frequency $\bar f_k$ as before and the results above are unchanged (with $B$ now counting pairs and $2mB$ fits in total). Let $p_k(\delta)\coloneq\mathbb P(k\in\hat S(A;\delta))$, the probability being taken over the data, the half-sample and the jitter noise.
\begin{theorem}[Error control]\label{thm:errctrl}
Assume that the observations $(X_i,Y_i)$ are i.i.d.\ (for instance under the Gaussian design of Theorem~\ref{thm:gauss_rate}, with $W_0$ having i.i.d.\ rows) and that the jitter noise is drawn independently for every half-sample, {that the complementary pairs are drawn independently of the data and of the jitter, and that the base selector applied to a half-sample uses only that half-sample, with tuning parameters that are fixed or computed from that half-sample alone}. Then, for every $\tau\in(1/2,1]$ and every set $L\subset\{1,\dots,p\}$ of variables with $\theta_L\coloneq\max_{k\in L,\,\delta\in\mathcal G}p_k(\delta)$,
\[
\mathbb E\,\big|\hat S^{\mathrm{cp}}_\tau\cap L\big|\ \le\ \frac{1}{2\tau-1}\sum_{\delta\in\mathcal G}w_\delta\sum_{k\in L}p_k(\delta)^2\ \le\ \frac{\theta_L}{2\tau-1}\sum_{\delta\in\mathcal G}w_\delta\,\mathbb E\,\big|\hat S(A;\delta)\cap L\big| .
\]
In particular, for $L=S^c$ and $V_\tau\coloneq|\hat S^{\mathrm{cp}}_\tau\cap S^c|$, if the indicators $(1\{k\in\hat S(A;\delta)\})_{k\in S^c}$ are exchangeable for every $\delta$, then, with $\bar q_\delta\coloneq\mathbb E|\hat S(A;\delta)\cap S^c|$,
\[
\mathbb E\,V_\tau\ \le\ \frac{\sum_{\delta}w_\delta\,\bar q_\delta^{\,2}}{(p-s)(2\tau-1)}\ \le\ \frac{\max_\delta\bar q_\delta\ \sum_\delta w_\delta\bar q_\delta}{(p-s)(2\tau-1)} .
\]
\end{theorem}
Proof is shown in the appendix. {The first form of the bound needs no condition on the base selector beyond those in the theorem, or on the design beyond i.i.d.\ observations; the second form also needs exchangeability of the irrelevant variables. Both bound an expectation over datasets and not the error on a given dataset. The closed form fails without exchangeability (an irrelevant variable that is always selected is a counterexample), and confounders violate it, so only the first form is exchangeability-free. Neither form covers tuning or standardisation chosen on the full data (as in the rat analysis), or the largest-gap rule.} {Both forms hold for the weighted average over noise levels.} We note two consequences. First, {the sharper form of the bound is the weighted average over the noise levels of the single-level bounds $\bar q_\delta^{\,2}/((p-s)(2\tau-1))$, so it never exceeds the largest of them. When $0\in\mathcal G$ and the expected number of irrelevant selections does not increase with $\delta$, it is at most the bound of the $\delta=0$ slice, that is, of Stability Selection, and strictly smaller whenever some $\bar q_\delta<\bar q_0$. This monotonicity is not guaranteed: in Figure~\ref{fig:ic_violated} the confounders are selected less often as $\delta$ grows, but the mean frequency of the other irrelevant variables rises from $0.013$ to about $0.04$. The expected total number of irrelevant selections in that design is larger at $\delta=3$ ($\approx40$) than at $\delta=0$ ($\approx17$).} Second, consider the \emph{size-matched selector}, which takes $\hat S(A;\delta)$ to be the first $q_\star$ variables to enter the Lasso path computed on the jittered half-sample $(X_A+W,Y_A)$. Then $\bar q_\delta\le q_\star$ for every $\delta$, and hence, {under the exchangeability of the second form,} $\mathbb E\,V_\tau\le q_\star^2/((p-s)(2\tau-1))$ for any grid and any weights. This is a bound of the type of \citet{meinshausen2010stability}. Their bound has $p$ in place of $p-s$ in the denominator, which traces to their combined exchangeability-and-no-worse-than-random-guessing argument. Ours arises directly from exchangeability of the selection indicators over $S^c$ alone. The recovery results above are stated for a fixed penalty level and do not directly cover this selector; the price of the size-matched selector is thus the loss of a fixed $\lambda$, and the gain is the error bound together with the absence of the shrinkage confound discussed in Remark~\ref{r:tuning}. The bound applies to a fixed threshold $\tau>1/2$, not to the data-driven threshold of the largest-gap rule, and it uses independence of the observations, so it does not apply to a fixed design.

\subsection{Beyond the Lasso: other base selectors}\label{s2:other}
Algorithm~\ref{alg:delta_averaged_selection} applies to any base selector. The single-fit guarantee of Theorem~\ref{thm:single_level} carries over to the elastic net, to weighted $\ell_1$ penalties and, with weaker conclusions, to folded-concave penalties (Proposition~\ref{prop:other} in the Appendix).

\section{Numerical Results}\label{s3}
{
\subsection*{Study 1: Corollary~\ref{cor:benefit} in a Single Fit}
Study 1 examines the single-fit mechanism of Corollary~\ref{cor:benefit} on three designs, each with one confounder. In Example~\ref{ex:jy} the Lasso condition fails and isotropic noise repairs it. In Example~\ref{ex:corr} the condition fails and isotropic noise does not repair it, and in the control design the condition holds. We use $n_1\in\{25,\dots,5{,}000\}$ and $16$ noise levels with $1{,}000$ trials each.

In Example~\ref{ex:jy} the plain fit recovers $S$ exactly in $31$--$34\%$ of fits for $n_1\ge50$ and does not improve with $n_1$. The best jittered fit rises from $0.21$ ($n_1=25$) to $0.71$ ($n_1=5{,}000$), a gain of $0.39$. In Example~\ref{ex:corr} the best jittered rate falls with $n_1$, from $0.32$ to about $0$, while the plain fit stays at $34$--$45\%$. In the control, the plain fit already recovers $S$ in $79$--$84\%$ of fits for $n_1\ge100$, and jittering does not improve on it. These results follow the directions that the theory predicts (see Figure~\ref{fig:corollary1_curves} and the Appendix).

}

{
\subsection*{Study 2: Robustness of the Aggregated Procedure}\label{study:future_noise}
Study 1 examines the mechanism for a single fit. Study 2 asks whether it still holds in the aggregated procedure of Algorithm~\ref{alg:delta_averaged_selection}, by comparing Stability Selection and 2D-SS. We use Example~\ref{ex:jy} ($\beta_S=(-4,-2,0.5,0.6,0.7)^\top$, $\Sigma_{SS}=I_5$) with $K\in\{1,5\}$ confounders, embedded among $p=1{,}000$ variables. The confounders have unit variance and covariance vector $a=(1/8,1/4,1/2,1/2,1/2)^\top$ with the relevant variables, and the remaining variables are independent standard normal. The training data have $\delta_{\mathrm{obs}}=0$ and $n\in\{100,200,400\}$.

Stability Selection and 2D-SS use $\lambda=\sqrt{\log p/n_1}$. 2D-SS uses $B=50$ replicates at each of $m=5$ equally spaced noise levels in $[0,1]$, which is $mB=250$ fits. Stability Selection uses the same number of half-sample fits, $250$, at $\delta=0$, so the two have equal compute (Remark~\ref{r:compute}). Jitter-only uses the same noise grid and the full sample instead of half-samples. The Lasso, ridge regression and the elastic net (ENet, mixing parameter $0.5$) are single fits. These and Jitter-only use $\lambda=\sqrt{\log p/n}$ on the full sample.

\paragraph{Selection at training time.} Over $n_{\mathrm{rep}}=200$ training sets we record the fraction of confounders selected, the true and false positives, and whether $S$ is recovered exactly. We use the largest-gap rule for Stability Selection, Jitter-only and 2D-SS, and the nonzero coefficients for the single-fit methods. We also estimate the quantity $p_k(\delta)$ of Section~\ref{s2:errctrl} by complementary pairs ($N=200$ pairs, each from an independently drawn dataset), averaged over the five noise levels for 2D-SS.

\begin{figure}[H]
\centering
\includegraphics[width=\textwidth]{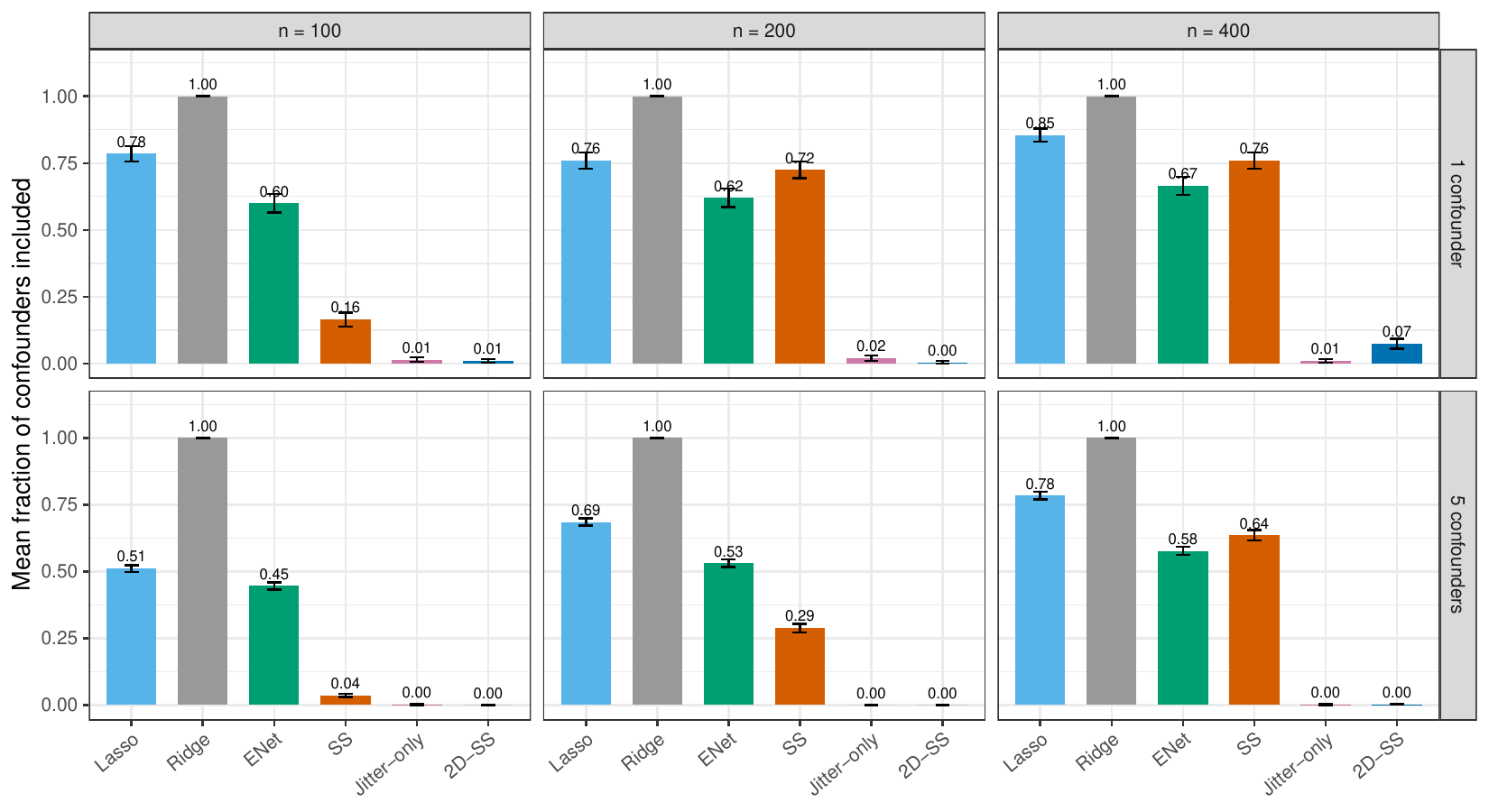}
\caption{Mean fraction of confounders selected at training time, all six methods, by
sample size $n$ and number of confounders $K$ ($n_{\mathrm{rep}}=200$; bars $\pm$ one
standard error).}
\label{fig:confounder_inclusion}
\end{figure}

Figure~\ref{fig:confounder_inclusion} shows a clear separation. It plots the rate per confounder, so the rates for $K=5$ of selecting at least one such variable, quoted below, are not plotted. The single-fit methods include the confounder in most training sets (Lasso $0.76$--$0.855$ and ENet $0.60$--$0.67$ for $K=1$). For $K=5$, at least $99.5\%$ of Lasso fits contain some confounder. Stability Selection reduces the rate at $n=100$ ($0.165$) but not at larger $n$. Its rate rises to $0.725$ and $0.760$ at $n=200,400$, because the selection frequency of the confounder grows with the signal. Jitter-only and 2D-SS stay low throughout ($0.010$--$0.020$ and $0.005$--$0.075$ for $K=1$). For $K=5$, at most $0.015$ of their training sets contain any such variable, against $0.775$--$0.98$ for Stability Selection at $n\ge200$.

The Appendix figure (Figure~\ref{fig:confounder_inclusion_cp}) gives the same ordering with the quantity used in the theory, the probability $p_k(\delta)$ that a single half-sample fit selects a confounder $k$. For $K=1$ it is $0.42$, $0.75$, $0.79$ for Stability Selection and $0.13$, $0.30$, $0.36$ for 2D-SS at $n=100,200,400$, a reduction by a factor of $2.2$--$3.3$. The pattern is similar for $K=5$ ($0.21$--$0.69$ versus $0.07$--$0.30$).

Appendix Figure~\ref{fig:selection_metrics} shows the cost. 2D-SS removes the false positives, which for Stability Selection are mostly the confounder. At $n=400$ it recovers $S$ exactly in $64\%$ of training sets for $K=1$, against $24\%$ for Stability Selection. At $n\le200$ it loses power on the weak coefficients in exchange for excluding the confounder. Jitter-only has a similar inclusion rate and more true positives at $n=400$, since it uses the full sample. It has no counterpart of Theorem~\ref{thm:errctrl}, whose complementary-pairs argument requires subsampling. The finite-sample error bound therefore applies to Stability Selection and 2D-SS with a fixed threshold $\tau$ or a size-matched selector (not the largest-gap rule), and not to Jitter-only.

\paragraph{Prediction under future noise.} Evaluating the refitted supports on fresh data (Appendix) shows the same from the prediction side. The cost of Stability Selection's confounder appears when the test noise differs from the training noise: at $\delta_{\mathrm{test}}=3$, the fits that included it have a $13\%$ higher error than those that excluded it. Averaged over all training sets, 2D-SS is worse than Stability Selection on clean test data and modestly better at larger noise. Its benefit is therefore the exclusion of the confounder, not a lower average prediction error.
}

{
\subsection*{Study 3: Does the PFER Bound Hold?}
Study 3 tests the bound of Theorem~\ref{thm:errctrl} for the size-matched selector, using $N=150$ complementary pairs, over $240$ settings per method ($n\in\{100,200\}$, $K\in\{1,5\}$, five values of $\delta_{\mathrm{obs}}$, $q\in\{5,10,20,40\}$, $\tau\in\{0.6,0.75,0.9\}$). 2D-SS satisfies the bound in every setting (Figure~\ref{fig:pfer_bound} shows $q=5$, $\tau=0.6$), and the largest ratio of the empirical PFER to the bound is $0.40$. Stability Selection exceeds the bound in $21$ of the $240$ settings, $19$ of them at $\delta_{\mathrm{obs}}=0$. There, the selection probability of the confounder ($0.61$ at $n=200$, $K=1$) breaks the exchangeability that the second form of the bound assumes: $0.61^2\approx0.37$ is about the empirical PFER at $K=1$. 2D-SS lowers this probability to $0.17$. For $\delta_{\mathrm{obs}}\ge0.5$, both methods are below the bound in all but two settings, and the inclusion rates are at most $0.02$ at $q=5$. The advantage is therefore concentrated at zero observed error (figures and details in the Appendix).
}

{
\begin{figure}[H]
\centering
\includegraphics[width=0.9\textwidth]{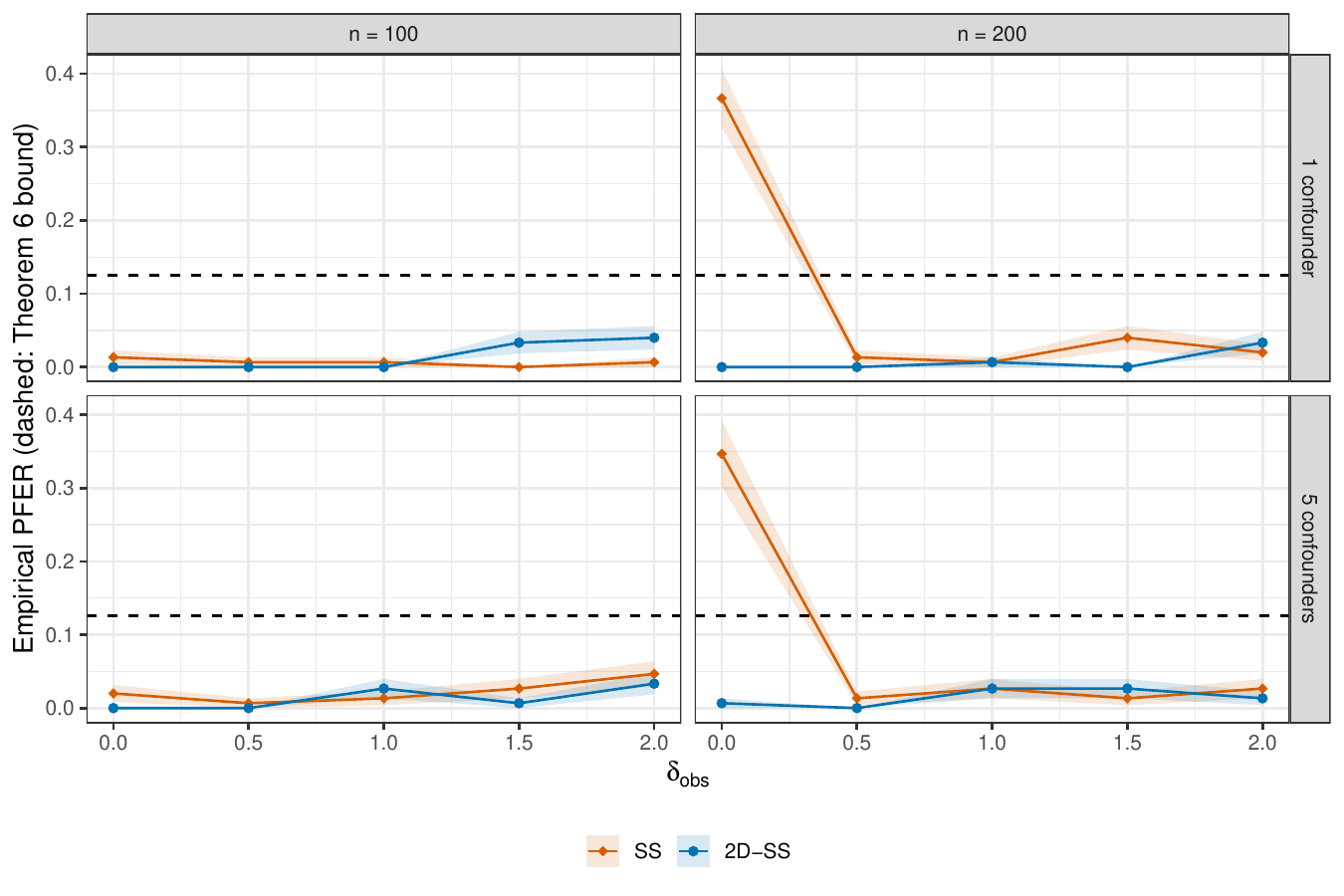}
\caption{Empirical PFER against $\delta_{\mathrm{obs}}$ for Stability Selection and
2D-SS ($q=5$, $\tau=0.6$); the dashed line is the bound of
Theorem~\ref{thm:errctrl}. Bands are $\pm$ one standard error over $N=150$ pairs.}
\label{fig:pfer_bound}
\end{figure}
}

{
\subsection*{Study 4: Comparison with Measurement-Error Baselines}\label{s3:baselines}
As a stress test, we compare 2D-SS with the measurement-error methods named in Section~\ref{s1:related}: the corrected Lasso \citep{loh2012high}, GMUS \citep{rosenbaum2010sparse,sorensen2018covariate} and CoCoLasso \citep{datta2017cocolasso}, together with the Lasso, Stability Selection and Jitter-only. The design is Example~\ref{ex:jy} embedded in $p=200$ variables, with $n=100$, $\delta_{\mathrm{obs}}\in\{0,0.5,1,1.5,2\}$ and $K\in\{0,1,5\}$ confounders ($40$ replicates per setting); the set-up, the tuning of the baselines and the figures are in the Appendix. In this design 2D-SS matches the best runnable noise-aware method, GMUS, once the observed error is $\delta_{\mathrm{obs}}\ge1$ (at $0.5$, within two standard errors for $K=1$ and just beyond them for $K=0$ and $5$), without a noise level or an error bound as input. It is ahead of or tied with the corrected Lasso and CoCoLasso under their package tuning, in particular when the noise level supplied to them is wrong, and it avoids the confounder at zero observed error where Stability Selection, GMUS and the corrected methods at their cross-validation minimum do not. Its $F_1$ score at low noise is below that of Stability Selection and GMUS, as the Appendix shows; part of that deficit is the largest-gap rule (Study~5). We therefore read the study as showing that 2D-SS is competitive with methods built for measurement error, not that it is better.
}

{
\subsection*{Study 5: When the Largest-Gap Rule Fails}
Study 5 asks where the largest-gap rule of Theorem~\ref{thm:gap} fails, with $s=5$ and $\beta=(6,5,w,0.9w,0.8w)^\top$ for $w\in\{0.3,0.5,0.8,1.2,2\}$ and $100$ replicates per setting. For Stability Selection the failure is confined to weak signals ($\hat s=2.77$ at $w=0.5$). For 2D-SS it is severe: $\hat s\approx2.0$ for every $w\le1.2$ because averaging over the noise grid lowers the weak signals' frequencies towards the noise floor. With $s$ known, the weak signals are recovered for $w\ge0.8$ by Stability Selection and for $w\ge1.2$ by 2D-SS (Appendix). Under the largest-gap rule Stability Selection is ahead of 2D-SS here, since there is no confounder for the jitter to repair, and with mixed signal strengths a known model size, or a threshold calibrated to the averaged frequencies, is the more reliable choice. The other rules, the figures and the details are in the Appendix.
}

\subsubsection*{Affymetrix Rat Genome 230 2.0 Array}
As a real-data example we use the Affymetrix Rat Genome 230 2.0 Array data of \citet{scheetz2006regulation}: $n=120$ twelve-week-old male F2 rats, with more than $31{,}000$ probe sets each. The response is the expression of the TRIM32 probe (\texttt{1389163\_at}), which has been linked to Bardet-Biedl syndrome \citep{chiang2006homozygosity}. After removing probes with low expression or a small range, $p=3{,}083$ probes remain. The columns are centred and scaled and the response is centred (details in the Appendix). We take $\lambda=\lambda_{\mathrm{1se}}$ from ten-fold cross-validation of the plain Lasso on the full data and keep it fixed for all procedures, so that differences reflect the aggregation and not the tuning (Remark~\ref{r:tuning}).

The plain Lasso selects $15$ probes. Stability Selection selects none at the usual threshold $\tau=0.6$, and $2$ probes at the threshold chosen by the largest-gap rule. Jitter-only selects $4$ probes, which include those $2$, and 2D-SS selects $5$, which include the $4$ of Jitter-only. There is no ground truth, so the different numbers of probes are not evidence that one method selects better, and the additional probes may be false discoveries.

To check the selections beyond a single fit, we repeated the analysis on $50$ random splits into $80$ training and $40$ test rats (Appendix). At equal model size ($k=3$, $5$, $10$ probes), the four procedures are close in test prediction, with relative error $0.54$--$0.73$. Jitter-only and the Lasso path are at least as good as 2D-SS in paired comparisons, and they have a lower mean error than Stability Selection. The difference between Stability Selection and 2D-SS is small. Under their own largest-gap rules, Stability Selection and 2D-SS keep only about two probes on average and predict worse than the Lasso, in line with Study~5. We read these results as a description of the procedures on one data set with weak signals, not as a predictive advantage. The claim of the paper concerns confounders and error control.

{
\subsection*{Study 6: Is the Random Noise Needed? Deterministic Elastic-Net Counterparts}\label{study:enet}
By Proposition~\ref{prop:ridge}, the expected effect of the noise on a Lasso fit is a ridge term, and Studies~1--5 do not separate that effect from the randomness of the noise, which is what 2D-SS adds to a deterministic ridge path. We therefore add three elastic-net procedures (Remark~\ref{r:enetss}; the matching is given in the Appendix). ENet-SS is the procedure of Algorithm~\ref{alg:delta_averaged_selection} with the injected noise replaced, at each level $\delta$, by the elastic net with ridge weight $\delta^2$. SS-ENet is Stability Selection with one fixed elastic net, whose ridge weight is $\gamma^2=0.25$, and 2D-SS-ENet is 2D-SS with that elastic net as base selector. The value of $\gamma$ was chosen once and not tuned. We run these together with Stability Selection and 2D-SS in the design of Study~2 ($p=1{,}000$, $K\in\{1,5\}$ confounders, $n\in\{100,200,400\}$, $\delta_{\mathrm{obs}}=0$, $200$ training sets), with fresh random numbers. All five use the same half-sample size, the same $\lambda=\sqrt{\log p/n_1}$ and the largest-gap rule. ENet-SS, 2D-SS and 2D-SS-ENet use the same $m=5$ levels in $[0,1]$ with $B=50$ fits each, whereas Stability Selection and SS-ENet use the matched $250$ fits at $\delta=0$. Because the random numbers are new, the entries for Stability Selection and 2D-SS differ slightly from those of Study~2.

\begin{center}
\small
\refstepcounter{table}\label{tab:enetss}
\begin{minipage}{0.95\textwidth}
\textbf{Table~\thetable.} Study 6. Percentage of training sets in which at least one confounder is selected and in which $S$ is recovered exactly, at $\delta_{\mathrm{obs}}=0$ with the largest-gap rule ($200$ training sets per row; Monte Carlo standard errors at most $3.6$ points).
\end{minipage}

\vspace{4pt}
\begin{tabular}{rrrrrrr}
\toprule
$K$ & $n$ & SS & SS-ENet & ENet-SS & 2D-SS & 2D-SS-ENet\\
\midrule
\multicolumn{7}{l}{\textit{Any confounder selected (\%)}}\\
1 & 100 & 16.5 & 0.0 & 0.0 & 0.0 & 0.0 \\
1 & 200 & 69.5 & 0.5 & 1.0 & 1.5 & 0.0 \\
1 & 400 & 73.0 & 2.0 & 2.0 & 5.0 & 0.0 \\
\addlinespace
5 & 100 & 15.5 & 0.0 & 0.5 & 0.0 & 0.0 \\
5 & 200 & 76.0 & 1.5 & 0.5 & 0.0 & 0.0 \\
5 & 400 & 98.5 & 2.5 & 3.5 & 0.0 & 0.0 \\
\midrule
\multicolumn{7}{l}{\textit{Exact recovery of $S$ (\%)}}\\
1 & 100 & 1.5 & 0.0 & 0.0 & 0.0 & 0.0 \\
1 & 200 & 21.0 & 18.0 & 16.0 & 4.5 & 1.0 \\
1 & 400 & 26.5 & 79.5 & 78.5 & 68.0 & 18.5 \\
\addlinespace
5 & 100 & 0.0 & 0.0 & 0.0 & 0.0 & 0.0 \\
5 & 200 & 0.5 & 13.5 & 5.0 & 0.0 & 0.5 \\
5 & 400 & 0.5 & 78.0 & 48.0 & 8.5 & 16.5 \\
\bottomrule
\end{tabular}
\end{center}

The three elastic-net procedures exclude the confounder as well as 2D-SS: for $n\ge200$ each of them selects it in at most $3.5\%$ of training sets (2D-SS: at most $5.0\%$), against $69.5$--$98.5\%$ for Stability Selection (Table~\ref{tab:enetss}). The complementary-pairs probabilities $p_k(\delta)$ of Section~\ref{s2:errctrl} agree; at $n=400$ and $K=1$ they are $0.79$ for Stability Selection, $0.39$ for 2D-SS, $0.33$ for ENet-SS, $0.17$ for SS-ENet and $0.15$ for 2D-SS-ENet. The fixed elastic net does as well as the path: at $n\ge200$, SS-ENet and ENet-SS recover $S$ exactly more often than 2D-SS (at $n=400$, $79.5\%$ and $78.5\%$ against $68.0\%$ for $K=1$, and $78.0\%$ and $48.0\%$ against $8.5\%$ for $K=5$). Adding noise to the fixed elastic net does not help. 2D-SS-ENet never selects the confounder at $n\ge200$, but at $n=400$ it recovers $S$ exactly in only $18.5\%$ ($K=1$) and $16.5\%$ ($K=5$) of training sets, since it keeps fewer relevant variables ($3.7$ of $5$ on average); its ridge weight is the largest of the five at every level (tied with SS-ENet at $\delta=0$), which is consistent with over-shrinkage. With measurement error already in the training data ($\delta_{\mathrm{obs}}\in\{0.5,1\}$), the confounder is rarely selected by any procedure and the procedures differ in how conservative they are: 2D-SS makes fewer false positives than each of the other four in every setting (at most $0.09$ per training set), at the price of fewer true positives (Appendix). On fresh data with larger noise the five predict alike, and at $\delta_{\mathrm{test}}=3$ the largest paired difference from ENet-SS is $4.0$ out of $183$ (Appendix).

For the Lasso, then, the effect of the noise axis in this design is carried by its ridge term, as Proposition~\ref{prop:ridge} says, and not by its randomness: the same exclusion is obtained without any noise, whether the ridge weight is fixed (SS-ENet) or follows a path (ENet-SS), and noise added to a fixed elastic net removes the confounder at a cost in exact recovery. The single elastic-net fit of Study~2 is not a counterexample, since it is a single fit on the full sample and is not aggregated over half-samples. We ran this comparison for one design, with the Lasso as the base selector of 2D-SS and one fixed ridge weight for the elastic net. We do not claim that the randomness is never useful. The one place where it changed the outcome is the noisy-training-data case, where it made the selection more conservative.
}

\section{Discussion}\label{s4}
\subsection{Summary of the findings}
Stability Selection can select an irrelevant variable that is correlated with the relevant ones in a large fraction of subsamples, because it resamples the observations and leaves the design unchanged. 2D-SS adds noise to the design over a grid of noise levels. On average the noise acts as a ridge penalty, so a jittered Lasso fit is a noisy elastic net. As the noise grows, the irrepresentable condition moves from the Lasso condition to a marginal-correlation condition. Jittering can therefore restore selection where the Lasso condition fails, and it can also hurt. We gave an exact calibration criterion for uncorrelated relevant variables (Corollary~\ref{cor:compound} gives the analogue for compound symmetry, for noise along the all-ones direction of $S$, which 2D-SS does not inject), recovery guarantees for a single fit and for the aggregated procedure, and a complementary-pairs bound on the expected number of false positives. Our guarantee needs no measurement-error covariance and covers designs where the clean-data Lasso condition fails, but its domain is not nested with that of CoCoLasso (Remark~\ref{r:cocolasso_compare}).

In the simulations, 2D-SS excludes a confounder that Stability Selection includes. At $\delta_{\mathrm{obs}}=0$ and $K=1$, the inclusion rate is $0.725$ for Stability Selection at $n=200$ and at most $0.075$ for 2D-SS at any $n$ (Study~2). The complementary-pairs rate falls from $0.61$ to $0.17$ at $n=200$ (Study~3). The size-matched version stays below the closed-form false-positive bound in all $240$ settings, whereas Stability Selection exceeds it in $21$. This is an empirical finding, because the closed form needs exchangeable irrelevant variables and confounders violate that.

For the Lasso, an elastic net without any noise, either along a deterministic path (ENet-SS) or with one fixed ridge weight (SS-ENet), excludes the confounder equally well and recovers the support more often than 2D-SS at $n\ge200$ (Study~6). Adding noise to the fixed elastic net (2D-SS-ENet) also excludes it but recovers the support less often than SS-ENet. The effect is therefore carried by the ridge term of Proposition~\ref{prop:ridge}, not by the randomness of the noise. We do not claim that 2D-SS has a higher $F_1$ score than Stability Selection, and in the main comparison (Study~4, largest-gap rule) it does not (Appendix).

\subsection{When the noise axis helps}
It helps when the clean-design irrepresentable condition fails and isotropic noise can repair it. In the uncorrelated design of Study~1, recovery rises with $\nu$ up to a peak near $\nu\approx0.3$--$0.4$, and the gain over the plain fit grows with the sample size. It also helps when Stability Selection includes a confounder in a substantial fraction of replications (Studies~2--3). Compared with methods built for measurement error, 2D-SS is ahead of or tied with the corrected Lasso and CoCoLasso under their package tuning, and it matches GMUS when $\delta_{\mathrm{obs}}\ge1$, without a noise level as input (Study~4).

It helps little in five cases. First, when the relevant variables are correlated in a way that isotropic noise worsens, or when the irrepresentable condition already holds (Study~1). Second, when the measurement error is already large, because the effective noise level of Stability Selection is then positive and the gap closes (Studies~3--4). Third, when signal strengths are very heterogeneous and the largest-gap rule must estimate the model size, since the rule can stop at the strong signals, more severely for 2D-SS (Study~5). Fourth, when the aim is prediction with weak signals, as in the rat data, where Jitter-only and the Lasso path are at least as good at equal model size. Fifth, when an elastic net without noise, fixed or along a path, is an option (Study~6). In every case the cost is $m$ times that of Stability Selection (Section~\ref{s2:delta_avg}).

\subsection{Practical use}
We suggest standardising the variables and fixing $\lambda$ across noise levels, for instance at $\sqrt{\log p/n_1}$ (Remark~\ref{r:tuning}). A few levels with uniform weights are enough, for example $m=5$ levels in $[0,1]$ with $B=50$ fits per level. The fits are independent and can be run in parallel. To compare with Stability Selection at equal compute, give it $mB$ half-sample fits at $\delta=0$. The largest-gap rule suits strong signals of similar size. With weak or mixed signals, use a known or estimated model size or a threshold calibrated to the averaged frequencies, and note that the error bound of Theorem~\ref{thm:errctrl} needs $\tau>1/2$. We recommend reporting the noise path of the frequencies as well as the selected set. If confounders are the main concern, little weight should be put near $\delta=0$; in the sensitivity study, omitting that slice lowered the confounder-inclusion rate of 2D-SS from $0.24$ to $0.03$ for $K=1$ with the model size known (Stability Selection: $0.57$). This departs from Stability Selection, and we have not studied it beyond that design. For the Lasso alone, ENet-SS or SS-ENet can be run alongside 2D-SS.

\subsection{Limitations and future work}
The guarantees are local and conservative, and they assume Gaussian measurement error with a common variance. The noise level is calibrated exactly only for $\Sigma_{SS}=I_s$, and Proposition~\ref{prop:directional} covers the correlated case only qualitatively, with noise that 2D-SS does not inject. The largest-gap rule can stop at the strong signals (Study~5). The comparison with the elastic-net variants covers one design, the Lasso as the base selector of 2D-SS and one fixed ridge weight, and the method does not address general distribution shift.

Open directions are non-Gaussian and non-additive error, estimates of the model size with finite-sample guarantees, and sharper constants, so that the guarantees are informative at the sample sizes simulated here. Another direction is a variant that decorrelates the variables before jittering \citep{nouraie2025stability}, which could extend the exact calibration criterion from uncorrelated to correlated relevant variables.

\section*{Competing Interests}
The authors declare that they have no conflict of interest.
\section*{Author Contributions Statement}
Mahdi Nouraie was responsible for drafting the manuscript, the development of the research methodology and for writing the computer code used throughout. Samuel Muller and Houying Zhu provided critical feedback on the content of the manuscript, refining the clarity and scope of the manuscript and the computer code.
\section*{Use of AI Tools}
The authors used Claude (Anthropic) for language editing, for checking derivations and for help with the code. They reviewed all AI-assisted output and take full responsibility for the content of the paper.
\section*{Data Availability}
The rat microarray data can be obtained from the National Center for Biotechnology Information (NCBI) website at \url{www.ncbi.nlm.nih.gov}, under accession number GSE5680.
{
\section*{Supplementary Material}
The R code that reproduces every study and figure in this paper, including the preprocessing of the rat data, is available at \url{https://github.com/MahdiNouraie/2DSS}.
}
\section*{Funding}
Mahdi Nouraie was initially supported by the Macquarie University Research Excellence Scholarship (20213605), and his work continued under the same Australian Research Council Discovery Project Grant (DP260100348) that supported Samuel Muller and Houying Zhu.
\bibliographystyle{bst/sn-mathphys-num}
\bibliography{sn-bibliography}
\begin{appendices}
{
\section*{Notation}
\begin{center}
\small
\begin{tabular}{@{}p{0.2\textwidth}p{0.74\textwidth}@{}}
\hline
$n$, $p$, $s$ & sample size, number of variables, number of relevant variables \\
$S$, $S^c$ & support of $\beta$ and its complement; $z=\operatorname{sign}(\beta_S)$ \\
$X_0$, $X$, $W_0$ & latent design, observed design $X=X_0+W_0$, measurement error with i.i.d.\ $\mathcal N(0,\delta_{\mathrm{obs}}^2)$ entries \\
$\delta_{\mathrm{obs}}$ & standard deviation of the measurement error already in the data \\
$\delta$, $W$ & standard deviation of the injected noise and the injected noise matrix \\
$\nu_\delta$, $\nu$ & effective noise level, $\nu_\delta^2=\delta_{\mathrm{obs}}^2+\delta^2$ \\
$\mathcal G$, $w_\delta$, $m$ & grid of injected noise levels, weights on it, number of levels \\
$B$, $n_1$, $I_b$ & replicates per level, half-sample size $\lfloor n/2\rfloor$, $b$th half-sample \\
$\hat f_j^{(\delta)}$, $\widehat{\bar f}_j$ & selection frequency of variable $j$ at level $\delta$, and averaged over $\mathcal G$ with weights $w_\delta$ \\
$\lambda$; $\lambda_1$, $\lambda_2$ & Lasso penalty; $\ell_1$ and ridge weights of the elastic net, normalised as in \citet{jia2010model} ($\lambda_1=2n_1\lambda$, $\lambda_2=n_1\nu^2$) in Section~\ref{s2:small}, and with the loss multiplied by $1/(2n_1)$ in Proposition~\ref{prop:other}(a) \\
$\tau$, $\hat\tau$ & threshold on selection frequencies; threshold chosen by the largest-gap rule \\
$\hat s$, $\hat S$ & estimated model size and selected set \\
$\vartheta$ & $\vartheta^2=\nu^2+\lambda_2$, the elastic-net weight in the Appendix section on other base selectors \\
$T$, $\lambda_T$, $\kappa_T$ & $T=\Sigma_{SS}+\nu^2I_s$, $\lambda_T=\lambda_{\min}(\Sigma_{SS})+\nu^2$ and $\kappa_T=\|T\|/\lambda_T$ (Theorem~\ref{thm:gauss_rate}) \\
$v_\nu$ & $v_\nu=z+(\nu^2/\lambda)\beta_S$ (Section~\ref{s2:small}) \\
$s_q$, $\mu_k$, $\sigma_x$ & $s_q^2$ is the variance of the response remainder, $\mu_k=\nu^2a_kT^{-1}\beta_S$, and $\sigma_x=\max_{k\in S^c}(\Sigma_{kk}+\nu^2)^{1/2}$ (Lemma~\ref{lem:gstruct} and Theorem~\ref{thm:gauss_rate}, Section~\ref{s2:single}) \\
$\varpi_k$ & the quantity $\varpi_k^2=s_k^2-\mu_k^2/s_q^2$ in Theorem~\ref{thm:exact} \\
$K$ & bound on the entries of $X_0$ (Section~\ref{s2}); number of confounders (Section~\ref{s3}) \\
\emph{confounder} & an irrelevant variable that is correlated with the relevant ones (not a causal confounder) \\
$A_k$, $B_k$ & $A_k=a_k\Sigma_{SS}^{-1}z$ and $B_k=a_k\Sigma_{SS}^{-1}\beta_S$ with $a_k=\Sigma_{kS}$, for $k\in S^c$ (Proposition~\ref{prop:repair}) \\
$c_\nu^k$, $c_\nu$ & elastic-net irrepresentable statistic of variable $k$ at noise level $\nu$, and $c_\nu=\max_{k\in S^c}|c_\nu^k|$; EIC$_\nu$ is $c_\nu<1$ (with margin $\eta$) \\
$\eta$, $\eta_0$ & margins by which EIC$_\nu$ holds ($c_\nu\le1-\eta$) or is violated ($c_\nu\ge1+\eta_0$) \\
$\kappa_\nu$, $\omega$, $\omega_0(\nu)$ & $\kappa_\nu=\|(\Sigma_{SS}+\nu^2I_s)^{-1}\|_\infty$; $\omega$ is the bound on the perturbation blocks of the Gram matrix and $\omega_0(\nu)$ the admissible tolerance, a sufficient value (Lemma~\ref{l1}) \\
$\beta^\star_\nu$, $\theta_\nu$ & shifted population solution and its smallest signed entry on $S$ (sign margin) \\
$\rho_\nu$, $\Gamma_\nu$, $\Psi_\nu$ & auxiliary constants of Theorem~\ref{thm:single_level}: $\rho_\nu$ bounds the fluctuation of the response remainder, and $\Psi_\nu$ is the bound that the sign margin $\theta_\nu$ must exceed in condition (iv) \\
$\alpha$, $\alpha_I$, $\bar\alpha$ & failure probabilities; $\bar\alpha=\alpha_I+2\alpha+e^{-n_1/8}$ for a single fit \\
$\mathcal G_{\mathrm{good}}$, $\pi$ & noise levels at which conditions (i)--(iv) of Theorem~\ref{thm:single_level} hold, and their total weight $\pi=\sum_{\delta\in\mathcal G_{\mathrm{good}}}w_\delta$ \\
$\gamma$, $\bar\varphi$ & separation margin of Assumption~\ref{as:avg_sep}; contribution of the remaining levels to the frequency of an irrelevant variable (Proposition~\ref{prop:bridge}) \\
$\varepsilon_{\mathrm{MC}}$, $\alpha_{\mathrm{MC}}$ & Monte Carlo accuracy and its failure probability (Theorem~\ref{thm:delta_avg_recovery}) \\
$\mathcal D$ & the data; $p_k(\delta)=\mathbb P(k\in\hat S(A;\delta))$ is a selection probability over data, half-sample and jitter \\
$\bar q_\delta$ & $\bar q_\delta=\mathbb E|\hat S(A;\delta)\cap S^c|$, the expected number of irrelevant variables selected by one fit at level $\delta$ (Theorem~\ref{thm:errctrl}) \\
$q_\star$, $V_\tau$ & number of variables per fit of the size-matched selector; number of irrelevant variables selected at threshold $\tau$ \\
PFER & per-family error rate, the expected number of false positives $\mathbb E\,V_\tau$ \\
\emph{local symbols} & some letters are reused with a local meaning in a single statement or proof (for example $m$ is also the margin $1-\eta$ in Proposition~\ref{prop:repair}, $a$ and $\gamma$ are the SCAD and MCP parameters in Proposition~\ref{prop:other}, $\kappa$ is also a probability level in the aggregation results and a weight tilt in the sensitivity study, $\gamma$ is also the ridge level of SS-ENet and 2D-SS-ENet, $\rho$ is also the compound-symmetry parameter and $\rho_\delta=\bar q_\delta/(p-s)$ in Remark~\ref{r:ipss}, $\theta_L$, $\hat\pi_k$, $g_k$ and $m_j$ appear only in Theorem~\ref{thm:errctrl} and Remark~\ref{r:ipss}, $\varpi_k$ also appears in the proof of Theorem~\ref{thm:gauss_rate}, $\hat c_\nu=(c_\nu^k)_{k\in S^c}$, and $\ell$, $L$, $q$, $V$, $T$, $\xi$ and $E$ are each used with more than one meaning); the meaning is always defined where the symbol is used, and the places where a letter changes meaning are flagged in the text \\
\hline
\end{tabular}
\end{center}
}

\section*{Additional Numerical Results}
{This appendix gives the full account of Studies~1, 3 and~5, which are summarised in Section~\ref{s3}, additional results for Studies~2 and~4, the implementation and further results of Study~6, the repeated train--test splits on the rat data, the sensitivity of 2D-SS to the noise grid and the weights, practical recommendations, and the proposition on other base selectors.}

{
\subsection*{Motivating example: design and computation}
\paragraph{Design.} The data are $n=100$, $p=1{,}000$ and $n_1=\lfloor n/2\rfloor=50$. Five variables are relevant and uncorrelated, $\Sigma_{SS}=I_5$, with coefficients $\beta_S=\lambda\,(4,4,4,4,15)^\top$, where $\lambda=\sqrt{\log(p)/n_1}\approx0.37$. We call the first four the \emph{weak} and the last one the \emph{strong} relevant variable. Ten irrelevant variables have unit variance and covariance $a=(0.4,0.4,0.4,0.4,-0.447)^\top$ with the relevant variables (each is the linear predictor $a^\top x_S$ plus independent noise). The remaining $985$ irrelevant variables are equicorrelated with correlation $0.1$. One design matrix is generated and kept fixed for all fits, and $Y=X\beta+\varepsilon$ with $\varepsilon\sim\mathcal N(0,I_n)$.
\paragraph{Irrepresentable statistic.} This design has the structure of Example~\ref{ex:jy}. With $z=\operatorname{sign}(\beta_S)$, the Lasso irrepresentable statistic is $|a^\top z|=1.153>1$, so the condition fails. The statistic of the jittered fit (Section~\ref{s2:small}) is $c_\nu=|1.153-0.3\nu^2|/(1+\nu^2)$, which is less than $1$ for $\nu$ larger than about $0.34$. The code re-standardises the jittered columns, which amounts to the level-dependent penalty $\lambda\sqrt{1+\nu^2}$, so this formula is approximate.
\paragraph{Computation.} We use $\delta_{\mathrm{obs}}=0$ (so $\nu=\delta$) and $B=100$ half-samples at each of $14$ levels, $\delta\in\{0,0.025,0.05,0.1,0.15,0.2,0.3,0.4,0.5,0.75,1,1.5,2,3\}$, where $\delta=0$ is the Stability Selection slice. Each half-sample receives noise of variance $\delta^2$ and is fitted by the Lasso in \texttt{glmnet}~\citep{friedman2010regularization} at the fixed $\lambda=\sqrt{\log(p)/n_1}$, so that only $\delta$ changes between fits. The mean frequency of the ten confounders is $0.47$ at $\delta=0$, $0.41$ at $\delta=0.1$, $0.09$ at $\delta=0.5$ and at most $0.04$ for $\delta\ge0.75$, in line with $c_\nu$ falling below one at $\nu\approx0.34$. At $\delta=0.75$ the weak relevant variables (mean $0.29$) are all above the correlated irrelevant ones (at most $0.05$). The loss of the relevant variables at large $\delta$ agrees with the dependence of condition (iii) of Theorem~\ref{thm:single_level} on the noise level.

\subsection*{Study 1: Corollary~\ref{cor:benefit} in a Single Fit}
\begin{figure}[H]
\centering
\includegraphics[width=\textwidth]{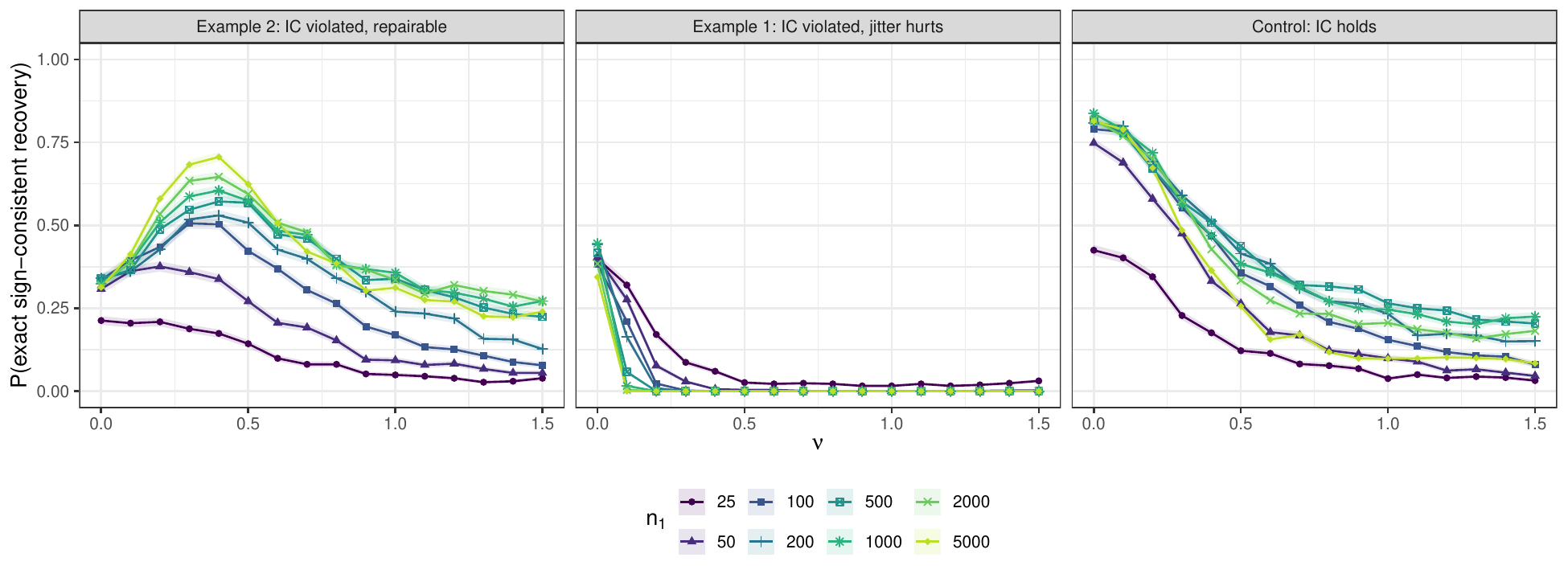}
\caption{Probability of exact sign-consistent recovery as a function of the
noise level $\nu$ for a single plain ($\nu=0$) or jittered fit, by design
(Example~\ref{ex:jy}, left; Example~\ref{ex:corr}, middle; control with the
irrepresentable condition satisfied, right) and sample size $n_1$. Shaded bands are
$\pm$ one Monte Carlo standard error over the $1{,}000$ trials at each point.}
\label{fig:corollary1_curves}
\end{figure}
Before turning to the aggregated procedure, we illustrate the mechanism it builds
on directly: a single plain fit versus a single jittered fit, as in
Corollary~\ref{cor:benefit}. We use three designs, each with one confounder. Each fit standardises the jittered columns before the penalty is applied, so the population statistic $c_\nu$ is approximately, not exactly, the statistic that the fit sees.
Example~\ref{ex:jy} ($\Sigma_{SS}=I_5$,
$\beta_S=(-4,-2,0.5,0.6,0.7)^\top$, $a=(1/8,1/4,1/2,1/2,1/2)^\top$, $|a^\top z|=1.125$) is a
design where isotropic jittering repairs the violated irrepresentable condition.
Example~\ref{ex:corr} ($\Sigma_{SS}$ correlated, $\beta_S=(2.5,1.5)^\top$, $a=(0.8,1.2)^\top$)
is a design where it does not and can make it worse (Remark~\ref{r:nonmono}). As a
control, we keep $\Sigma_{SS}$ and $\beta_S$ of Example~\ref{ex:jy} but take
$a=(0.05,0.1,0.1,0.1,0.1)^\top$, so that the Lasso condition holds ($|a^\top z|=0.15$) and
there is nothing to repair. For each design and each sample size
$n_1\in\{25,50,100,200,500,1000,2000,5000\}$, we draw $1{,}000$ datasets at
each of $16$ noise levels $\nu\in[0,1.5]$, fit the Lasso at
$\lambda=\sqrt{\log p/n_1}$ {(with $p=s+1$: $p=6$ for Example~\ref{ex:jy} and the control, $p=3$ for Example~\ref{ex:corr})}, and record the fraction of fits
that recover $S$ exactly with the correct signs while excluding the
confounder.

\begin{figure}[H]
\centering
\includegraphics[width=\textwidth]{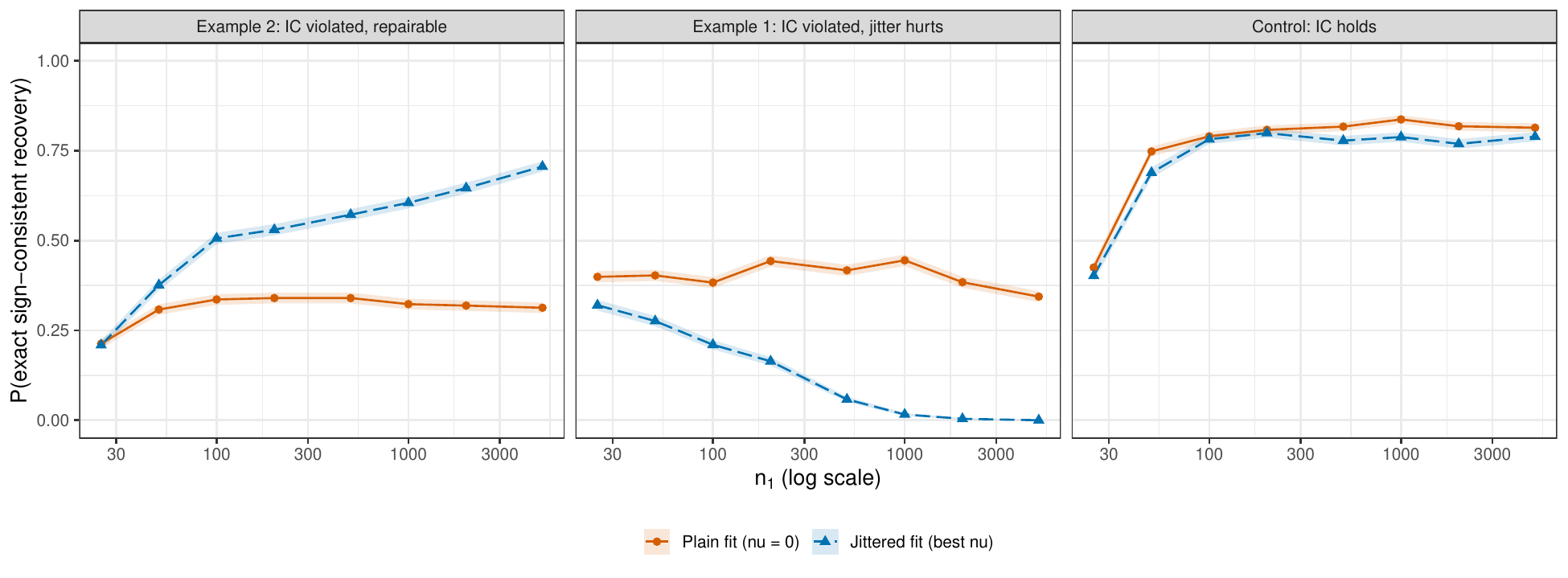}
\caption{The same comparison summarised at each $n_1$: the plain fit
($\nu=0$) against the best jittered fit (the noise level $\nu>0$ with the
highest recovery rate at that $n_1$), by design. The best noise level is chosen
after the fact, so the gap is an upper bound for a fixed $\nu$.}
\label{fig:corollary1_vs_n1}
\end{figure}

The designs separate in the directions the theory predicts.
Figure~\ref{fig:corollary1_curves} shows how recovery changes as $\nu$ moves away
from $0$. In Example~\ref{ex:jy} (left), the curves for $n_1\ge100$ rise to a peak at
$\nu\approx0.3$--$0.4$ before falling back, with the peak growing taller as $n_1$
increases; the population statistic $c_\nu$ of Section~\ref{s2:small} is below one
from {$\nu\approx0.30$ at $n_1=25$, falling to $\nu\approx0.14$ at $n_1=5{,}000$}, and, for large $n_1$ (small $\lambda$), it
returns above one at larger $\nu$ ($\nu\approx0.95$ at $n_1=2{,}000$ and
$\nu\approx0.7$ at $n_1=5{,}000$). In Example~\ref{ex:corr} (middle), every curve
falls quickly once noise is added (most steeply for large $n_1$), and $c_\nu$ increases with $\nu$
from $c_0=1.33$ at every $n_1$. In the control (right), recovery is highest at
$\nu=0$ and decreases with $\nu$ overall.
Figure~\ref{fig:corollary1_vs_n1} summarises this as a function of $n_1$. In
Example~\ref{ex:jy}, the plain fit recovers $S$ exactly in about $31$--$34\%$ of
fits for $n_1\ge50$ and does not improve with $n_1$, as expected when the condition
is violated. The best jittered fit does better and the gap widens with $n_1$: the
best rate is $0.21$ at $n_1=25$ (equal to the plain fit, where little data leaves
little room for noise to help), $0.51$ at $n_1=100$, $0.61$ at $n_1=1{,}000$ and
$0.71$ at $n_1=5{,}000$, a gain of $0.39$ over the plain fit, in line with
Corollary~\ref{cor:benefit}. In Example~\ref{ex:corr} the plain fit recovers $S$ in
about $34$--$45\%$ of fits at every $n_1$, while the best achievable jittered rate
\emph{declines} with $n_1$, from $0.32$ at $n_1=25$ to $0.02$ at $n_1=1{,}000$ and
essentially $0$ at $n_1=2{,}000$--$5{,}000$: with more data, isotropic jittering
moves further from helping, not closer. In the control, the plain fit already
recovers $S$ in $79$--$84\%$ of fits for $n_1\ge100$ and jittering does not
improve on it: the best jittered rate is $0.77$--$0.80$, and recovery falls to about
$0.57$ at $\nu=0.3$ for $n_1=500$. This decline is not a failure of the population
statistic, which stays below one up to $\nu\approx{0.66}$ at $n_1=500$; the jittered
fit also carries response noise of variance $\nu^2\|\beta_S\|_2^2\approx21\nu^2$, which
enters through the remainder in condition (iii) of Theorem~\ref{thm:single_level}. The three designs together
show when jittering helps (violated condition that noise repairs), when it does
not (violation that noise worsens) and what it costs when there is nothing to
repair.

\begin{figure}[H]
\centering
\includegraphics[width=\textwidth]{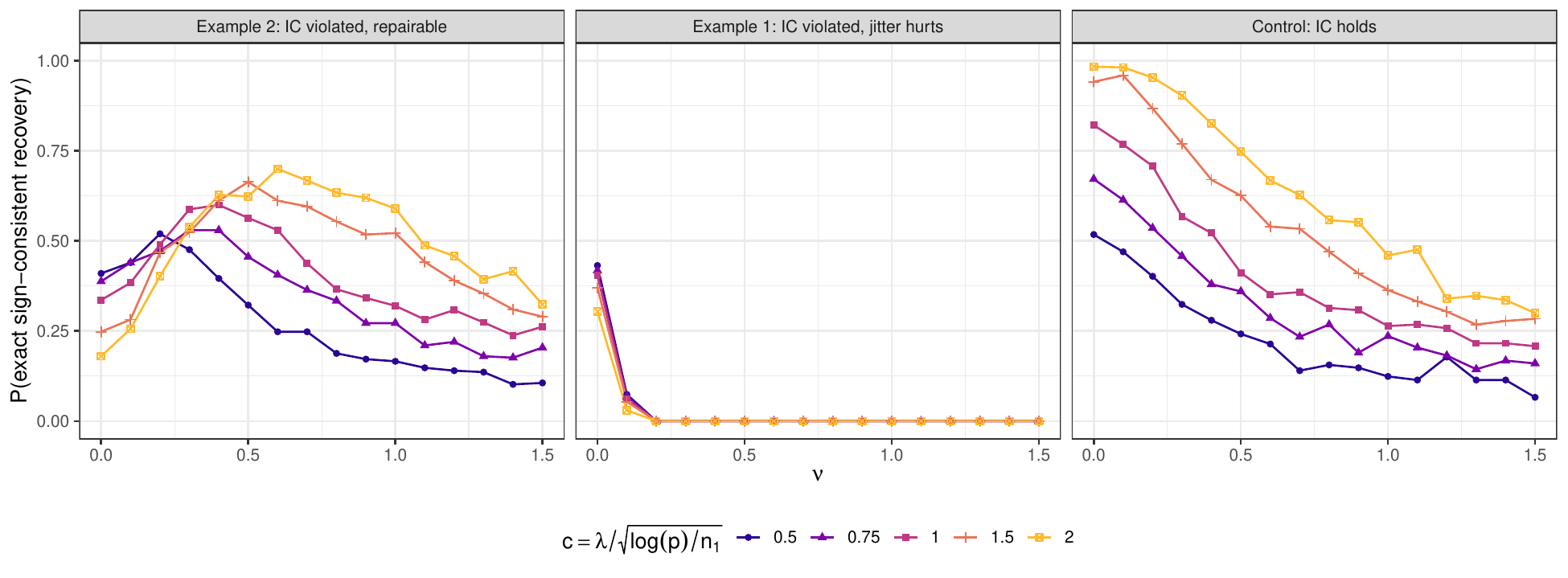}
\caption{Sensitivity to the tuning parameter: probability of exact sign-consistent
recovery at $n_1=500$ for $\lambda=c\sqrt{\log p/n_1}$ with
$c\in\{0.5,0.75,1,1.5,2\}$, by design (as in Figure~\ref{fig:corollary1_curves}).
Each point is based on $500$ trials.}
\label{fig:corollary1_lambda}
\end{figure}

Figure~\ref{fig:corollary1_lambda} repeats the comparison at $n_1=500$ for
different values of the tuning constant $c$. In Example~\ref{ex:jy}, the plain fit
falls from $0.41$ to $0.18$ as $c$ increases from $0.5$ to $2$, whereas the best
jittered rate rises from $0.52$ to $0.70$, so that the gain over the plain fit
grows from $0.11$ to $0.52$. In Example~\ref{ex:corr}, jittering lowers recovery at
every $c$ (from $0.30$--$0.43$ for the plain fit to at most $0.074$ at $\nu=0.1$).
In the control, it does not help at any $c$: at $c=1.5$ and $c=2$ the rates at
$\nu=0$ and $\nu=0.1$ are $0.94$ and $0.96$, and $0.98$ and $0.98$, within Monte Carlo
error, and they decrease for larger $\nu$.
}

{
\subsection*{Study 3: Does the PFER Bound Hold?}\label{study:pfer}
Studies 1--2 concerned support recovery. We now test the error-control guarantee of
Theorem~\ref{thm:errctrl}. We use the design of Study 2 ($p=1{,}000$,
$\beta_S=(-4,-2,0.5,0.6,0.7)^\top$) with $K\in\{1,5\}$ confounders and $n\in\{100,200\}$, and
the size-matched selector, which takes the first $q$ variables to enter the Lasso path
(variables that never enter are ranked at random) and has no tuning parameter (Remark~\ref{r:tuning}(c)), with complementary pairs exactly as in
Section~\ref{s2:errctrl}: each pair is drawn from an independently sampled dataset ($B=1$, so $V_\tau$ is computed from one pair per dataset and averaged over the $N$ replicates) and
split into two disjoint half-samples, consistent with $p_k(\delta)$ being a fully
marginal probability over data, half-sample and noise. For each of five
measurement-error levels $\delta_{\mathrm{obs}}\in\{0,0.5,1,1.5,2\}$ and each method
(Stability Selection: $\delta=0$; 2D-SS: five levels in $[0,1]$) we use $N=150$ pairs and
check the bound $q^2/((p-s)(2\tau-1))$ over $q\in\{5,10,20,40\}$ and
$\tau\in\{0.6,0.75,0.9\}$, which gives $240$ settings per method. We also track the
confounder's own complementary-pairs rate $\bar f^{\,\mathrm{cp}}_k$ (averaged over the
$K$ confounders).

\begin{figure}[H]
\centering
\includegraphics[width=0.6\textwidth]{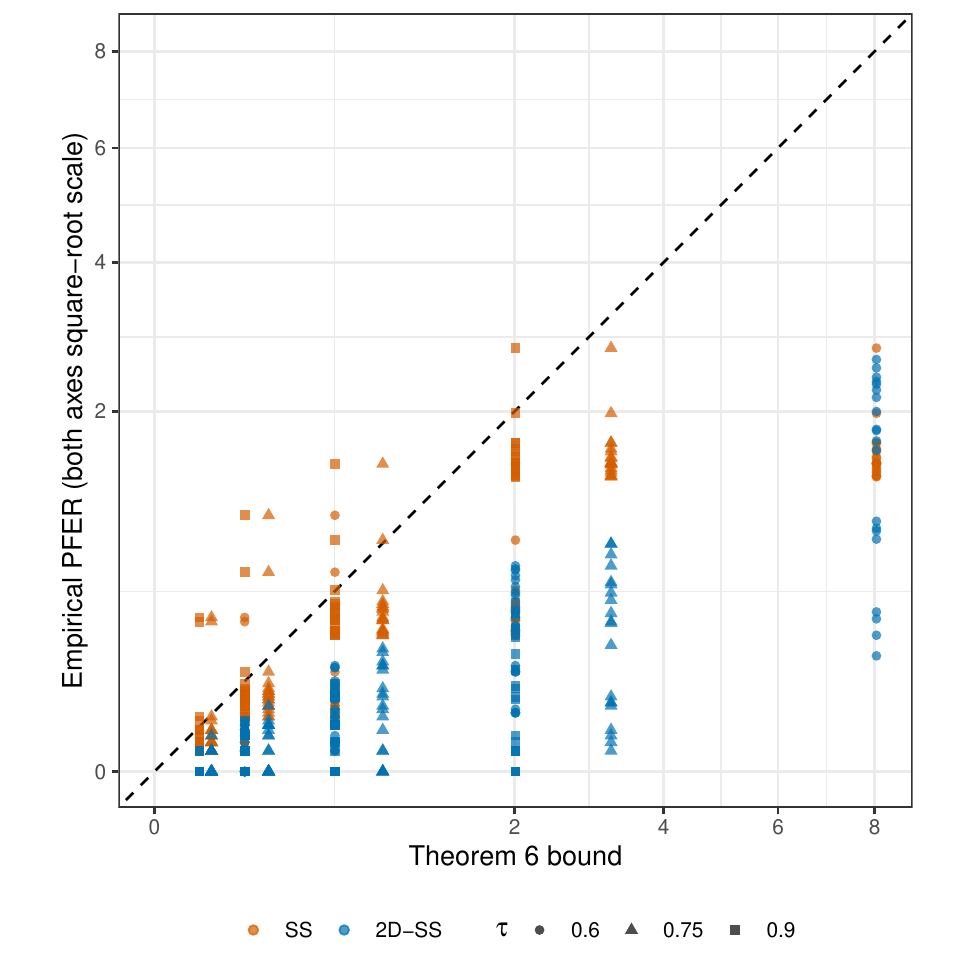}
\caption{Empirical PFER against the bound in every setting
($n,K,\delta_{\mathrm{obs}},q,\tau$; square-root scales). Points on or below the
diagonal satisfy the bound.}
\label{fig:pfer_all}
\end{figure}

\begin{figure}[H]
\centering
\includegraphics[width=0.9\textwidth]{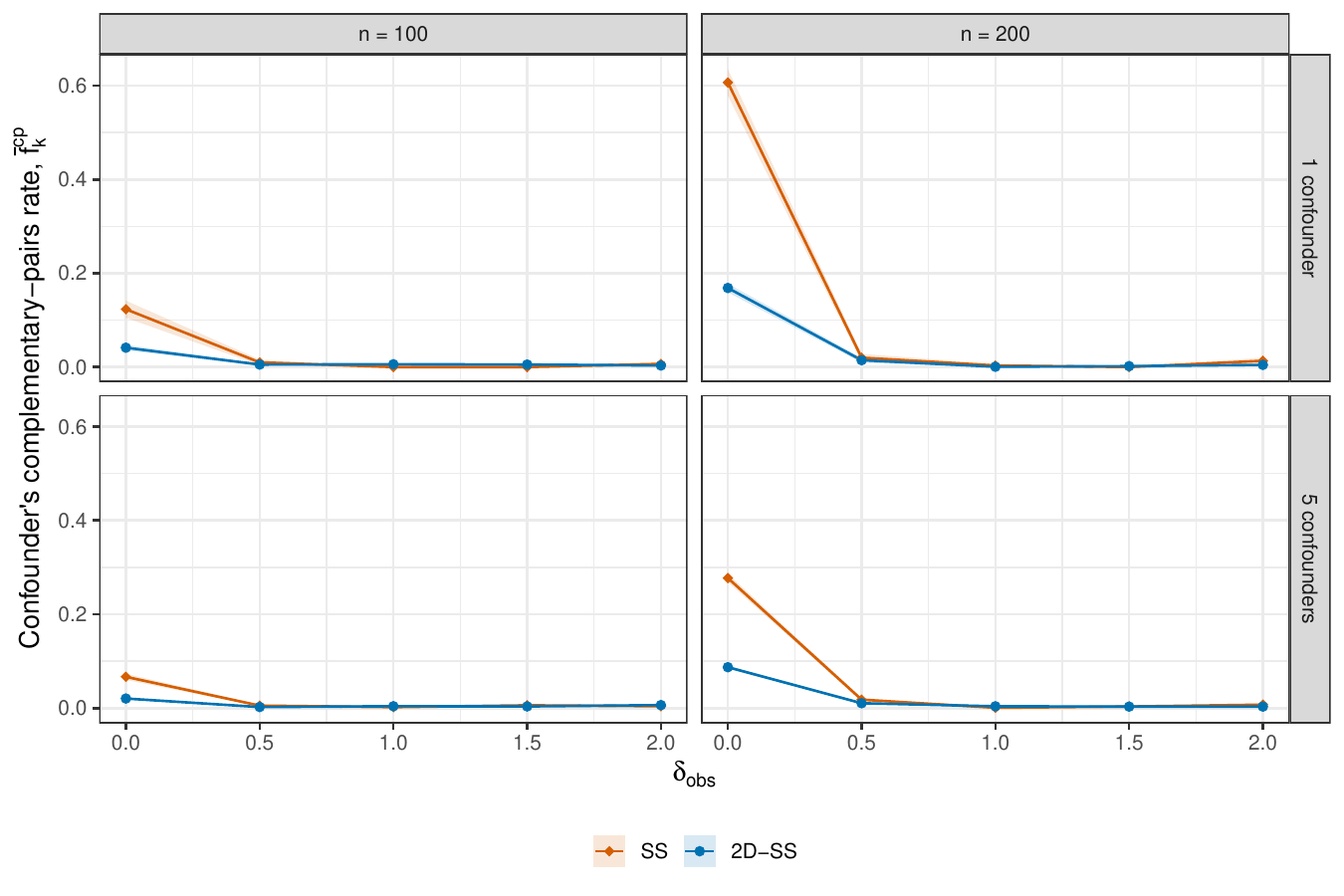}
\caption{The confounder's complementary-pairs rate $\bar f^{\,\mathrm{cp}}_k$
against $\delta_{\mathrm{obs}}$ ($q=5$).}
\label{fig:pfer_confounder}
\end{figure}

\begin{figure}[H]
\centering
\includegraphics[width=0.9\textwidth]{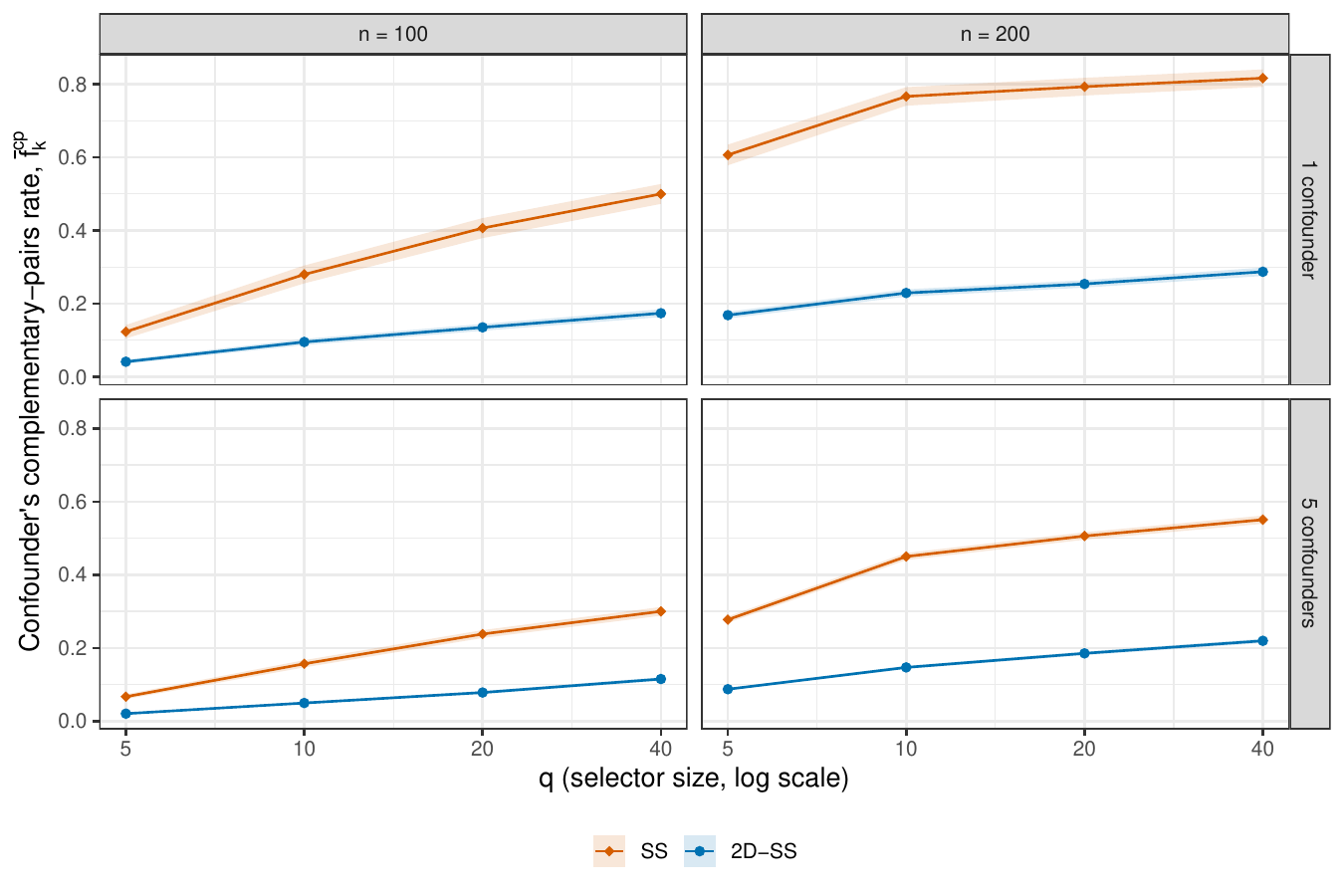}
\caption{The confounder's rate against the selector size $q$ at
$\delta_{\mathrm{obs}}=0$.}
\label{fig:pfer_confounder_q}
\end{figure}

Figure~\ref{fig:pfer_all} shows that 2D-SS satisfies the bound in all $240$ settings, with the
largest ratio of empirical PFER to bound equal to $0.40$, whereas Stability Selection does not:
its empirical PFER exceeds the bound in $21$ of $240$ settings ($16$ by more than two standard
errors), $19$ of them at $\delta_{\mathrm{obs}}=0$. At $n=200$ and $q=5$, $\tau=0.6$ it is
$0.37$ ($K=1$) and $0.35$ ($K=5$) against a bound of $0.126$, and larger ratios occur at
larger $\tau$, where the bound shrinks but the empirical value does not
(Figure~\ref{fig:pfer_bound} for $\tau=0.6${, and Figure~\ref{fig:pfer_all} for the other values}). The cause is visible in
Figure~\ref{fig:pfer_confounder}: at $n=200$, $\delta_{\mathrm{obs}}=0$, the
confounder's selection probability under Stability Selection is $0.61$, and $0.61^2=0.37$
accounts for essentially all of the PFER. The confounder is far from exchangeable with the
other irrelevant variables, whose selection probability is of order $q/(p-s)\approx0.005$,
and exchangeability is the hypothesis behind the second form of the bound in
Theorem~\ref{thm:errctrl}. 2D-SS lowers the confounder's probability to $0.17$, so that its contribution
is $0.17^2\approx0.03$ (a heuristic, not a measured PFER, and by Jensen's inequality a lower bound for the weighted quantity $\sum_\delta w_\delta p_k(\delta)^2$ of Theorem~\ref{thm:errctrl}) and the bound holds despite the violation. For
$\delta_{\mathrm{obs}}\ge0.5$ both methods are below the bound in all but two
settings ($\tau=0.9$, $q=5$, where Stability Selection exceeds it by less than $0.02$,
within Monte Carlo error).

Figures~\ref{fig:pfer_confounder} and~\ref{fig:pfer_confounder_q} give the confounder
rates. At $\delta_{\mathrm{obs}}=0$ and $q=5$ they are $0.12$ against $0.04$ ($n=100$) and
$0.61$ against $0.17$ ($n=200$) for $K=1$, and $0.07$ against $0.02$ and $0.28$ against
$0.09$ per confounder for $K=5$, a reduction by a factor of {$3.0$--$3.6$ at $q=5$}, which persists as $q$
grows ({factor $2.5$--$2.9$ across $n$ and $K$ at $q=40$ and $2.5$--$3.3$ over $q\in\{10,20,40\}$; at $q=40$, $n=200$, $K=1$: $0.82$ against $0.29$}). For $\delta_{\mathrm{obs}}\ge0.5$ both
rates fall to at most $0.02$ {at $q=5$ (at most $0.13$ over all $q$)} and the gap disappears. This agrees with Study~2: once
$\delta_{\mathrm{obs}}>0$, Stability Selection's own effective noise level is already
$\nu=\delta_{\mathrm{obs}}$ rather than $0$, so real measurement error does part of 2D-SS's
job for it. The advantage of 2D-SS over Stability Selection, in false-positive control as
in confounder exclusion, is accordingly concentrated at $\delta_{\mathrm{obs}}=0$, where
Corollary~\ref{cor:benefit} says it should be.
}

{
\subsection*{Study 2: additional results}
\paragraph{Selection metrics.} Figure~\ref{fig:selection_metrics} shows the trade. Stability Selection's false
positives are mostly the confounder, and 2D-SS removes them ($0.005$--$0.075$ false
positives for $K=1$, against $0.20$--$0.83$). At $n=400$, 2D-SS recovers $S$ exactly in
$64\%$ of training sets for $K=1$ against $24\%$ for Stability Selection, and in
$10.5\%$ against $0.5\%$ for $K=5$. At $n\le200$, 2D-SS trades power on the weak
coefficients for confounder exclusion: it recovers fewer true variables, and its
exact-recovery rate at $n=200$ is lower ($0.01$ against $0.18$, $K=1$). Jitter-only
attains a similar inclusion rate with more true positives at $n=400$ ($4.8$ and $4.5$ of the $5$ true variables for $K=1,5$; exact
recovery $0.57$ and $0.37$) since it uses the full sample; it has, however, no counterpart of
Theorem~\ref{thm:errctrl}, whose complementary-pairs argument requires the subsampling,
so {the finite-sample error bound applies to {subsampled procedures (Stability Selection and 2D-SS, with a fixed threshold $\tau$ or a size-matched selector, not the largest-gap rule), not to Jitter-only}}.

\paragraph{Consequences for prediction under future noise.} Each method's selected
support defines a fixed prediction model: Stability Selection, Jitter-only, 2D-SS and the oracle refit it by ordinary least squares on the training data, whereas Lasso, Ridge and ENet use their shrunken fits, evaluated on fresh data at future noise levels
$\delta_{\mathrm{test}}\in\{0,0.5,\dots,3\}$ unknown at training time ($n=100$,
$n_{\mathrm{rep}}=200$ training sets, $n_{\mathrm{test}}=500$ observations each). The
oracle refits on the true $S$.

Splitting Stability Selection by its own outcome (Figure~\ref{fig:surprise_curve_cond})
isolates the effect of the confounder. Fits that excluded it track 2D-SS closely at every
$\delta_{\mathrm{test}}\ge1$ ($82.1$ against $81.5$ at $2$, $K=1$), while fits that included
it are worse by about $10\%$ at $\delta_{\mathrm{test}}=2$ ($90.7$ against $82.1$) and {$13\%$ at $3$ ($204.7$ against $181.3$); for $K=5$ the figures are $8\%$ and $13.5\%$}. Including the confounder helps {clearly} only on clean
test data ($1.23$ against $1.91$ at $\delta_{\mathrm{test}}=0${; at $\delta_{\mathrm{test}}=0.5$ it is small ($0.18$ for $K=1$, $0.13$ for $K=5$, about two standard errors or less)}): its cost appears exactly when
the test noise differs from the training noise. Averaged over all training sets
(Figure~\ref{fig:surprise_curve}) the effect is diluted, since Stability Selection excludes
the confounder in most of them at $n=100$: 2D-SS is worse than Stability Selection at
$\delta_{\mathrm{test}}=0$ ($2.30$ against $1.79$), nearly equal at $1$ (within $1\%$), and modestly
better at larger noise ($81.5$ against $83.7$ at $2$; $179.9$ against $185.6$ at $3$).
{Lasso and ENet have a lower MSE than 2D-SS at every $\delta_{\mathrm{test}}$ and than the oracle refit from $\delta_{\mathrm{test}}=0.5$; Ridge does so from $1.5$. No unshrunken refit on any support, the oracle included, achieves this.} The benefit of 2D-SS is thus the
exclusion of the confounder, not a lower average prediction error.

The figures below belong to Study~2 (Section~\ref{s3}).

\begin{figure}[H]
\centering
\includegraphics[width=0.85\textwidth]{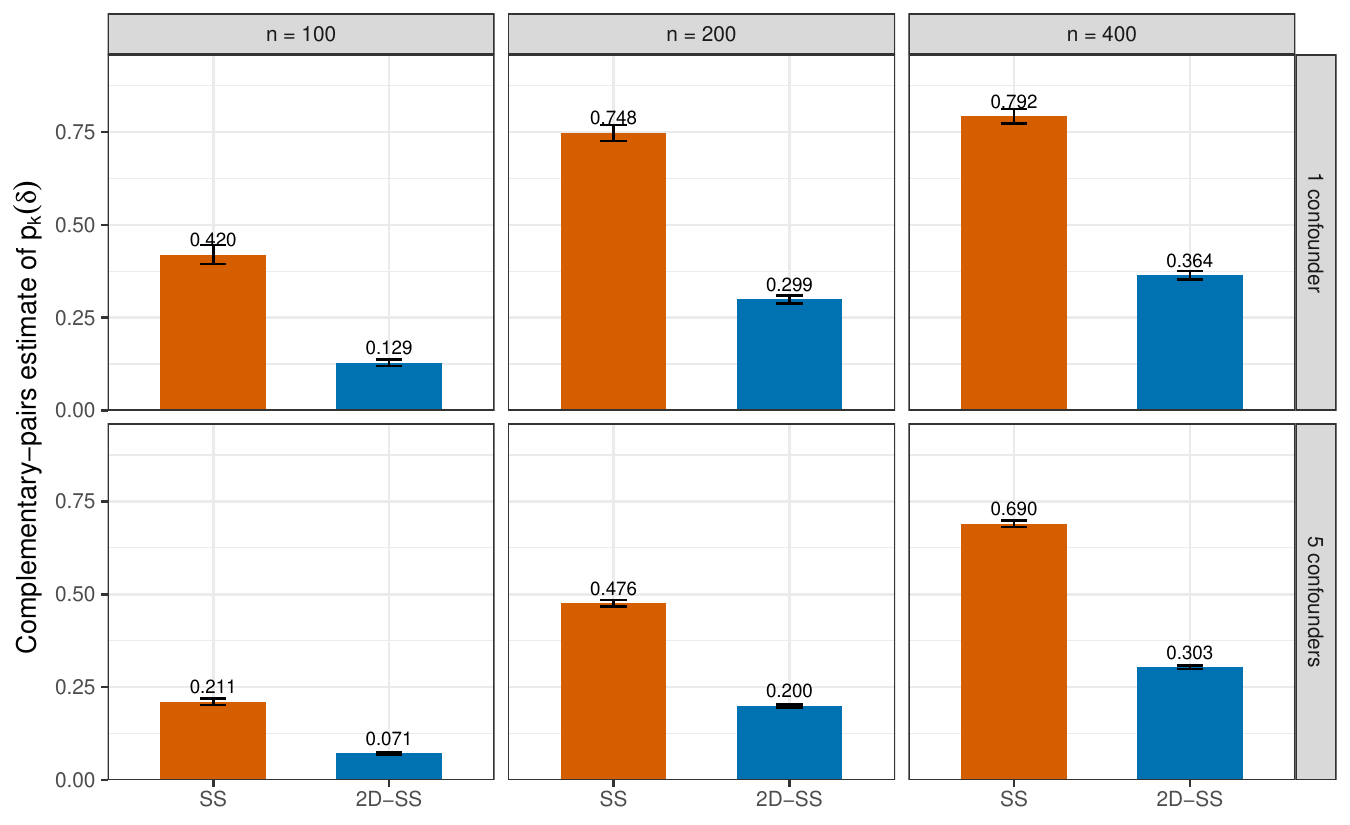}
\caption{Complementary-pairs estimate of $p_k(\delta)$ (Section~\ref{s2:errctrl}) for
Stability Selection and 2D-SS: the raw probability that a single half-sample fit
selects a confounder, before aggregation or thresholding.}
\label{fig:confounder_inclusion_cp}
\end{figure}

\begin{figure}[H]
\centering
\includegraphics[width=0.85\textwidth]{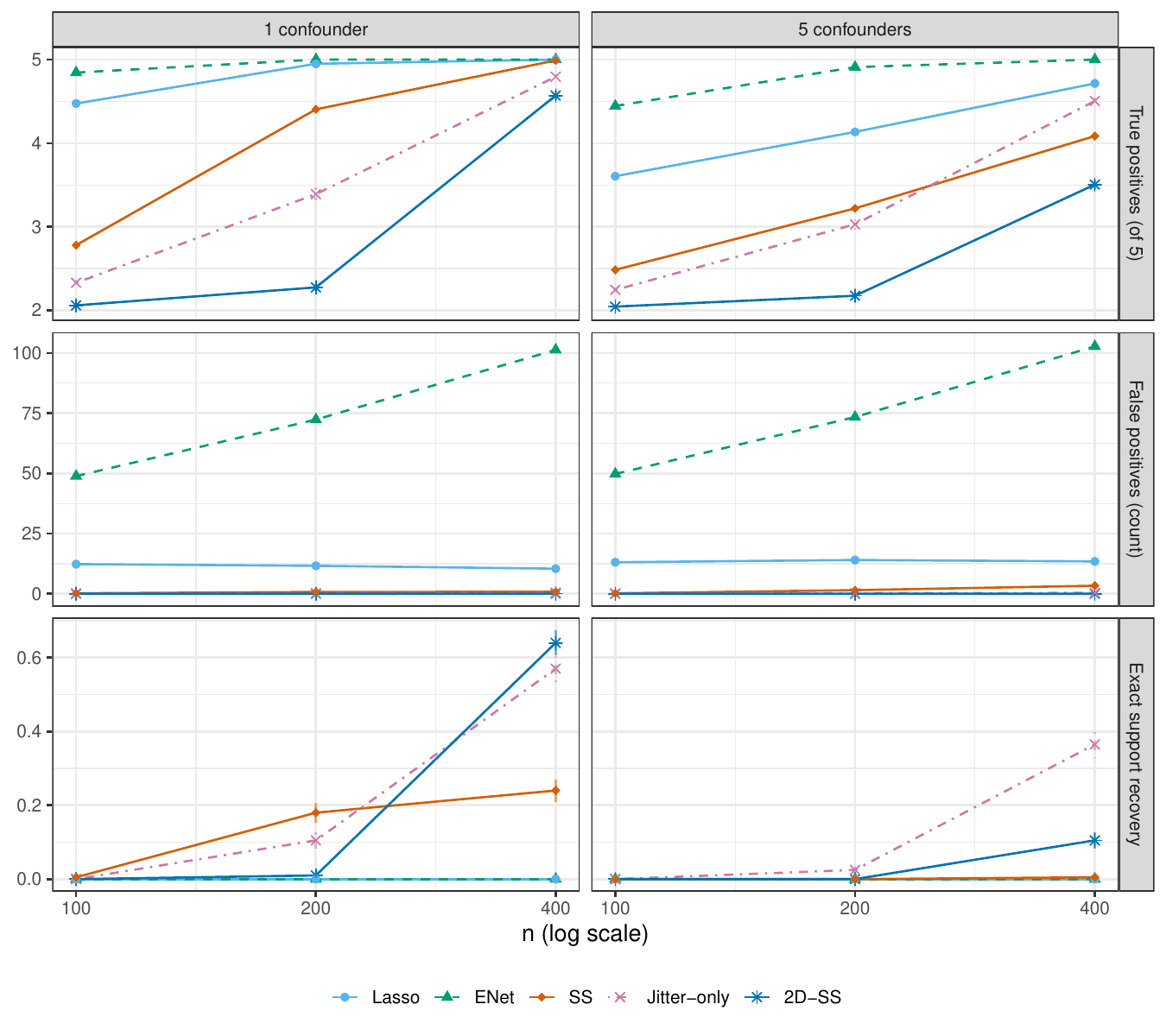}
\caption{True positives (of $5$), false positives and exact recovery of $S$ against
$n$ (Ridge omitted: it selects every variable).}
\label{fig:selection_metrics}
\end{figure}

\begin{figure}[H]
\centering
\includegraphics[width=\textwidth]{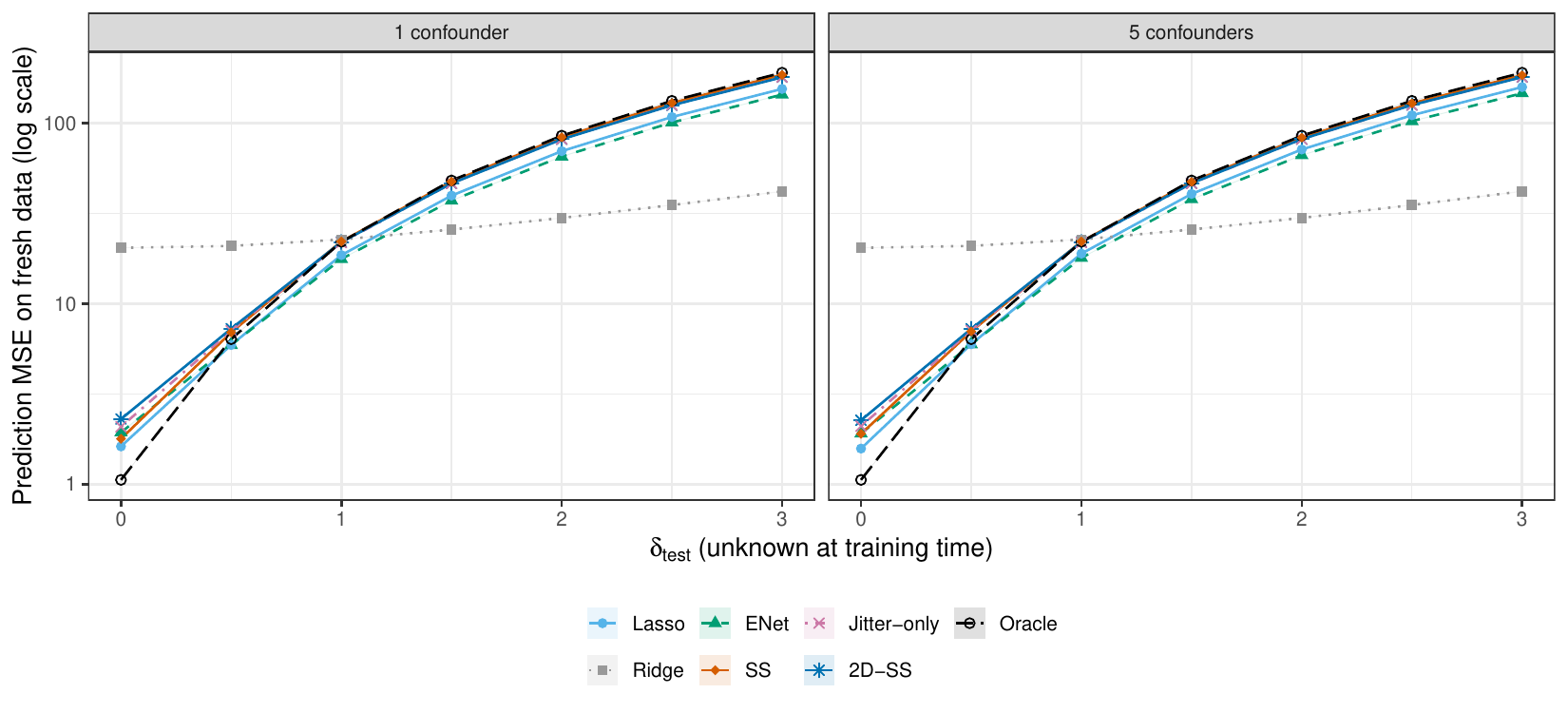}
\caption{Prediction MSE on fresh data against the future noise level
$\delta_{\mathrm{test}}$ (log scale), including the oracle refit on $S$.}
\label{fig:surprise_curve}
\end{figure}

\begin{figure}[H]
\centering
\includegraphics[width=\textwidth]{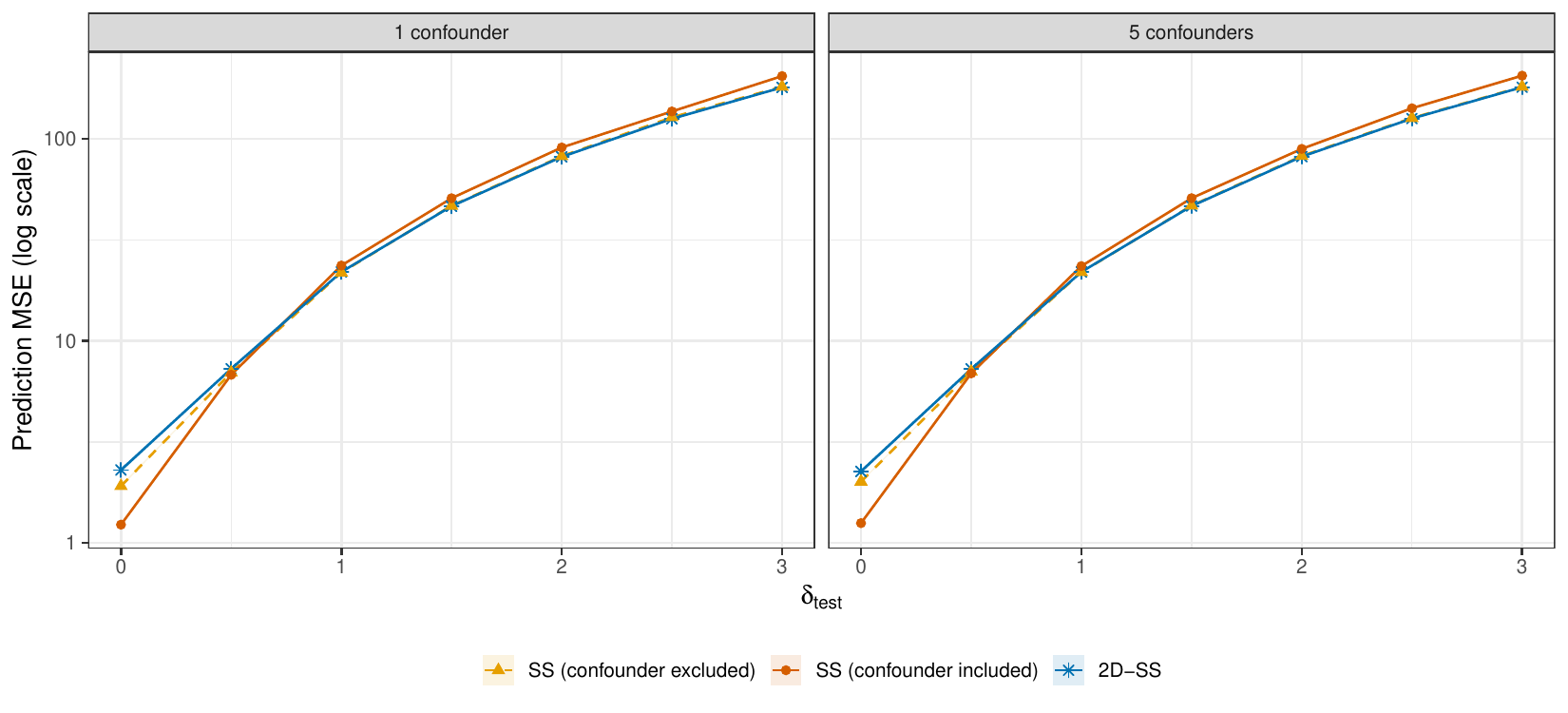}
\caption{As Figure~\ref{fig:surprise_curve}, with Stability Selection split by whether its
training fit included a confounder{; 2D-SS is shown as a reference}.}
\label{fig:surprise_curve_cond}
\end{figure}

\subsection*{Study 4: implementation details and additional results}
\paragraph{Design.} The design is Example~\ref{ex:jy} embedded in $p=200$ variables (reduced from $p=1{,}000$ because the baseline solvers are much slower than \texttt{glmnet}), with $n=100$, $\delta_{\mathrm{obs}}\in\{0,0.5,1,1.5,2\}$ and $40$ replicates per setting, and three designs: $K=1$ and $K=5$ confounders, and $K=0$, which has no confounder. Figure~\ref{fig:baselines_f1} gives the $F_1$ score for each method.
\paragraph{How the baselines are run.} The corrected Lasso and CoCoLasso need the
error variance, which is supplied as $\sigma^2 I$ (corrected Lasso) or through $\tau=\sigma$
(CoCoLasso, with \texttt{scale.Z=FALSE} so that $\tau$ is on the scale of the data).
They are run with the true $\sigma=\delta_{\mathrm{obs}}$ (\emph{oracle}), with $\sigma=0$
(\emph{naive}) and, for $K=1$, with half and double the true value. Their radius or
penalty is chosen by the packages' own cross-validation, at the minimum and at the
one-standard-error rule. GMUS takes no noise level: following
\citet{sorensen2018covariate}, $\lambda$ is set by cross-validation (here the \texttt{cv.glmnet} value, in place of their cross-validated Dantzig-selector value), and
$\delta$ is chosen on a grid by the elbow of the number of nonzeros against $\delta$
(we use the point furthest below the chord of the normalised curve, with the
one-standard-error $\lambda$ as a second variant). A coefficient is nonzero if its absolute
value exceeds $10^{-10}$. Lasso, Stability Selection, Jitter-only and 2D-SS use the fixed
$\lambda=\sqrt{\log p/n}$ (or $\sqrt{\log p/n_1}$ on half-samples) used throughout,
with the largest-gap rule for the last three{; Jitter-only and 2D-SS use $B=50$ fits at each of $\delta\in\{0,0.25,0.5,0.75,1\}$ with equal weights, and Stability Selection uses $250$ half-sample fits. The corrected Lasso uses $5$-fold cross-validation over $20$ radii, and GMUS a grid $\delta\in\{0,0.05,\dots,0.5\}$}. Because that tuning is not the same as
the baselines', we also report a \emph{ceiling} for each baseline: the point on its
tuning path with the highest $F_1$ in that replicate. The ceiling uses the true support and is
an optimistic bound for any tuning rule, not a method that could be run.

Lasso, Stability Selection, Jitter-only and 2D-SS centre and scale each noisy (sub)sample before fitting (as in Section~\ref{s2}), whereas the corrected Lasso and CoCoLasso receive the design on its original scale.

\paragraph{Ceilings.} The ceilings of all three baselines are well above 2D-SS at low noise (about $0.8$ against $0.6$
at $\delta_{\mathrm{obs}}=0$, $K=1$), so a baseline whose tuning is chosen with knowledge of the truth would beat 2D-SS there. The
ceilings fall towards 2D-SS as noise grows and are within $0.02$--$0.10$ of it at $\delta_{\mathrm{obs}}=2$
(Figure~\ref{fig:baselines_tuning}). Exact support recovery is close to zero for every method at $n=100${ when $K\ge1$ (at most $0.075$, for the CoCoLasso ceiling); at $K=0$ and $\delta_{\mathrm{obs}}=0$ it is $0.48$ for Stability Selection, $0.30$ for GMUS and $0.50$--$0.63$ for the ceilings, against $0.025$ for 2D-SS}.

\paragraph{Role of the selection rule.} Study~5 shows that the largest-gap rule can stop at the
strong signals, and the signal here ($\beta=(-4,-2,0.5,0.6,0.7)^\top$) has that structure. To see how much of the
low-noise deficit of 2D-SS is due to the rule, we repeated the design with $40$ further replicates, running only
Stability Selection and 2D-SS, and replaced the largest-gap rule by the $s$ most frequent variables, which
uses the true $s$ and so should be compared with the ceilings, not with the package-tuned baselines. At
$\delta_{\mathrm{obs}}=0$ this raises the $F_1$ of 2D-SS from $0.59$ to $0.78$ at $K=1$ (paired difference $0.18$,
standard error $0.02$), from $0.63$ to $0.82$ at $K=0$ and from $0.59$ to $0.65$ at $K=5$
(Figure~\ref{fig:addon_f1}). The gap to Stability Selection under the same rule falls from $0.21$, $0.25$
and $0.05$ to $0.065$, $0.11$ and $0.015$ ($K=1,0,5$), so most of the gap at $K=1$ and $K=5$ is the rule,
though 2D-SS stays below Stability Selection at $K=1$ and $K=0$. With $s$ known, 2D-SS is comparable to
GMUS ($0.78$, $0.82$, $0.65$ against $0.72$, $0.83$, $0.65$) and below the best ceilings ($0.83$, $0.96$,
$0.73$). The inclusion rate (Figure~\ref{fig:addon_conf}) is where it differs: with $s$ known, 2D-SS
selects a confounder in $10\%$ ($K=1$) and $32.5\%$ ($K=5$) of replicates, against $47.5\%$ and $87.5\%$ for
Stability Selection (paired differences $0.375$ and $0.55$, standard errors $0.08$) and $45\%$ and $55\%$ for GMUS.
For $\delta_{\mathrm{obs}}\ge0.5$ knowing $s$ brings no gain and is sometimes worse than the
largest-gap rule (for example $-0.05$ at $K=5$, $\delta_{\mathrm{obs}}=1$), since five variables must then
be selected from a profile in which few are distinguishable. A fixed threshold $\tau=0.3$ comes within about $0.05$ of the known-$s$ $F_1$ at zero error ($0.76$, $0.73$, $0.61$ against $0.82$, $0.78$, $0.65$ for $K=0,1,5$) but does not adapt to noise ({for 2D-SS, }$F_1$ between $0.15$ and $0.40$ and FDR at least
$0.71$ for $\delta_{\mathrm{obs}}\ge0.5$), and $\tau=0.6$ changes the $F_1$ of 2D-SS by at most $0.08$ relative to the largest-gap rule.

\paragraph{Paired comparisons.} Differences below are paired by replicate and are called ties when within two standard errors
(typically about $0.04$).

\paragraph{Against Stability Selection.} 2D-SS does not improve on Stability Selection in
$F_1$. At $\delta_{\mathrm{obs}}=0$ it is lower by $0.27$, $0.16$ and {$0.06$} for
$K=0,1,5$ (for example $0.607$ against $0.769$ at $K=1$), the power cost of the added noise on
a small sample. It is still slightly lower at $\delta_{\mathrm{obs}}\in\{0.5,1\}$
(by $0.03$--$0.05${, except a tie at $K=0$, $\delta_{\mathrm{obs}}=1$}) and tied for $\delta_{\mathrm{obs}}\ge1.5$. Its false discovery
proportion is below $0.06$ for $\delta_{\mathrm{obs}}\le1$ and rises to $0.24$--$0.33$ at
$\delta_{\mathrm{obs}}=2$, as for Stability Selection. What 2D-SS changes is the confounder
(Figure~\ref{fig:baselines_confounder}): at $\delta_{\mathrm{obs}}=0$ it includes none in any of the $40$
replicates, whereas Stability Selection includes one in $37.5\%$ of replicates ($K=1$) and
$45\%$ ($K=5$), the Lasso in $72.5\%$ and $100\%$, and the corrected Lasso
and CoCoLasso at the cross-validation minimum in {$62.5$--$100\%$ (the one-standard-error CoCoLasso in at most $2.5\%$)}. Stability Selection also {almost} excludes the confounder for
$\delta_{\mathrm{obs}} \geq 0.5$ in this design{ (at most $5\%$)}, so the advantage is concentrated at zero observed
error, where Study~3 shows that the error bound holds for size-matched 2D-SS but not for Stability Selection.

\paragraph{Against the measurement-error baselines.} GMUS with the elbow rule is the strongest
baseline that can be run. It is ahead of 2D-SS at $\delta_{\mathrm{obs}}=0$ ($0.724$ against
$0.607$ at $K=1$), {within about two standard errors} at $\delta_{\mathrm{obs}}=0.5$--$1$ (just beyond two at $\delta_{\mathrm{obs}}=0.5$ for $K=0$ and $5$), and tied or nominally
behind 2D-SS at $\delta_{\mathrm{obs}}\ge1.5$ (by $0.01$--$0.03$). 2D-SS is clearly ahead of
the corrected Lasso and CoCoLasso under the cross-validation minimum {at $\delta_{\mathrm{obs}}\le0.5$ and ahead or tied at larger $\delta_{\mathrm{obs}}$}, even with the true
noise level supplied: their false discovery proportions are $0.42$--$0.83$ for $\delta_{\mathrm{obs}}\le1$
and their $F_1$ is $0.25$--$0.45$ at most settings. Their one-standard-error variants are closer: the corrected Lasso is tied with 2D-SS
for $\delta_{\mathrm{obs}}\ge1$ (and for $K=0$ throughout) but behind it by {$0.07$--$0.11$} at low noise when $K\ge1${ (a tie, $0.04$, at $K=1$, $\delta_{\mathrm{obs}}=0.5$)}, and CoCoLasso is behind it by $0.04$--$0.19$ almost everywhere. When the supplied noise level is wrong
(Figure~\ref{fig:baselines_sensitivity}), the corrected Lasso and CoCoLasso {are mostly worse still (not in every setting)} and 2D-SS is ahead
{by up to $0.24$ at $\delta_{\mathrm{obs}}>0$ for $K=1$ (at $\delta_{\mathrm{obs}}=0$ every supplied level equals zero, so the gap there, $0.27$--$0.32$, is the oracle gap), with ties for double-level CoCoLasso at $\delta_{\mathrm{obs}}=0.5$--$1.5$ and double-level corrected Lasso at $2$}; 2D-SS needs no such input. The
Lasso and Jitter-only, which have no subsampling, have false discovery proportions of {$0.73$--$0.95$} at
$\delta_{\mathrm{obs}}\ge1$ and are {$0.20$--$0.36$} below 2D-SS from $\delta_{\mathrm{obs}}=0.5$ (Lasso) or $\delta_{\mathrm{obs}}=1$ (Jitter-only).

\subsection*{Study 4: additional figures}
\begin{figure}[H]
\centering
\includegraphics[width=0.9\textwidth]{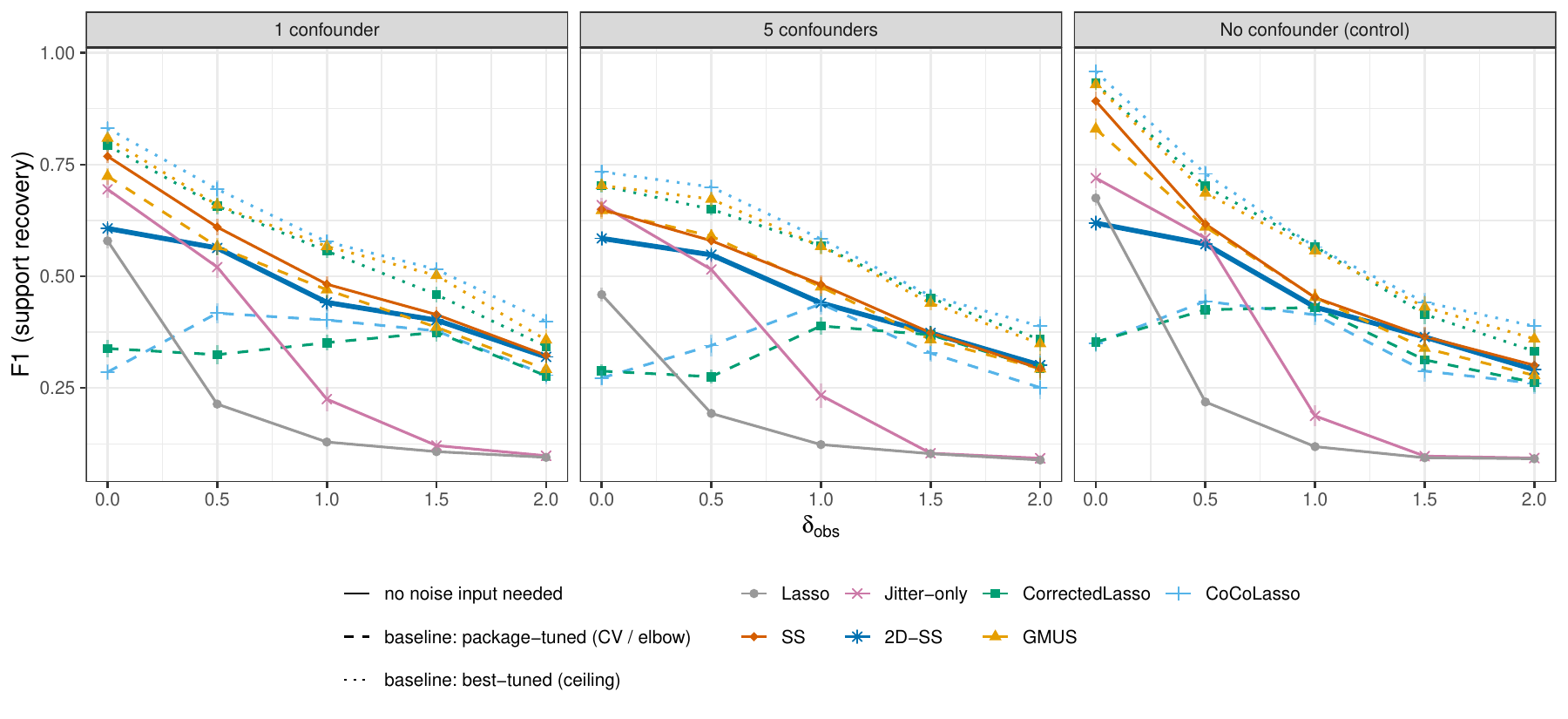}
\caption{$F_1$ score against $\delta_{\mathrm{obs}}$ for each design, for Lasso, Stability
Selection, Jitter-only, 2D-SS and the baselines. Dashed: package tuning (cross-validation or elbow); dotted: best-tuned ceiling. The third panel has no confounder.}
\label{fig:baselines_f1}
\end{figure}
\begin{figure}[H]
\centering
\includegraphics[width=0.9\textwidth]{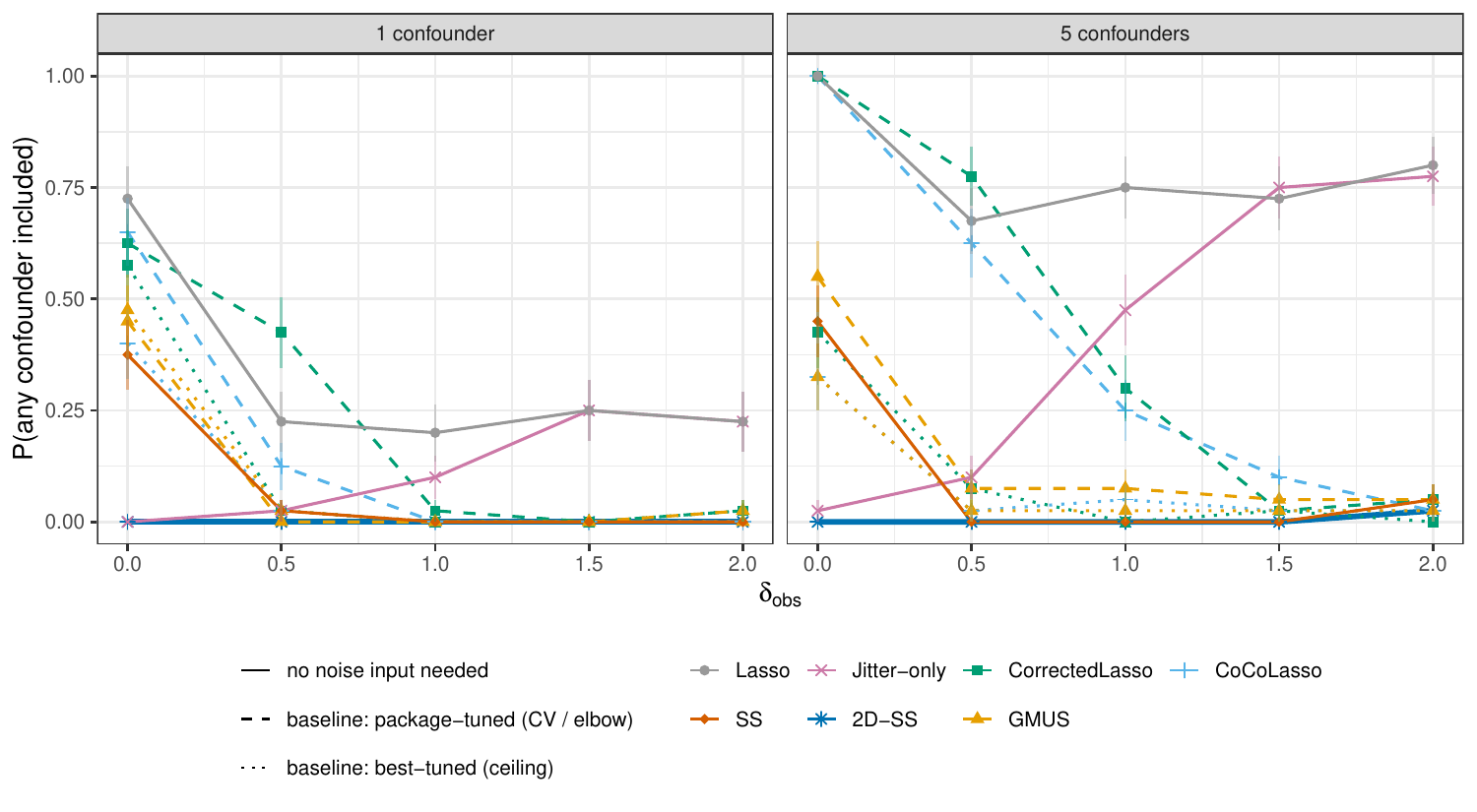}
\caption{Rate at which a confounder is selected (at least one, for $K=5$) against $\delta_{\mathrm{obs}}${; bars are $\pm$ one standard error}.}
\label{fig:baselines_confounder}
\end{figure}

\begin{figure}[H]
\centering
\includegraphics[width=0.9\textwidth]{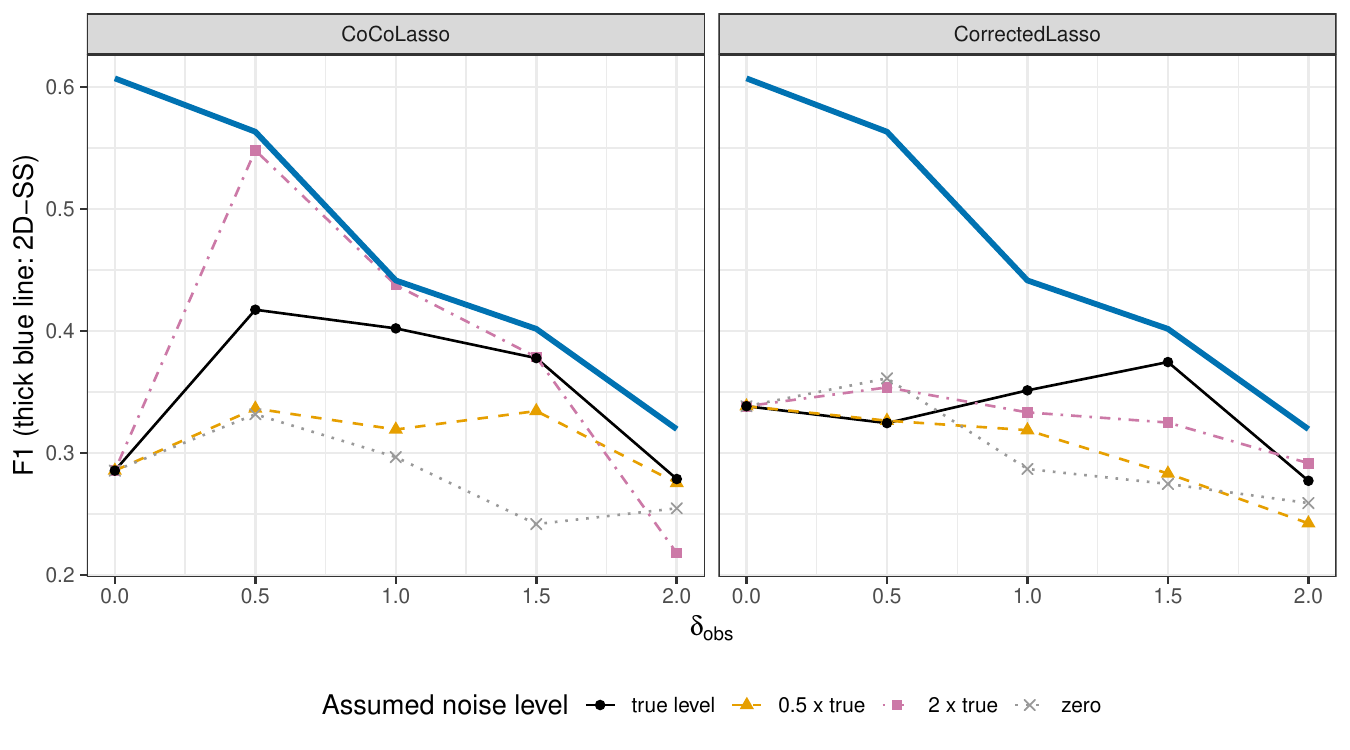}
\caption{$F_1$ of the corrected Lasso and CoCoLasso when the supplied noise level is the true value,
zero, half or double ($K=1$){; thick blue line: 2D-SS}.}
\label{fig:baselines_sensitivity}
\end{figure}

\begin{figure}[H]
\centering
\includegraphics[width=0.9\textwidth]{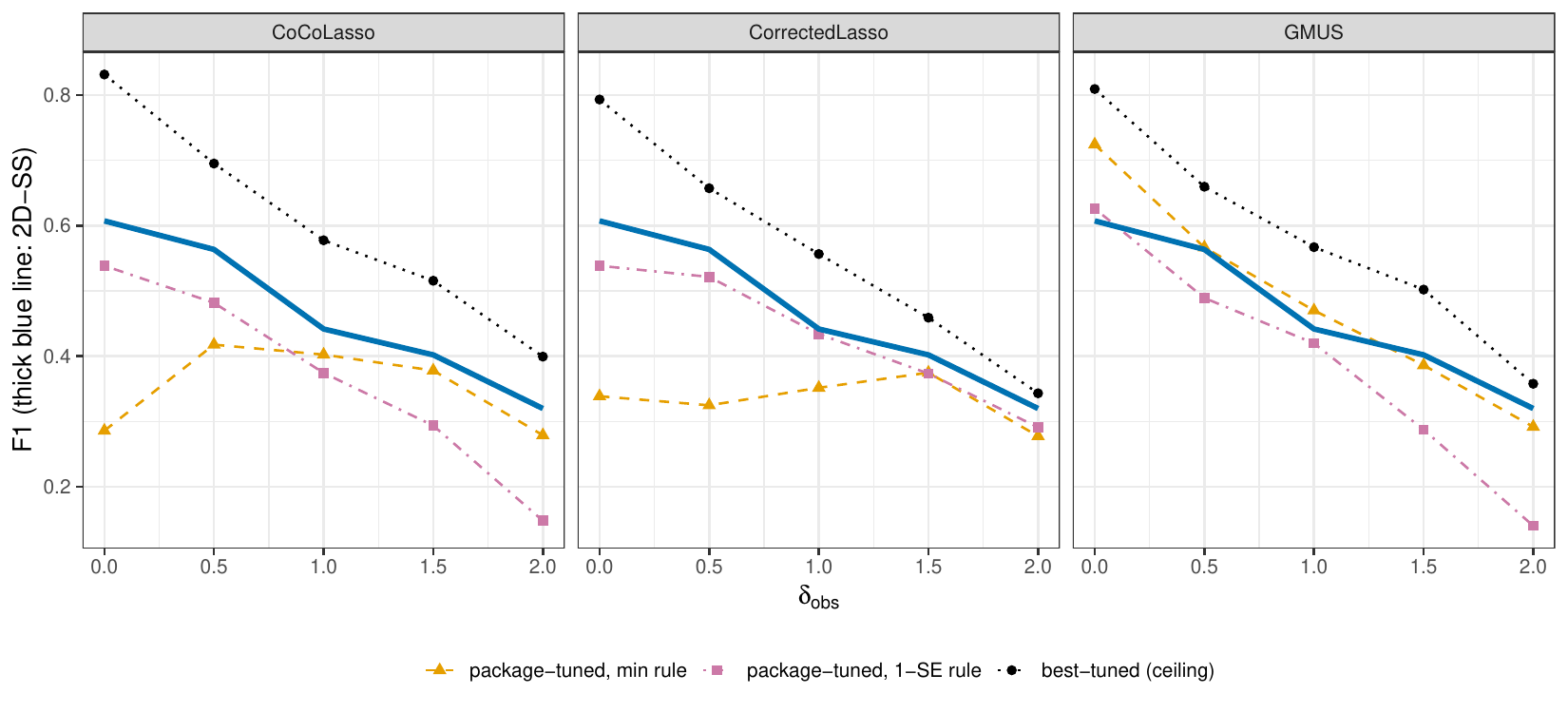}
\caption{$F_1$ of each baseline under its cross-validation minimum, its one-standard-error rule and its ceiling{; thick blue line: 2D-SS}.}
\label{fig:baselines_tuning}
\end{figure}

}

{
Figures~\ref{fig:addon_f1} and~\ref{fig:addon_conf} belong to the known-$s$ and fixed-threshold comparison of Study~4 (paragraph ``Role of the selection rule'').

\begin{figure}[H]
\centering
\includegraphics[width=0.95\textwidth]{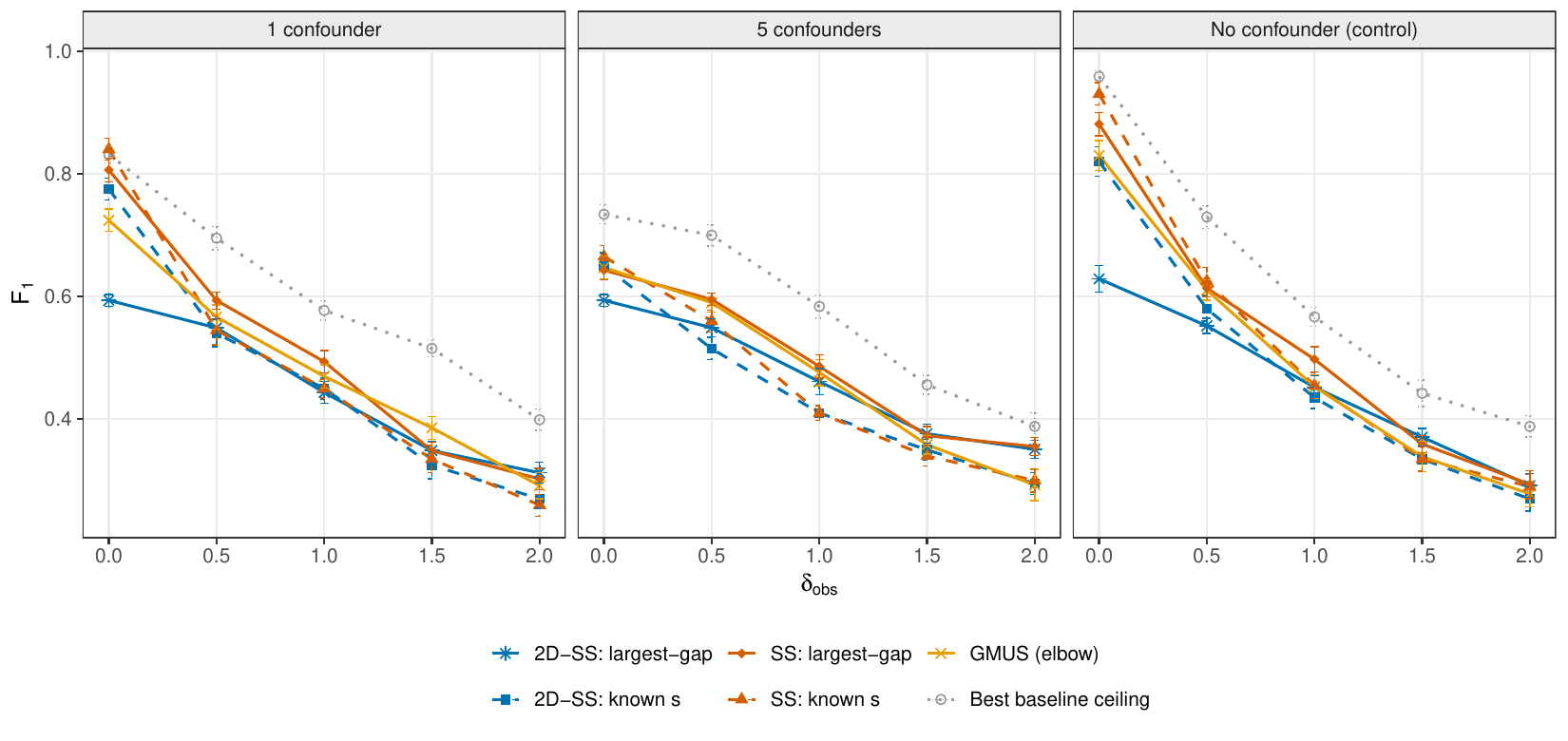}
\caption{$F_1$ against $\delta_{\mathrm{obs}}$ for Stability Selection (SS) and 2D-SS with the largest-gap rule
(solid) and with the true $s$ known (dashed), with GMUS (elbow) and the best baseline ceiling from the main comparison{ (dotted); bars are $\pm$ one standard error}.}
\label{fig:addon_f1}
\end{figure}

\begin{figure}[H]
\centering
\includegraphics[width=0.75\textwidth]{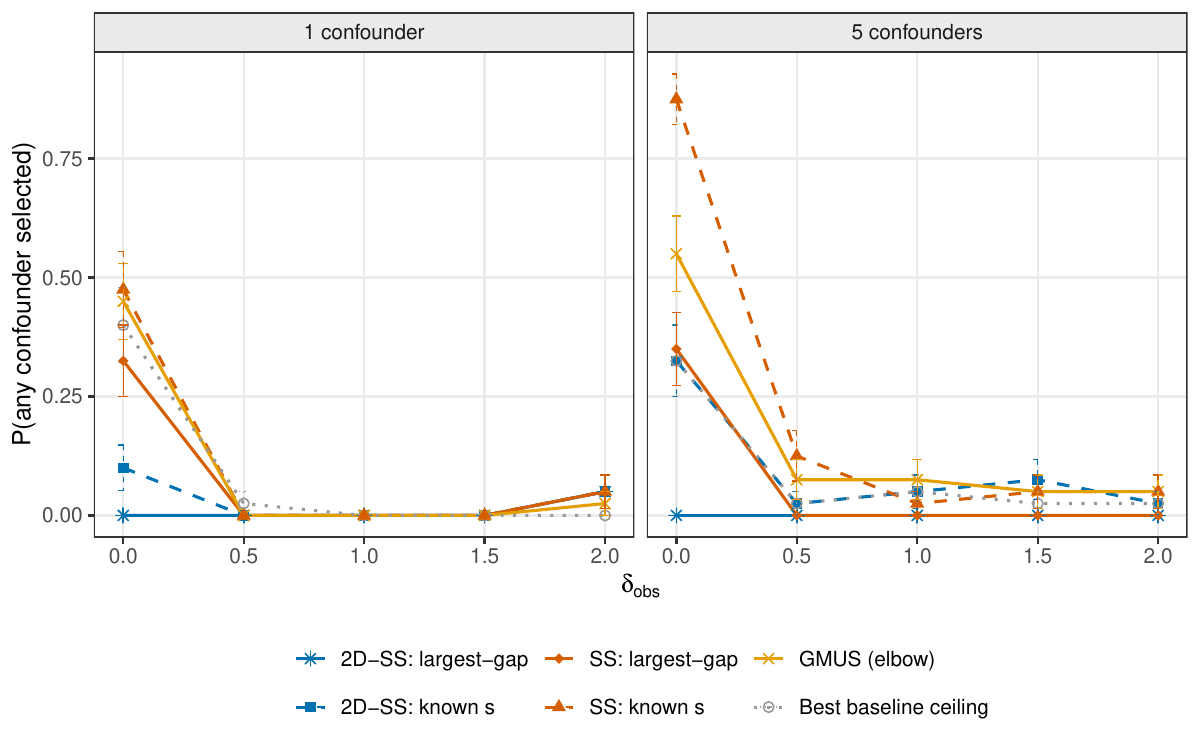}
\caption{Rate at which at least one confounder is selected, for the same selection rules and methods as Figure~\ref{fig:addon_f1}{ (bars: $\pm$ one standard error)}.}
\label{fig:addon_conf}
\end{figure}

\subsection*{Study 5: When the Largest-Gap Rule Fails}\label{study:gap_failure}
Theorem~\ref{thm:gap} requires the boundary gap to exceed every within-group gap; the
discussion following it notes that this can fail when $S$ mixes strong and weak signals,
in which case the largest gap falls \emph{within} $S$ and the rule underestimates $s$.
We ask where this happens, how much it costs, and what should replace the rule. The
design has $s=5$, $\Sigma=I_p$, $n=100$, $p=500$ and
$\beta=(6,5,w,0.9w,0.8w)^\top$, with the weak-signal strength $w\in\{0.3,0.5,0.8,1.2,2\}$
($w=0.5$ is $(0.5,0.45,0.4)$). There is no confounder and no
additional measurement error, since the failure is a property of the thresholding step.
We aggregate with $\lambda=\sqrt{\log p/n_1}$ and $B=100$ half-sample fits {per level}, either at
$\delta=0$ (Stability Selection) or over the grid $\delta\in\{0,0.25,0.5,0.75,1\}$ with equal
weights (2D-SS), and apply five rules: the largest-gap rule, the $s$ most frequent variables
(known $s$), and fixed thresholds $\tau\in\{0.15,0.3,0.6\}$. Each setting has $100$ replicates.

\begin{figure}[H]
\centering
\includegraphics[width=\textwidth]{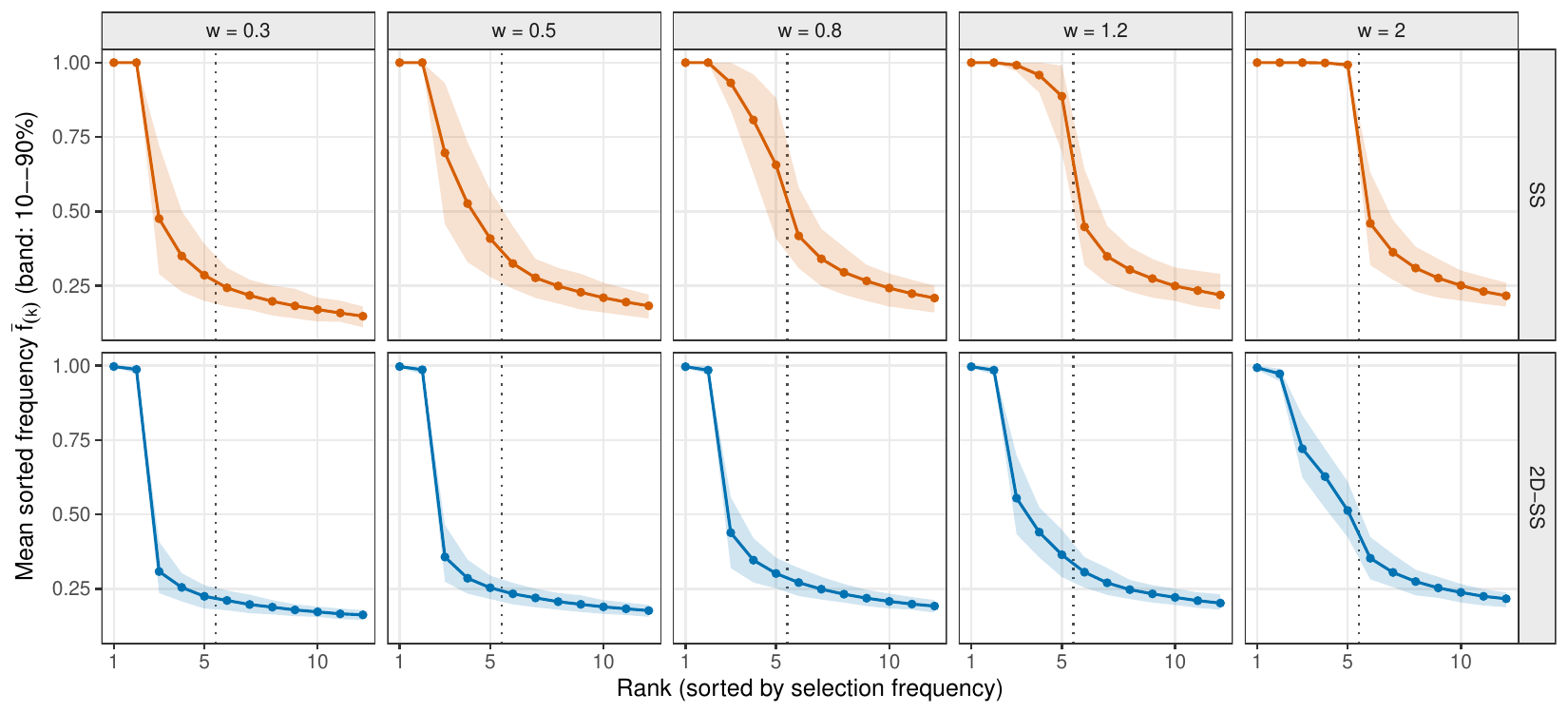}
\caption{Mean sorted selection frequency (band: $10$--$90\%$ over replicates) for Stability Selection
(top) and 2D-SS (bottom), by weak-signal strength $w$. The dotted line is the true boundary{, between ranks $5$ and $6$}.}
\label{fig:gap_failure_profile}
\end{figure}

\begin{figure}[H]
\centering
\includegraphics[width=0.9\textwidth]{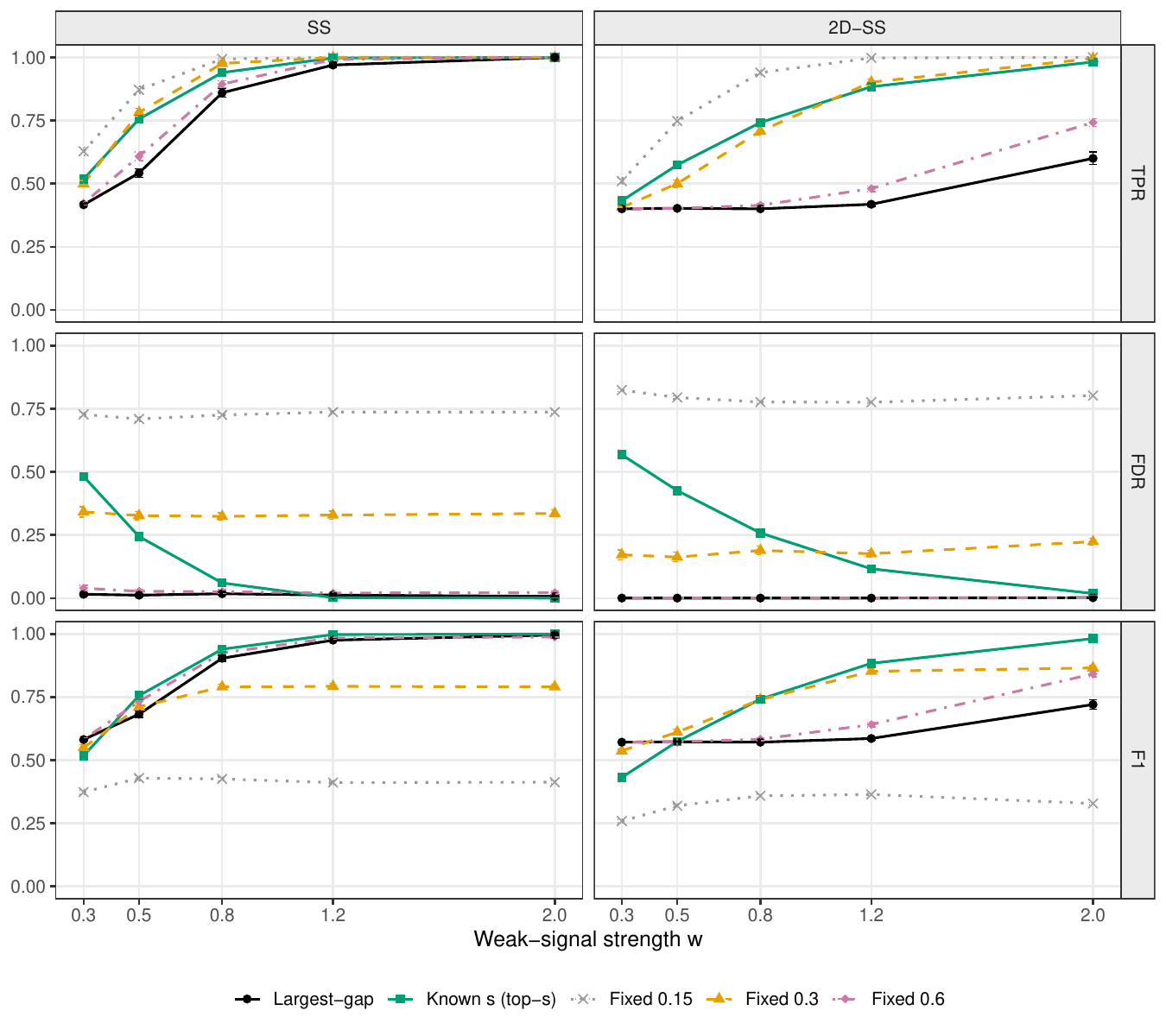}
\caption{TPR, FDR and $F_1$ of each rule against $w$, for Stability Selection (left) and 2D-SS (right).}
\label{fig:gap_failure_metrics}
\end{figure}

\begin{figure}[H]
\centering
\includegraphics[width=0.9\textwidth]{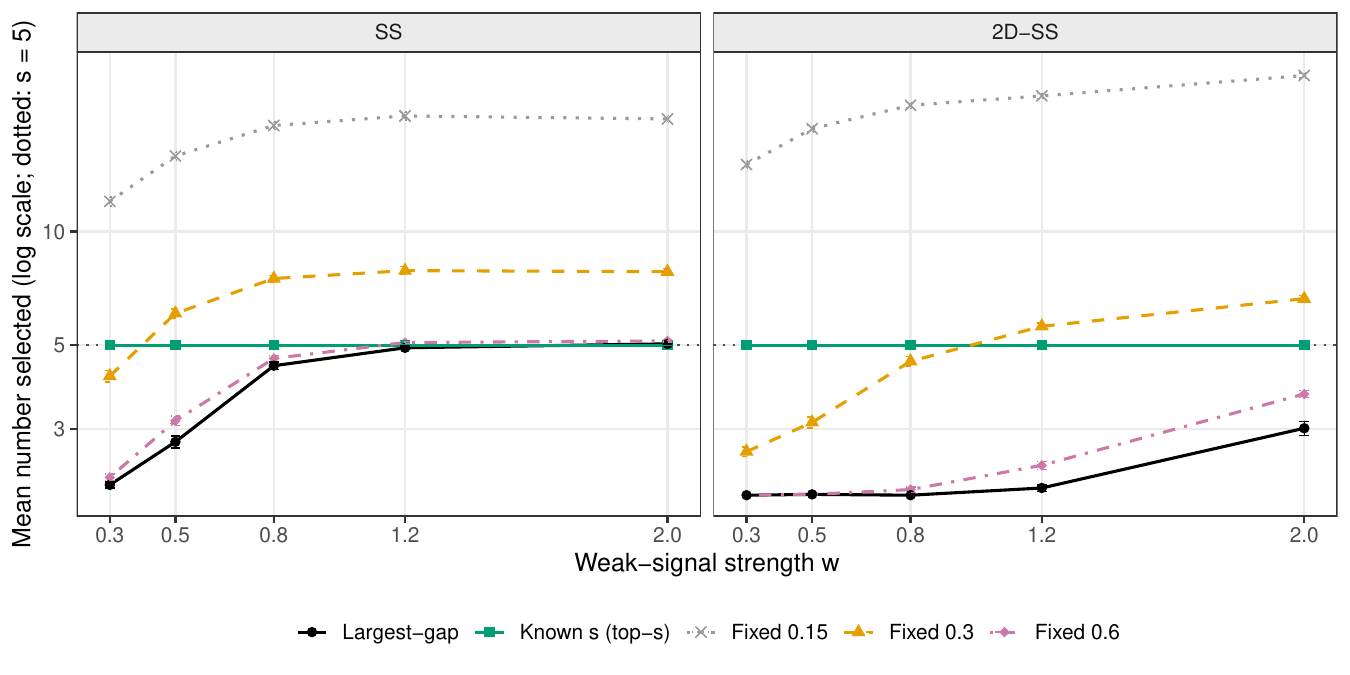}
\caption{Mean number of variables selected by each rule against $w$ (log scale; dotted: $s=5$).}
\label{fig:gap_failure_size}
\end{figure}

\paragraph{Stability Selection.} Figure~\ref{fig:gap_failure_profile} shows the mechanism. At $w=0.5$ the
weak signals have mean frequencies $0.70$, $0.53$ and $0.41$, so the largest gap ($1.00\to0.70$) is
between the strong and the weak signals and the boundary gap ($0.41\to0.32$) is much smaller. The rule
then selects $\hat s=2.77$ variables on average and $\hat s=s$ in only $5\%$ of replicates. As $w$
grows the weak signals move towards frequency $1$ and the failure disappears: $\hat s=4.41$,
$4.92$ and $5.04$ at $w=0.8$, $1.2$ and $2$, with $\hat s=s$ in $50\%$, $82\%$ and $96\%$ of replicates.
The rule is conservative when it fails (FDR at most $0.02$ throughout) and its $F_1$ is $0.58$, $0.68$,
$0.90$, $0.98$ and $1.00$. Knowing $s$ helps only where the rule fails, by $0.07$ in $F_1$ at $w=0.5$
and $0.04$ at $w=0.8$, and by at most $0.02$ beyond that; at $w=0.3$, where the weak signals are barely detectable, it is worse by $0.06$. A threshold of $\tau=0.6$ does as well as or
slightly better than the largest-gap rule for $w\le1.2$ (by $0.05$ at $w=0.5$, within $0.03$ otherwise).
A low threshold does not help: $\tau=0.15$ admits about $12$--$20$ variables and has $F_1$ between $0.37$
and $0.43$ (FDR about $0.73$), and $\tau=0.3$ has an FDR of about $0.33$.

\paragraph{2D-SS.} The failure is considerably more severe. Averaging over the grid lowers the
frequencies of the weak signals towards the noise floor (Figure~\ref{fig:gap_failure_profile}, bottom): at
$w=1.2$ the weak signals have frequencies $0.55$, $0.44$ and $0.36$ against $0.31$ for the next variable,
so the boundary gap is about $0.06$ while the gap after the strong signals is about $0.43$. The largest-gap
rule selects $\hat s\approx2.0$ for every $w\le1.2$ (TPR $0.40$--$0.42$, FDR $0$, $F_1$ $0.57$--$0.59$) and
$\hat s=3.0$ at $w=2$ ($F_1$ $0.72$, $\hat s=s$ in $21\%$ of replicates). Knowing $s$ gives $F_1$ of $0.74$,
$0.88$ and $0.98$ at $w=0.8$, $1.2$ and $2$, higher than the largest-gap rule by $0.17$, $0.30$ and $0.26$.
The comparison with fixed thresholds shows that the scale of the averaged frequencies matters:
$\tau=0.3$ reaches $F_1$ of $0.74$, $0.85$ and $0.87$ at those values of $w$ (FDR $0.18$--$0.22$), ahead of
the largest-gap rule by {$0.145$--$0.27$}, whereas $\tau=0.6$, the value compatible with
Theorem~\ref{thm:errctrl}, keeps the FDR at $0$ but recovers the weak signals appreciably only at $w=2$
($F_1$ $0.57$, $0.57$, $0.58$, $0.64$ and $0.84$ across the five values of $w$). The threshold
$\tau=0.15$ is again poor (FDR about $0.8$, $F_1\le0.36$).

\paragraph{Reading.} The failure mode of Theorem~\ref{thm:gap} is real, and for 2D-SS it
is not a corner case: with mixed signal strengths the rule stops at the strong signals, and a known
model size or a fixed threshold calibrated to the averaged frequencies does much better. There is a
trade-off in the choice of threshold, since $\tau>1/2$, which the error bound of
Theorem~\ref{thm:errctrl} requires, keeps false positives out but also misses weak signals, while $\tau$
near $0.3$ recovers them at an FDR of about $0.2$ and lies outside that bound. Stability Selection is
ahead of 2D-SS in $F_1$ under every rule here{ except fixed $\tau=0.3$ at $w\ge1.2$, and tied at $w=0.3$ under $\tau=0.3$ and $\tau=0.6$, and nearly so under the largest-gap rule} (for example $0.98$ against $0.59$ for the largest-gap rule at $w=1.2$);
this design has no confounder for the jitter to repair, so it shows the power cost of the noise axis, as in
Study~4 at zero observed error, where a TPR of $0.44$ with FDR $0$ for 2D-SS at $K=1$ corresponds to about
$2.2$ of the $5$ true variables and is consistent with the behaviour seen here; the known-$s$ comparison in Study~4 shows that the rule accounts for most of the deficit at $K=1$.
}

{
\subsection*{Rat data: repeated train--test splits}\label{study:rat_splits}
\paragraph{Data and preprocessing.} We discard the probes whose maximum expression over the $120$ rats does not exceed the $25$th percentile of the probe-wise maxima, and then those whose expression range is below $2$. This leaves $p=3{,}083$ probes. Screening and column scaling use all $120$ rats.
\paragraph{Procedures on the full data.} Stability Selection uses $B=100$ half-sample fits ($n_1=60$) at $\delta=0$, so its computing budget is a tenth of that of the other two procedures. Jitter-only and 2D-SS use $B=100$ fits at each of $10$ equally spaced levels $\delta\in[0.05,2.5]$ with equal weights. The largest-gap rule is used unless a threshold is stated. It gives the thresholds $\hat\tau=0.425$ for Stability Selection, $\hat\tau\approx0.41$ for Jitter-only and $\hat\tau=0.118$ for 2D-SS.
\paragraph{Selected probes.} On the full data, 2D-SS selects the probes \texttt{1376747\_at}, \texttt{1377791\_at}, \texttt{1389457\_at}, \texttt{1390539\_at} and \texttt{1392982\_at}.
\paragraph{Design.} The data and screening are as in Section~\ref{s3}, with $n=120$ and $p=3{,}083$ (screening and column scaling use all $120$ rats, without reference to the response, and are not repeated within splits). In each of $50$ random splits, $80$ rats form the training set and $40$ the test set. On the training set we compute $\lambda_{\mathrm{1se}}$ by ten-fold cross-validation of the Lasso and use it for all procedures; Stability Selection uses $B=50$ half-sample fits ($n_1=40$) at $\delta=0$, and Jitter-only (full sample) and 2D-SS use $B=50$ fits at each of the ten equally spaced levels $\delta\in[0.05,2.5]$ with equal weights. Each selected set is scored by the test mean squared error of an ordinary least-squares refit on the selected probes, divided by the test error of the training-mean variable (an empty set predicts the training mean). We report this relative error and the stability of the selected sets across splits, by the Nogueira index; {the mean pairwise Jaccard index (not tabulated) orders the four procedures in the same way in every row of the table}. Two comparisons are made: (a) each procedure with its own rule (largest-gap for Stability Selection, Jitter-only and 2D-SS; the Lasso at $\lambda_{\mathrm{1se}}$), and (b) at equal size, taking the $k\in\{3,5,10\}$ probes with the highest selection frequency (ties broken at random) for the three aggregated procedures and the first $k$ probes to enter the Lasso path.

\begin{center}
\small
\refstepcounter{table}\label{tab:rat_splits}
\begin{minipage}{0.95\textwidth}
\textbf{Table~\thetable.} Rat data, $50$ random splits. Entries are the mean relative test error (standard error) and the Nogueira stability index. ``Own rule'' is the largest-gap rule for the three aggregated procedures and $\lambda_{\mathrm{1se}}$ for the Lasso; the mean number of selected probes is $18.0$, $1.8$, $3.9$ and $2.2$ for the Lasso, Stability Selection, Jitter-only and 2D-SS.
\end{minipage}

\vspace{4pt}
\begin{tabular}{llllll}
\hline
 & & Lasso & Stability Selection & Jitter-only & 2D-SS \\
\hline
Own rule & error & $0.573\ (0.023)$ & $0.844\ (0.033)$ & $0.749\ (0.038)$ & $0.802\ (0.032)$ \\
 & stability & $0.158$ & $0.041$ & $0.089$ & $0.072$ \\
$k=3$ & error & $0.660\ (0.025)$ & $0.730\ (0.032)$ & $0.685\ (0.028)$ & $0.705\ (0.027)$ \\
 & stability & $0.105$ & $0.072$ & $0.104$ & $0.098$ \\
$k=5$ & error & $0.605\ (0.023)$ & $0.681\ (0.027)$ & $0.638\ (0.024)$ & $0.652\ (0.025)$ \\
 & stability & $0.137$ & $0.086$ & $0.132$ & $0.115$ \\
$k=10$ & error & $0.542\ (0.021)$ & $0.615\ (0.023)$ & $0.558\ (0.021)$ & $0.589\ (0.022)$ \\
 & stability & $0.187$ & $0.127$ & $0.165$ & $0.144$ \\
\hline
\end{tabular}
\end{center}

\begin{figure}[H]
\centering
\includegraphics[width=\textwidth]{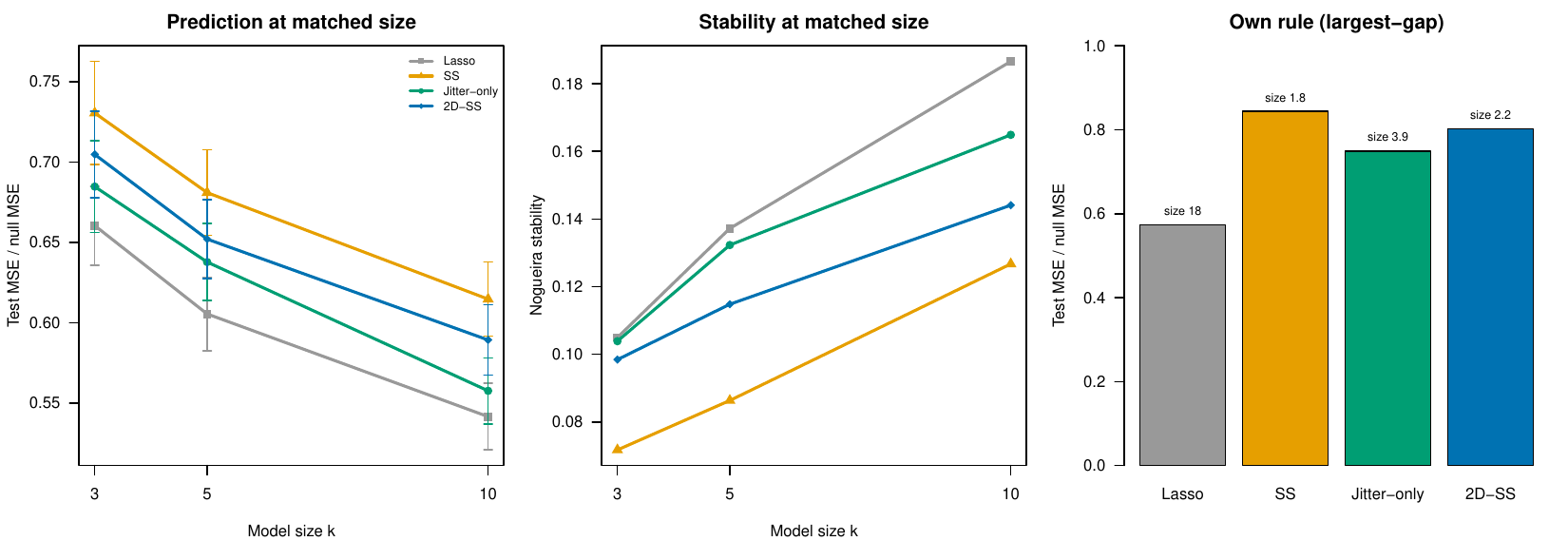}
\caption{Rat data, $50$ random train--test splits. Left: relative test error against the model size $k$ (mean $\pm$ one standard error). Middle: Nogueira stability index against $k$. Right: relative test error under each procedure's own rule, with the mean number of selected probes.}
\label{fig:rat_splits}
\end{figure}

\paragraph{Results.} Under their own rules (Table~\ref{tab:rat_splits}, top), Stability Selection and 2D-SS retain about two probes and predict clearly worse than the Lasso, which keeps $18$ (relative error $0.84$ and $0.80$ against $0.57$); the largest-gap rule stops early on these data, as in Study~5, and the comparison mixes model size with method. At equal size (Figure~\ref{fig:rat_splits}, left and middle), the Lasso path has the lowest error at every $k$, followed by Jitter-only, 2D-SS and Stability Selection. Paired over splits, the relative error of Stability Selection exceeds that of 2D-SS by $0.026$ ($\mathrm{se}\ 0.024$), $0.029$ ($0.018$) and $0.025$ ($0.012$) at $k=3$, $5$, $10$, and 2D-SS has the higher stability index at each $k$ ($0.098$, $0.115$, $0.144$ against $0.072$, $0.086$, $0.127$). Jitter-only is better than 2D-SS by $0.020$ ($0.016$), $0.014$ ($0.011$) and $0.032$ ($0.012$), and the Lasso path by $0.044$ ($0.023$), $0.047$ ($0.024$) and $0.048$ ($0.014$). The differences between 2D-SS and Stability Selection are small (about two standard errors at $k=10$ and less at $k=3$), Jitter-only is at least as good as 2D-SS, and selection stability is low for every method (Nogueira index at most $0.19$). With no ground truth and a weak signal, we read these results as a description of how the procedures behave on one data set and not as evidence for a general ranking.
}

{
\subsection*{Sensitivity to the noise grid and the weights}\label{study:sens}
\paragraph{Design.} The grid $\mathcal G$ and the weights $w_\delta$ are inputs of 2D-SS. We vary them in the design of Study~2 (Example~\ref{ex:jy}, $p=1{,}000$, $n=200$, $\delta_{\mathrm{obs}}=0$, $K\in\{1,5\}$ confounders, $100$ replicates). For each replicate we compute the selection frequencies once on the fine grid $\delta\in\{0,0.125,\dots,3\}$ ($B=50$ half-sample fits per level), and every configuration below is formed from the same fits. The reference is Stability Selection with $250$ half-sample fits at $\delta=0$. The configurations are: (a) $m=5$ equally spaced levels in $[0,\delta_{\max}]$ with uniform weights, $\delta_{\max}\in\{0.5,1,2,3\}$; (b) $m\in\{3,5,9\}$ levels in $[0,1]$; (c) weights $w_\delta\propto\exp(-\kappa\,\delta/\delta_{\max})$ on $m=5$ levels (here $\kappa$ is a tilt parameter, not the probability level of Proposition~\ref{prop:bridge}), $\kappa\in\{-2,0,2,4\}$, for $\delta_{\max}\in\{1,2\}$, so that $\kappa<0$ favours large $\delta$ and $\kappa>0$ favours small $\delta$; and (d) the four levels $\{0.25,0.5,0.75,1\}$, without the slice $\delta=0$. Each is evaluated under the largest-gap rule and with the model size $s=5$ known. Figures~\ref{fig:sens_f1}--\ref{fig:sens_margin} show $F_1$, the rate of selecting at least one confounder and the separation margin (the lowest frequency among relevant variables minus the highest among confounders).

\begin{figure}[H]
\centering
\includegraphics[width=\textwidth]{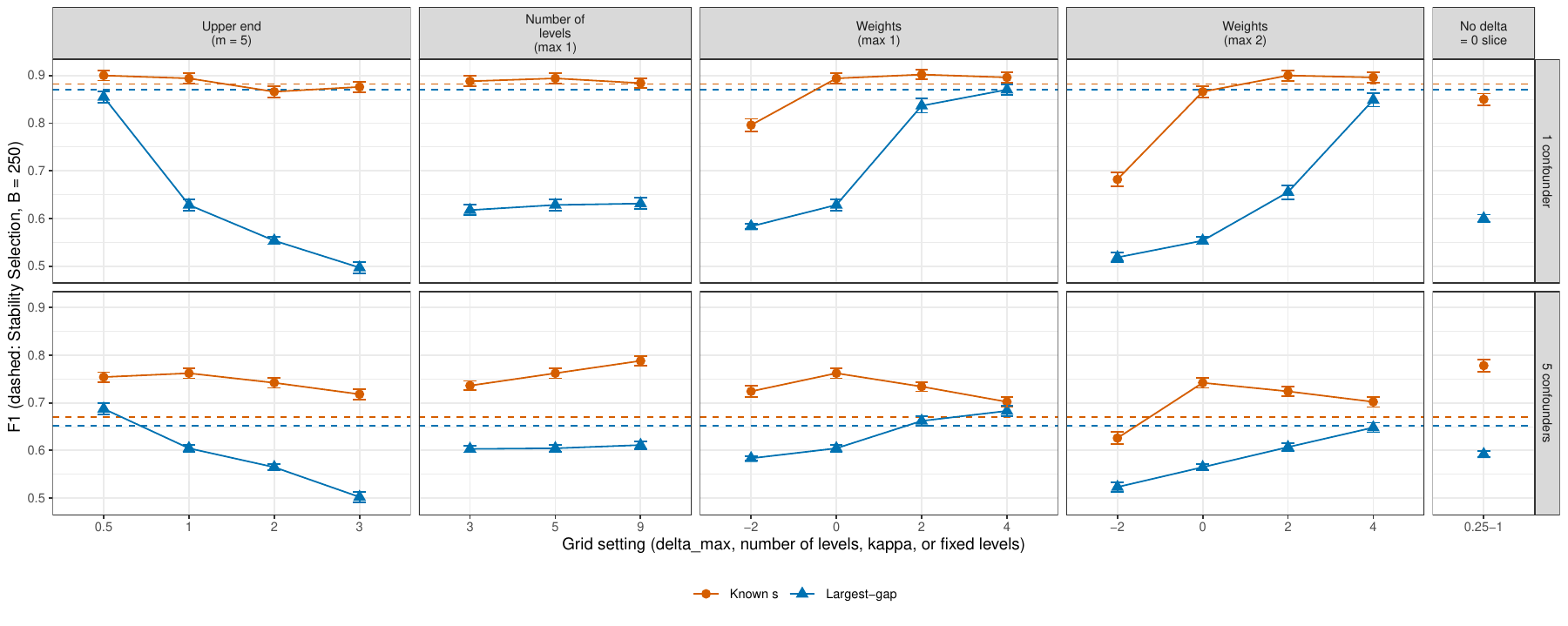}
\caption{Sensitivity to the grid and the weights: $F_1$ by configuration for the largest-gap rule and for known $s$, with $K=1$ (top) and $K=5$ (bottom) confounders. Dashed lines: Stability Selection with $250$ half-sample fits. Bars are $\pm$ one standard error over $100$ replicates; $\kappa$ is the weight parameter of the text.}
\label{fig:sens_f1}
\end{figure}

\begin{figure}[H]
\centering
\includegraphics[width=\textwidth]{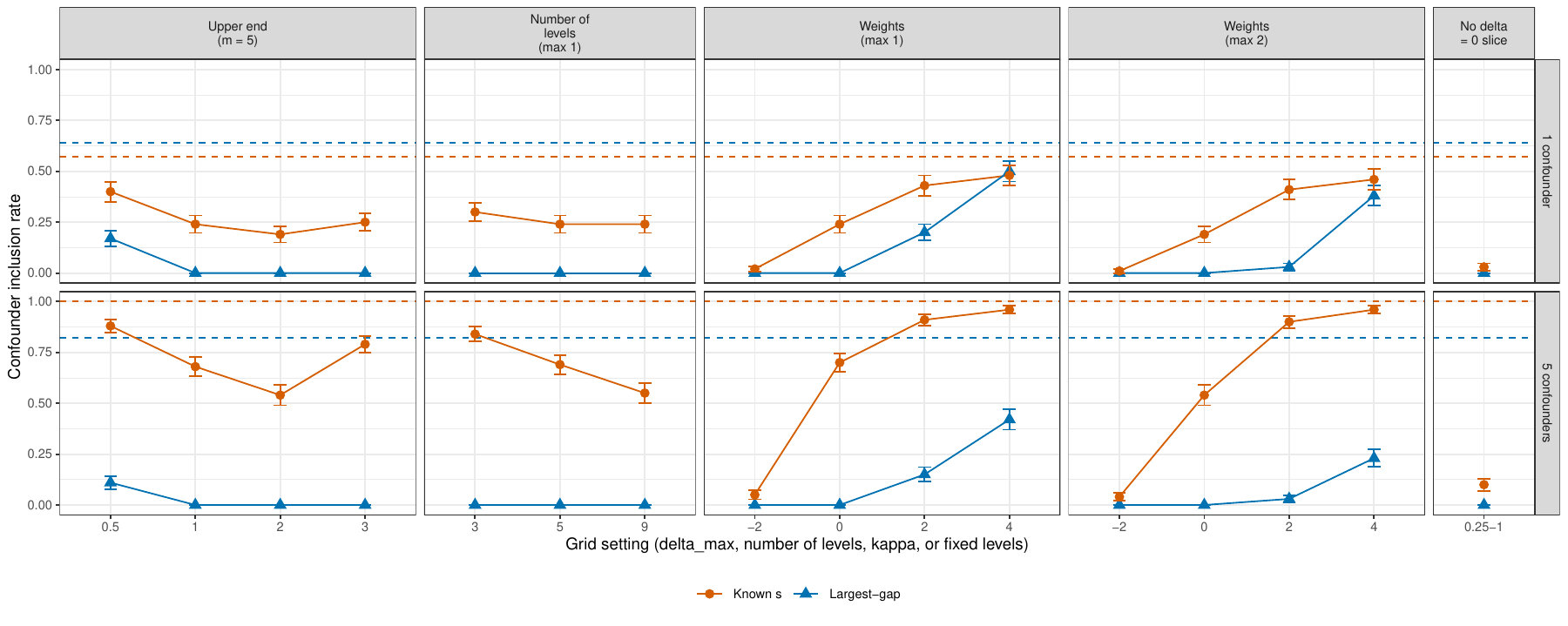}
\caption{As Figure~\ref{fig:sens_f1}, for the rate at which at least one confounder is selected.}
\label{fig:sens_confounder}
\end{figure}

\begin{figure}[H]
\centering
\includegraphics[width=\textwidth]{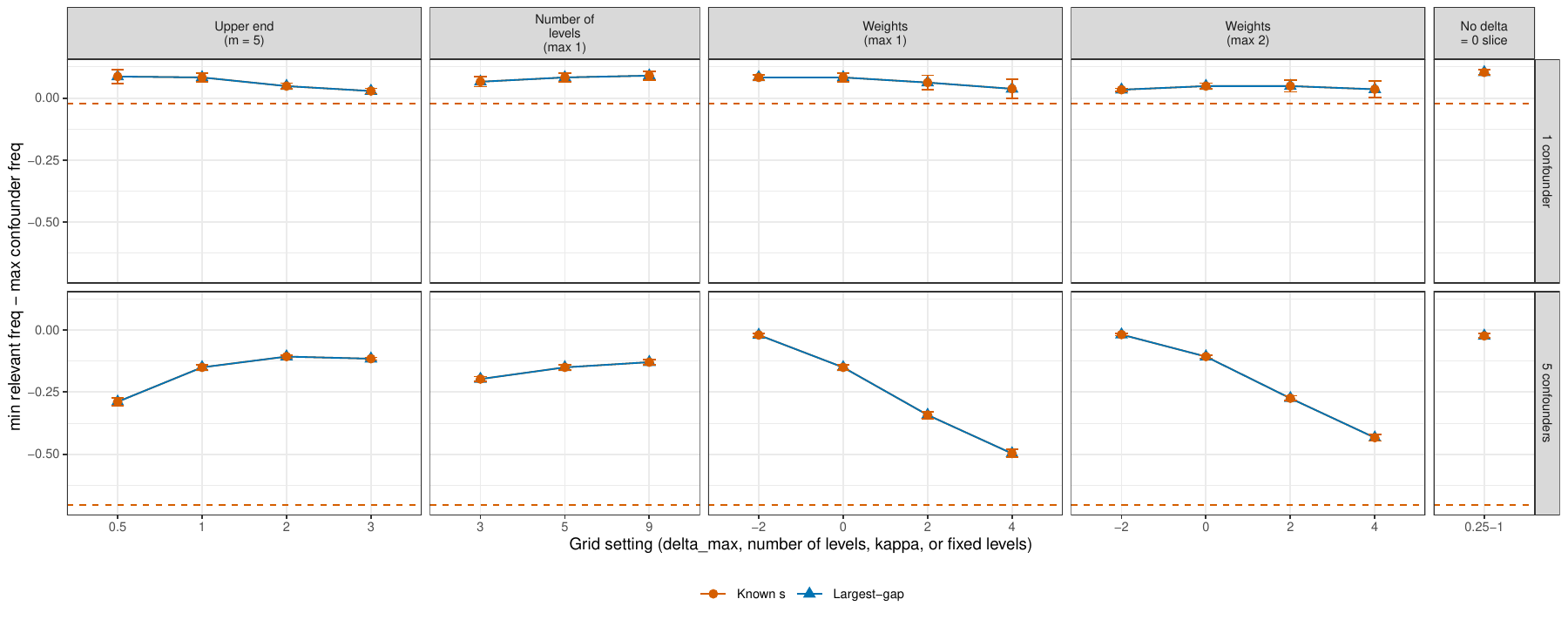}
\caption{As Figure~\ref{fig:sens_f1}, for the margin between the lowest selection frequency of a relevant variable and the highest of a confounder.}
\label{fig:sens_margin}
\end{figure}

\paragraph{Upper end and number of levels.} With one confounder, Stability Selection selects it in $64\%$ of replicates under the largest-gap rule ($57\%$ with $s$ known). Moving the upper end from $0.5$ to $3$ lowers the rate to $17\%$ at $\delta_{\max}=0.5$ and $0$ from $\delta_{\max}=1$ on, but under the largest-gap rule it also lowers $F_1$, from $0.855$ to $0.628$, $0.554$ and $0.497$, because averaging over larger noise levels pulls the frequencies of the weak signals towards the noise floor (Study~5). With $s$ known the effect on $F_1$ is small ($0.87$--$0.90$ across $\delta_{\max}$, against $0.88$ for Stability Selection) while the inclusion rate falls from $0.57$ to $0.19$--$0.40${ ($0.40$, $0.24$, $0.19$, $0.25$ for $\delta_{\max}=0.5,1,2,3$)}. The number of levels does not matter {under the largest-gap rule}: $m=3$, $5$, $9$ give $F_1$ of $0.62$, $0.63$, $0.63$ with no confounder selected{; with $s$ known the inclusion rate is $0.30$, $0.24$, $0.24$}. With five confounders the picture is similar but weaker{, and the number of levels matters under known $s$ (inclusion rate $0.84$, $0.69$, $0.55$ for $m=3$, $5$, $9$; the configuration $m=5$ on $[0,1]$ is run separately in each family, which explains $0.68$--$0.70$ across the three)}. Stability Selection selects a confounder in $82\%$ of replicates ($100\%$ with $s$ known), and the margin is $-0.70$, so the strongest confounder is far more frequent than the weakest relevant variable; the {upper-end and number-of-levels} grids raise the margin to between $-0.29$ and $-0.11$ and lower the rate under known $s$ only to $0.54$--$0.88$.

\paragraph{Weights.} Weights that favour small $\delta$ bring the procedure back towards Stability Selection. With one confounder and $\delta_{\max}=1$, $\kappa=4$ gives $F_1=0.87$ under the largest-gap rule, level with Stability Selection, but selects a confounder in $50\%$ of replicates; $\kappa=2$ gives $0.84$ and $20\%$. Weights that favour large $\delta$ ($\kappa=-2$) remove the confounder almost entirely ($0\%$ under the largest-gap rule, $2\%$ with $s$ known) at a cost in $F_1$ ($0.58$ and $0.80$). With five confounders and $s$ known the inclusion rate at $\delta_{\max}=1$ rises from $0.05$ at $\kappa=-2$ to $0.70$ at $\kappa=0$ and $0.96$ at $\kappa=4$.

\paragraph{Dropping the slice $\delta=0$.} Excluding the Stability Selection slice {lowers the inclusion rate ($0.03$ and $0.10$ for $K=1,5$ with $s$ known, zero under the largest-gap rule) at a smaller cost in $F_1$ than the weights that favour large $\delta$ ($\kappa=-2$ at $\delta_{\max}=1$: $0.80$ and $0.72$ for $K=1$ and $5$ under known $s$).} With $s$ known, one confounder gives $F_1=0.85$ and a inclusion rate of $0.03$ (against $0.88$ and $0.57$ for Stability Selection and $0.89$ and $0.24$ for uniform weights on $[0,1]$), and five confounders give $F_1=0.78$ and a rate of $0.10$ (against $0.67$ and $1.00$ for Stability Selection and $0.76$ and $0.70$ for uniform weights). The slice $\delta=0$ is where the confounders are most frequent, so removing it removes most of their contribution to $\widehat{\bar f}_j$; this agrees with Proposition~\ref{prop:bridge}, which asks for somewhat more than half of the total weight on levels at which the single-fit conditions hold. Under the largest-gap rule the confounder is never selected ($F_1=0.60$ and $0.59$).

\paragraph{Reading.} No setting dominates. Putting more weight on larger noise levels excludes confounders more reliably, and puts the weak relevant variables at risk, particularly under the largest-gap rule; putting more weight near $\delta=0$ does the reverse. The number of levels is immaterial{ for one confounder under the largest-gap rule}. Uniform weights on $[0,1]$ with five levels are a reasonable default when the model size is known and there is one confounder, but with several confounders the inclusion rate stays high unless little weight is placed near $\delta=0$.
}

{
\subsection*{Study 6: implementation and additional results}\label{study:enet_app}
\paragraph{Matching ENet-SS to the jittered fit.} subsampling, grid, weights, $B$, $\lambda$ and the selection rule are those of 2D-SS. At level $\delta$, ENet-SS replaces the Lasso fit on $(X_I+W,Y_I)$ by an elastic net on the un-jittered half-sample $(X_I,Y_I)$. The jittered columns are re-standardised before fitting (Section~\ref{s2:small}) and have variance close to $1+\delta^2$, so in the coefficients $b'=b/\sqrt{1+\delta^2}$ of the un-jittered standardised design the expected jittered objective is
\[
\frac1{2n_1}\|Y_I-X_Ib'\|_2^2+\frac{\delta^2}2\|b'\|_2^2+\lambda\sqrt{1+\delta^2}\,\|b'\|_1 .
\]
ENet-SS therefore uses ridge weight $\delta^2$ and $\ell_1$ weight $\lambda\sqrt{1+\delta^2}$; at $\delta=0$ it is the Lasso, that is, Stability Selection. In \texttt{glmnet}, whose Gaussian family standardises the response and so divides the ridge part by the standard deviation $\mathrm{sd}(Y_I)$ of the response in the fit, this is the penalty $L\{\tfrac{1-a}{2}\|b\|_2^2+a\|b\|_1\}$ with $L=\lambda\sqrt{1+\delta^2}+\delta^2\,\mathrm{sd}(Y_I)$ and $a=\lambda\sqrt{1+\delta^2}/L$. In a check with $n=200{,}000$ and $p=6$, the coefficients of the jittered Lasso and of this elastic net agreed to within $0.003$ at $\delta\in\{0,0.5,1\}$ (\texttt{glmnet} 5.1). The matching is exact for $\delta_{\mathrm{obs}}=0$. With $\delta_{\mathrm{obs}}>0$ (Table~\ref{tab:enetss_obs}) the observed columns already have variance $1+\delta_{\mathrm{obs}}^2$, so the exactly matching ridge weight would be $\delta^2/(1+\delta_{\mathrm{obs}}^2)$; the same $\delta^2$ is used there, so ENet-SS and 2D-SS-ENet carry somewhat more ridge than the jittered fits they mirror.

\paragraph{The fixed elastic net.} SS-ENet and 2D-SS-ENet use as base selector the elastic net with ridge weight $\gamma^2=0.25$ and $\ell_1$ weight $\lambda\sqrt{1+\gamma^2}$, that is, ENet-SS at the single level $\delta=\gamma=0.5$, in the same \texttt{glmnet} parametrisation. SS-ENet applies it to the half-samples without noise, with $250$ fits. 2D-SS-ENet applies it to the re-standardised jittered half-samples at the levels of the grid. In the coefficients $b'$ of the un-jittered standardised design, the expected objective of 2D-SS-ENet at level $\delta$ then has ridge weight $\gamma^2+\delta^2(1+\gamma^2)$ and $\ell_1$ weight $\lambda\sqrt{(1+\gamma^2)(1+\delta^2)}$, so its ridge weight exceeds that of ENet-SS at every level.

\paragraph{Training data with measurement error.} Table~\ref{tab:enetss_obs} gives the design of Study~2 with noise $\delta_{\mathrm{obs}}\in\{0.5,1\}$ added to the training data ($n\in\{200,400\}$, $100$ training sets per row). The confounder is rarely selected by any procedure (in at most $8.0\%$ of training sets); the procedures differ in how conservative they are. 2D-SS makes fewer false positives than each of the other four in every setting (at most $0.09$ per training set), and 2D-SS-ENet makes fewer than SS-ENet and ENet-SS. The two procedures that add noise also keep the fewest true positives. At $\delta_{\mathrm{obs}}=1$ none of them recovers more than about $2.4$ of the five relevant variables, which resembles the behaviour of the largest-gap rule in Study~5, where it can stop at the strong signals, and exact recovery of $S$ is zero for all five; at $\delta_{\mathrm{obs}}=0.5$ it is at most $19.0\%$ (Stability Selection, $n=400$, $K=1$; Table~\ref{tab:enetss_obs}), against $7.0\%$ for SS-ENet, $9.0\%$ for ENet-SS, $4.0\%$ for 2D-SS and $3.0\%$ for 2D-SS-ENet in that setting.

\begin{center}
\footnotesize
\setlength{\tabcolsep}{4pt}
\refstepcounter{table}\label{tab:enetss_obs}
\begin{minipage}{0.95\textwidth}
\textbf{Table~\thetable.} Study 6, training data with measurement error: percentage of training sets with at least one confounder selected, mean numbers of true positives (of $5$) and false positives, and percentage of training sets in which $S$ is recovered exactly ($100$ training sets per row, largest-gap rule).
\end{minipage}

\vspace{4pt}
\begin{tabular}{rrrrrrrr}
\toprule
$\delta_{\mathrm{obs}}$ & $K$ & $n$ & SS & SS-ENet & ENet-SS & 2D-SS & 2D-SS-ENet\\
\midrule
\multicolumn{8}{l}{\textit{Any confounder selected (\%)}}\\
0.5 & 1 & 200 & 0.0 & 0.0 & 0.0 & 0.0 & 0.0 \\
0.5 & 1 & 400 & 1.0 & 0.0 & 0.0 & 0.0 & 0.0 \\
0.5 & 5 & 200 & 2.0 & 2.0 & 1.0 & 0.0 & 0.0 \\
0.5 & 5 & 400 & 8.0 & 1.0 & 0.0 & 0.0 & 0.0 \\
\addlinespace
1 & 1 & 200 & 0.0 & 0.0 & 0.0 & 0.0 & 0.0 \\
1 & 1 & 400 & 0.0 & 0.0 & 0.0 & 0.0 & 0.0 \\
1 & 5 & 200 & 0.0 & 1.0 & 0.0 & 0.0 & 0.0 \\
1 & 5 & 400 & 0.0 & 0.0 & 1.0 & 0.0 & 0.0 \\
\midrule
\multicolumn{8}{l}{\textit{True positives (of $5$)}}\\
0.5 & 1 & 200 & 2.42 & 2.29 & 2.26 & 2.06 & 2.06 \\
0.5 & 1 & 400 & 3.96 & 3.68 & 3.53 & 2.87 & 2.70 \\
0.5 & 5 & 200 & 2.33 & 2.27 & 2.21 & 2.05 & 2.06 \\
0.5 & 5 & 400 & 3.53 & 3.63 & 3.39 & 2.65 & 2.68 \\
\addlinespace
1 & 1 & 200 & 1.98 & 1.95 & 1.97 & 1.80 & 1.80 \\
1 & 1 & 400 & 2.33 & 2.35 & 2.35 & 2.14 & 2.19 \\
1 & 5 & 200 & 2.00 & 1.95 & 1.98 & 1.80 & 1.80 \\
1 & 5 & 400 & 2.26 & 2.28 & 2.32 & 2.12 & 2.14 \\
\midrule
\multicolumn{8}{l}{\textit{False positives}}\\
0.5 & 1 & 200 & 0.14 & 0.28 & 0.16 & 0.01 & 0.02 \\
0.5 & 1 & 400 & 0.31 & 0.50 & 0.33 & 0.03 & 0.12 \\
0.5 & 5 & 200 & 0.16 & 0.32 & 0.18 & 0.00 & 0.02 \\
0.5 & 5 & 400 & 0.39 & 0.48 & 0.33 & 0.01 & 0.13 \\
\addlinespace
1 & 1 & 200 & 0.07 & 0.31 & 0.21 & 0.03 & 0.18 \\
1 & 1 & 400 & 0.22 & 0.53 & 0.41 & 0.07 & 0.18 \\
1 & 5 & 200 & 0.10 & 0.35 & 0.23 & 0.03 & 0.15 \\
1 & 5 & 400 & 0.22 & 0.54 & 0.37 & 0.09 & 0.15 \\
\midrule
\multicolumn{8}{l}{\textit{Exact recovery of $S$ (\%)}}\\
0.5 & 1 & 200 & 0.0 & 0.0 & 0.0 & 0.0 & 0.0 \\
0.5 & 1 & 400 & 19.0 & 7.0 & 9.0 & 4.0 & 3.0 \\
0.5 & 5 & 200 & 0.0 & 0.0 & 0.0 & 0.0 & 0.0 \\
0.5 & 5 & 400 & 7.0 & 7.0 & 6.0 & 2.0 & 3.0 \\
\addlinespace
1 & 1 & 200 & 0.0 & 0.0 & 0.0 & 0.0 & 0.0 \\
1 & 1 & 400 & 0.0 & 0.0 & 0.0 & 0.0 & 0.0 \\
1 & 5 & 200 & 0.0 & 0.0 & 0.0 & 0.0 & 0.0 \\
1 & 5 & 400 & 0.0 & 0.0 & 0.0 & 0.0 & 0.0 \\
\bottomrule
\end{tabular}
\end{center}

\paragraph{Prediction under future noise.} As in Study~2, the support selected by each procedure is refitted by ordinary least squares on the training data and evaluated on $500$ fresh observations at noise level $\delta_{\mathrm{test}}$ ($n=100$, training data clean, $200$ training sets). Table~\ref{tab:enetss_mse} gives paired differences from ENet-SS. The procedures are close: SS-ENet differs from ENet-SS by at most $0.48$ and 2D-SS by at most $0.25$ at any $\delta_{\mathrm{test}}$, both below $1\%$ of the MSE for $\delta_{\mathrm{test}}\ge1$. 2D-SS-ENet has a higher MSE than ENet-SS for $\delta_{\mathrm{test}}\le0.5$ (by $0.26$--$0.37$) and, from $\delta_{\mathrm{test}}=1.5$ on, the lowest MSE of the five, $3.7$ below ENet-SS at $\delta_{\mathrm{test}}=3$ ($2\%$ of the MSE). Stability Selection has the lower MSE on clean test data and the higher MSE for $\delta_{\mathrm{test}}\ge1$, by $4.0$ at $\delta_{\mathrm{test}}=3$.

\begin{center}
\footnotesize
\setlength{\tabcolsep}{4pt}
\refstepcounter{table}\label{tab:enetss_mse}
\begin{minipage}{0.95\textwidth}
\textbf{Table~\thetable.} Study 6, prediction under future noise ($n=100$, training data clean, $200$ training sets, OLS refit on the selected set). For each $K$: mean test MSE of ENet-SS, and the paired difference in mean MSE from ENet-SS of each other procedure, with paired standard errors in parentheses. A positive entry means a larger MSE than ENet-SS.
\end{minipage}

\vspace{4pt}
\begin{tabular}{rrrrrr}
\toprule
 & ENet-SS & \multicolumn{4}{c}{Difference from ENet-SS}\\
\cmidrule(lr){3-6}
$\delta_{\mathrm{test}}$ & mean MSE & SS & SS-ENet & 2D-SS & 2D-SS-ENet\\
\midrule
\multicolumn{6}{l}{\textit{$K=1$}}\\
0 & 2.11 & -0.32 (0.02) & 0.02 (0.01) & 0.13 (0.04) & 0.37 (0.08) \\
0.5 & 7.10 & -0.19 (0.02) & 0.00 (0.01) & 0.11 (0.03) & 0.26 (0.06) \\
1 & 22.29 & 0.12 (0.05) & -0.04 (0.02) & 0.10 (0.03) & -0.05 (0.05) \\
1.5 & 47.13 & 0.73 (0.13) & -0.09 (0.06) & 0.10 (0.08) & -0.61 (0.17) \\
2 & 81.29 & 1.52 (0.26) & -0.21 (0.08) & -0.04 (0.15) & -1.44 (0.36) \\
2.5 & 126.64 & 2.70 (0.39) & -0.29 (0.14) & 0.04 (0.24) & -2.33 (0.59) \\
3 & 183.27 & 4.03 (0.59) & -0.48 (0.18) & -0.25 (0.36) & -3.71 (0.91) \\
\midrule
\multicolumn{6}{l}{\textit{$K=5$}}\\
0 & 2.11 & -0.20 (0.02) & 0.01 (0.01) & 0.08 (0.03) & 0.36 (0.08) \\
0.5 & 7.10 & -0.08 (0.02) & 0.00 (0.01) & 0.08 (0.03) & 0.26 (0.06) \\
1 & 22.30 & 0.22 (0.05) & -0.04 (0.02) & 0.13 (0.03) & -0.06 (0.06) \\
1.5 & 47.16 & 0.74 (0.13) & -0.11 (0.05) & 0.15 (0.06) & -0.63 (0.18) \\
2 & 81.29 & 1.66 (0.24) & -0.21 (0.07) & 0.12 (0.12) & -1.45 (0.37) \\
2.5 & 126.67 & 2.58 (0.39) & -0.32 (0.13) & 0.25 (0.19) & -2.36 (0.60) \\
3 & 183.27 & 4.00 (0.57) & -0.47 (0.16) & 0.16 (0.27) & -3.70 (0.91) \\
\bottomrule
\end{tabular}
\end{center}
}

{
\subsection*{Remark on bias-corrected methods and SIMEX (Section~\ref{s2:small})}
\begin{remark}[Relation to bias-corrected methods and SIMEX]\label{r:simex}
The matrix $\Sigma^{(\delta)}-\nu^2I_p$ (see \eqref{eq:gram}) is the noise-corrected Gram surrogate used in measurement-error Lasso methods \citep{loh2012high,datta2017cocolasso}, which aim to undo the effect of a known or estimated error covariance and thereby target the clean-design model. 2D-SS deliberately fits the \emph{uncorrected} Gram matrix, so that the ridge term of Proposition~\ref{prop:ridge} is present, and it uses the behaviour of selection under increasing noise, rather than a correction, as the source of information; the fluctuation $\widetilde\Delta$ below is the difference between the corrected and the uncorrected Gram matrix, and is used only as an analysis device. SIMEX \citep{cook1994simulation} also adds noise at increasing levels, but in order to extrapolate estimates back to the noise-free case, whereas we aggregate selection frequencies along the path and do not extrapolate.
\end{remark}

}

{
\section*{Further results and remarks for Section~\ref{s2}}
\label{app:s2extra}
This section collects results and remarks that were moved out of Section~\ref{s2} to shorten it.

\noindent\textbf{A sharper sample-size condition under Gaussian design.} Theorem~\ref{thm:single_level} treats a fixed design and bounds the perturbation blocks in the $\infty$-norm, which costs a factor $s$ and a logarithm in the required sample size (Remark~\ref{r:range}). Suppose instead that the rows of $X_0$ are i.i.d.\ $\mathcal N(0,\Sigma)$, with $\Sigma$ now the population covariance matrix, the design used in classical sharp analyses of the Lasso \citep{wainwright2009sharp} and in our simulations; $\varepsilon$ and $W_0$ are as before. {Throughout this subsection $Y_I=X_{0,I}\beta+\varepsilon_I$ with $\beta_{S^c}=0$, $\varepsilon\sim\mathcal N(0,\sigma^2I)$ with $\sigma>0$ independent of $(X_0,W_0,W)$, $\lambda>0$, $p>s$ and $n_1\ge s$.} Then the rows of the jittered design $\widetilde X\coloneq X_I+W$ are i.i.d.\ $\mathcal N(0,\Sigma_\nu)$ with $\Sigma_\nu=\Sigma+\nu^2I_p$, and the dual vector of the jittered Lasso admits an exact decomposition whose random part is conditionally Gaussian. We use the following notation, all of it for the analysis only: $\widehat\Sigma\coloneq\widetilde X^\top\widetilde X/n_1$, $\widehat C\coloneq\widehat\Sigma_{SS}$, $\Omega\coloneq\widehat C^{-1}$, $T\coloneq\Sigma_{SS}+\nu^2I_s$, $P_\nu\coloneq\Sigma_{SS}T^{-1}$, and, for $k\in S^c$, $a_k\coloneq\Sigma_{kS}$, $c_k\coloneq T^{-1}a_k^\top$ and $e_k\coloneq\widetilde X_k-\widetilde X_Sc_k$; further $q\coloneq Y_I-\widetilde X_SP_\nu\beta_S$, $h\coloneq\widetilde X_S^\top q/n_1$ and $\psi\coloneq\widetilde X_S\Omega(z-h/\lambda)/n_1+q/(\lambda n_1)$. In this subsection $q$ is the response remainder vector (not the probability $q_\delta$ of Proposition~\ref{prop:bridge}), $c_k$ is a vector (not the statistic $c_\nu^k$ of Section~\ref{s2:small}), $r_S$ is the subvector of $r$ (not the scalar rate $r_s$ of Theorem~\ref{thm:gauss_rate}), and $\kappa_T$ (a condition number) and $\kappa_\nu$ (a matrix norm) are not the probability level $\kappa$ of Proposition~\ref{prop:bridge}.
\begin{lemma}[Exact dual decomposition]\label{lem:decomp}
Let $\widehat C$ be invertible, $r\coloneq\widetilde X^\top Y_I/n_1$, $\hat\beta_S\coloneq\Omega(r_S-\lambda z)$ and $\hat z_{S^c}\coloneq(r_{S^c}-\widehat\Sigma_{S^cS}\hat\beta_S)/\lambda$. Then, for every $k\in S^c$,
\[
\hat z_k=c_\nu^k+F_k,\qquad F_k\coloneq e_k^\top\psi-\frac{\mu_k}{\lambda},\qquad \mu_k\coloneq\nu^2a_kT^{-1}\beta_S,
\]
and $\hat\beta_S-\beta^\star_\nu=\Omega\big(h+\lambda(\widehat C-T)T^{-1}z\big)$.
\end{lemma}
\begin{lemma}[Gaussian structure]\label{lem:gstruct}
Under Gaussian design, the triples $(\tilde x_{iS},e_{ik},q_i)$, $i\in I$, are i.i.d.\ and jointly Gaussian, $\tilde x_{iS}$ is independent of $(e_{ik},q_i)$, and
\[
\begin{gathered}
q_i\sim\mathcal N(0,s_q^2),\qquad s_q^2\coloneq\sigma^2+\nu^2\beta_S^\top\Sigma_{SS}T^{-1}\beta_S,\\
\operatorname{Var}(e_{ik})=s_k^2\coloneq\Sigma_{kk}+\nu^2-a_kT^{-1}a_k^\top,\qquad \operatorname{Cov}(e_{ik},q_i)=\mu_k .
\end{gathered}
\]
\end{lemma}
Thus $\hat z_k$ equals the population value $c_\nu^k$ plus a fluctuation $F_k$ that, conditionally on $(\widetilde X_S,q)$, is Gaussian with a bounded mean shift. This is what removes the factor $s$: only the norms of the fixed direction $\psi$ enter, not worst-case bounds on the perturbation blocks.
\begin{theorem}[Exact recovery by a jittered fit under Gaussian design]\label{thm:gauss_rate}
Fix $\delta\ge0$, $\nu=\nu_\delta$, $\alpha\in(0,1)$ and $\ell\coloneq\log(9p/\alpha)$, and assume the Gaussian design with $\Sigma_{SS}$ positive definite. Let $\lambda_T\coloneq\lambda_{\min}(\Sigma_{SS})+\nu^2$, $\|T\|\coloneq\lambda_{\max}(\Sigma_{SS})+\nu^2$, $\kappa_T\coloneq\|T\|/\lambda_T$, $\bar\mu\coloneq\max_{k\in S^c}|\mu_k|$, $\sigma_x\coloneq\max_{k\in S^c}(\Sigma_{kk}+\nu^2)^{1/2}$, and
\[
\begin{gathered}
r_s\coloneq\sqrt{\tfrac s{n_1}}+\sqrt{\tfrac{2\ell}{n_1}},\qquad r_q\coloneq2\sqrt{\tfrac\ell{n_1}}+\tfrac{2\ell}{n_1},\\
G_s\coloneq\big(\sqrt s+\sqrt{2\ell}\big)^2,\qquad \bar h\coloneq s_q\sqrt{\tfrac{3\|T\|G_s}{2n_1}},\\
\Lambda_\nu\coloneq s_q\sqrt{\tfrac{4\ell s}{\lambda_Tn_1}}+\tfrac{s_q^2G_s}{\lambda n_1}+\tfrac{s_q^2r_q}{\lambda},\\
\Xi_\nu\coloneq\sqrt{\tfrac{2}{\lambda_Tn_1}}\Big(\sqrt s+\tfrac{\bar h}\lambda\Big)+\tfrac{s_q}{\lambda}\sqrt{\tfrac{1+r_q}{n_1}},\\
\Phi_\nu\coloneq\tfrac{\bar\mu}{s_q^2}\Lambda_\nu+\sigma_x\sqrt{2\ell}\,\Xi_\nu .
\end{gathered}
\]
Suppose that
\begin{enumerate}
\item[(i)] EIC$_\nu$ holds with margin $\eta\in(0,1)$;
\item[(ii)] $6\kappa_Tr_s\le1$;
\item[(iii)] $\Phi_\nu\le\eta/2$;
\item[(iv)] $\theta_\nu>\dfrac{2}{\lambda_T}\big(\bar h+3\lambda\kappa_Tr_s\sqrt s\big)$.
\end{enumerate}
Then, with probability at least $1-\alpha$ over $(X_0,W_0,W,\varepsilon,I)$, the Lasso fitted to $(X_I+W,Y_I)$ has a unique solution, its support equals $S$, and $\operatorname{sign}(\hat\beta_S)=\operatorname{sign}(\beta_S)$.
\end{theorem}
Proofs are shown in the appendix. Lemma~\ref{lem:decomp} is the exact analogue of the classical decomposition of the Lasso dual vector under Gaussian design; here the effective regression of $Y_I$ on the jittered design has noise variance $s_q^2$ and an additional covariance $\mu_k$ between the noise and the residual $e_k$ of an irrelevant variable, and $\mu_k/\lambda$ is precisely the ridge-bias term that distinguishes $v_\nu$ from $z$.
{
\begin{remark}[Comparison with CoCoLasso's sign-consistency guarantee]\label{r:cocolasso_compare}
\citet{datta2017cocolasso} establish sign consistency for CoCoLasso (their Theorem 2) {under a hypothesis that is, in our notation, the Lasso irrepresentable condition on the Gram matrix of the fixed latent design} $\|\Sigma_{S^cS}\Sigma_{SS}^{-1}\|_\infty<1$ together with $\lambda_{\min}(\Sigma_{SS})>0$; it implies, and is at least as strong as, EIC$_0$. Their theory has no analogue of EIC$_\nu$ for $\nu>0$: it offers no guarantee, and indeed no mechanism, when this condition fails on the clean design. Consequently, in a design such as Example~\ref{ex:jy}, where the clean IC is violated ($A=1.125>1$, {which also violates their irrepresentable condition)}, their Theorem 2 simply does not apply, {whereas Corollary~\ref{cor:benefit} gives conditions at a level $\delta_1$ (in Example~\ref{ex:jy}, EIC holds on $\nu\in[0.279,1.356]$) at which a single jittered Lasso fit recovers $S$ exactly, with probability at least $1-\bar\alpha$ under the sample-size and tuning conditions; the aggregated procedure additionally needs somewhat more than half of the grid weight on such levels (Proposition~\ref{prop:bridge}).} This is a difference in the \emph{domain of applicability} of the two guarantees, not merely a difference in rate; the two domains are not nested {(at fixed $n$ our local regime can be empty, for example when $\delta_{\mathrm{obs}}$ is large relative to the admissible window of Remark~\ref{r:range}, or when EIC$_\nu$ fails for every $\nu\ge\nu_0=\delta_{\mathrm{obs}}$, whereas CoCoLasso with a known error covariance still applies)}.

{On rate, no advantage over CoCoLasso is claimed. Under their stated scaling, CoCoLasso's sign consistency requires $n\gtrsim s^2\log p$ for a fixed design with sub-Gaussian noise and measurement errors. Our fixed-design condition (Theorem~\ref{thm:single_level}, via Remark~\ref{r:range}) carries the explicit factor $s^2\log(ps)$ times a constant to the power $-2$ that itself depends on the design, which grows polynomially in $s$, so it is not of a better order. The $s\log p$ rate of Theorem~\ref{thm:gauss_rate} concerns conditions (ii)--(iii), that is, dual feasibility and Gram concentration, only, holds under Gaussian design, is the classical rate for the plain Lasso at $\nu=0$ \citep{wainwright2009sharp}, and Remark~\ref{r:tuning} shows a larger requirement for $\nu>0$.}

A further structural difference connects directly to Section~\ref{s3:baselines}: CoCoLasso requires the measurement-error covariance $\Sigma_A$ to be known (their Section~4.1 and Corollary~2, which specialises Theorem~2 to additive errors; if replicates are available it can be estimated from them), an assumption they take from the corrected Lasso of \citet{loh2012high}, {on which their method builds and against which it is compared; the corrected Lasso has no sign-consistency result in \citet{loh2012high}, only} estimation- and optimisation-error bounds \citep[Theorems~1--2]{loh2012high}, as also noted by \citet{datta2017cocolasso}. This is precisely the oracle condition tested empirically in Section~\ref{s3:baselines}: the misspecified condition there is not an unfavourable choice of tuning parameter, but a setting in which the cited guarantees' own hypotheses are not met.
\end{remark}
}

\begin{remark}[Sample size and level-specific tuning]\label{r:tuning}
(a) Condition (ii) reads $n_1\ge36\kappa_T^2G_s$. {Up to lower-order terms, and using $|\mu_k|\le s_ks_q\le\sigma_xs_q$, the leading part of $\Phi_\nu$ is at most $\sigma_x\sqrt{\ell/n_1}\,\big(4\sqrt{s/\lambda_T}+(2+\sqrt2)s_q/\lambda\big)$, so (iii) holds, up to the dropped terms, once $n_1\ge4\ell\sigma_x^2\big(4\sqrt{s/\lambda_T}+(2+\sqrt2)s_q/\lambda\big)^2/\eta^2$, which is of order $\sigma_x^2(s+s_q^2\lambda_T/\lambda^2)\log(p/\alpha)/(\lambda_T\eta^2)$. The dropped terms include the bias term $(\bar\mu/s_q^2)\,s_q^2G_s/(\lambda n_1)=\bar\mu G_s/(\lambda n_1)\ge\bar\mu s/(\lambda n_1)$ (with $G_s\ge s$, and $G_s\approx s$ only for $s\gg\ell$), which comes from $\mathbb E\|\xi\|_2^2=s$ and is not of lower order in general: (iii) also requires $\bar\mu G_s/(\lambda n_1)\lesssim\eta$, in particular $\bar\mu s/(\lambda n_1)\lesssim\eta$. At $\nu=0$ we have $\mu=0$ and this term is absent, and for $\lambda$ bounded below the requirement is of order $s$; but for $\nu>0$ and $\lambda\asymp\sqrt{\log p/n_1}$ it reads $n_1\gtrsim(\bar\mu s/\eta)^2/\log p$ (the second leading term of $\Phi_\nu$ requires in addition that the constant in $\lambda\asymp\sqrt{\log p/n_1}$ dominate $s_q\sqrt{\ell/\log p}/\eta$). The dual-feasibility conditions (ii)--(iii) are therefore of order $s\log p$, as for the plain Lasso, only when the ridge-bias term $\bar\mu$ is small, and of order $s^2/\log p$ otherwise, in contrast to the fixed-design condition of Remark~\ref{r:range}, whose dependence on $s$ is at least $s^2\log(ps)$. }Condition (iv) is a beta-min condition on the shrunken target $\beta^\star_\nu$; because the proof bounds the sup-norm by the Euclidean norm, its sample-size requirement is not of order $s\log p$ in general. The constants in (ii)--(iv) are conservative; Theorem~\ref{thm:exact} below gives the exact finite-sample recovery probability under the same design, without concentration constants.
(b) All results of this section use the tuning parameter only through the conditions at the noise level under consideration, so $\lambda$ may be replaced by a level-specific $\lambda_\delta$; the default is $\lambda_\delta\equiv\lambda$. Condition (iii) contains the summand $\sigma_xs_q\sqrt{2\ell(1+r_q)/n_1}/\lambda$ and therefore requires $\lambda\ge2\sigma_xs_q(\nu)\sqrt{2\ell(1+r_q)/n_1}/\eta$, where $s_q(\nu)^2=\sigma^2+\nu^2\beta_S^\top\Sigma_{SS}(\Sigma_{SS}+\nu^2I)^{-1}\beta_S$ is the noise variance of the effective regression; it increases with $\nu$ (compare the heuristic inflation in Remark~\ref{r2}, which uses the larger variance $\sigma^2+\nu^2\|\beta\|_2^2$). Condition (iv) requires $\theta_\nu>0$ and thus bounds $\lambda$ from above (for $\Sigma_{SS}=I_s$, $\lambda<\min_j|\beta_j|$). The admissible values of $\lambda$ therefore form a window that moves with $\nu$, and a natural level-specific choice is $\lambda_\delta\propto s_q(\nu_\delta)\sqrt{\log p/n_1}$. With a fixed $\lambda$ some levels can fall outside the window, so that part of the decay of selection frequencies with $\delta$ may reflect tuning rather than instability.
(c) A tuning-free alternative that removes this confound is the size-matched selector of Section~\ref{s2:errctrl}: each fit selects exactly $q_\star$ variables, so frequencies cannot decay merely because of additional shrinkage.
\end{remark}

\noindent\textbf{Exact finite-sample recovery probability under Gaussian design.} The identity of Lemma~\ref{lem:decomp} can be made fully explicit. Let $H\coloneq\widetilde X_S\Omega\widetilde X_S^\top/n_1$ be the projection onto the column space of $\widetilde X_S$; then $\psi=\widetilde X_S\Omega z/n_1+(I-H)q/(\lambda n_1)$ is a sum of two orthogonal vectors, and all the quantities that drive $\hat z_{S^c}$ and $\hat\beta_S$ have known exact laws: a Wishart matrix, a $\chi^2$ variable and Gaussian vectors. This yields the exact distribution of the primal--dual witness, and hence of the recovery event, for every $n_1$.
\begin{theorem}[Exact finite-sample representation]\label{thm:exact}
Assume the Gaussian design with $\Sigma_{SS}$ positive definite, $\sigma>0$ and $p>s$, fix $\delta\ge0$, $\nu=\nu_\delta$ and $n_1\ge s+1$, and let $\zeta\coloneq z^\top T^{-1}z$, $\mu\coloneq(\mu_k)_{k\in S^c}$ and
\[
C_N\coloneq\Sigma_{S^cS^c}+\nu^2I_{p-s}-\Sigma_{S^cS}T^{-1}\Sigma_{SS^c}-\frac{\mu\mu^\top}{s_q^2}\ \ (\succeq0).
\]
Let $Z\in\mathbb R^{n_1\times s}$ be standard Gaussian, $\widehat C\coloneq T^{1/2}(Z^\top Z/n_1)T^{1/2}$, $\xi\sim\mathcal N(0,I_s)$, $\Upsilon\sim\chi^2_{n_1-s}$ and $N\sim\mathcal N(0,C_N)$, all independent, and define
\[
\begin{gathered}
h\coloneq\frac{s_q}{\sqrt{n_1}}\widehat C^{1/2}\xi,\qquad
\hat\beta\coloneq P_\nu\beta_S+\widehat C^{-1}(h-\lambda z),\\
\hat z\coloneq\Sigma_{S^cS}T^{-1}z+\frac{\mu}{s_q^2}\Big(h^\top\widehat C^{-1}z+\frac{s_q^2\Upsilon}{\lambda n_1}\Big)+\Big(\frac{z^\top\widehat C^{-1}z}{n_1}+\frac{s_q^2\Upsilon}{\lambda^2n_1^2}\Big)^{1/2}N .
\end{gathered}
\]
Then $(\hat\beta,\hat z)$ has the same distribution as the pair $(\hat\beta_S,\hat z_{S^c})$ of Lemma~\ref{lem:decomp} for the Lasso fitted to $(X_I+W,Y_I)$. Consequently, with $r_\delta\coloneq\mathbb P\big(\operatorname{sign}(\hat\beta)=z,\ \|\hat z\|_\infty<1\big)$,
\[
r_\delta\ \le\ \mathbb P\big(\text{the jittered Lasso has a solution with support }S\text{ and sign }z\big)\ \le\ \mathbb P\big(\operatorname{sign}(\hat\beta)=z,\ \|\hat z\|_\infty\le1\big),
\]
and the two bounds coincide when $\varpi_k^2\coloneq s_k^2-\mu_k^2/s_q^2>0$ for all $k\in S^c$.
\end{theorem}
\begin{corollary}[Zero--one law for fixed $p$ and $s$]\label{cor:zeroone}
Under the Gaussian design, fix $(p,s,\Sigma,\beta,\sigma,\lambda,\nu)$ with $p>s$ and $\sigma>0$, and let $n_1\to\infty$. If $c_\nu<1$ and $\theta_\nu>0$, the probability that the jittered Lasso has a solution with support $S$ and sign $z$ tends to $1$; if $c_\nu>1$ or $\theta_\nu<0$, it tends to $0$.
\end{corollary}
Proofs are shown in the appendix. Theorem~\ref{thm:exact} contains no concentration constants: it is an identity in distribution, so that $r_\delta$ can be computed to any accuracy by Monte Carlo in dimension of order $s+(p-s)$ (sampling $\widehat C$ by the Bartlett decomposition) at a cost that does not depend on $n_1$, whereas Theorem~\ref{thm:gauss_rate} only gives sufficient conditions. Corollary~\ref{cor:zeroone} (which leaves the boundary cases $c_\nu=1$ and $\theta_\nu=0$ open) is the counterpart, for the jittered Lasso, of the classical asymptotic characterisation of the Lasso by the irrepresentable condition, with EIC$_\nu$ and the sign margin $\theta_\nu$ in place of the IC and the beta-min condition. For instance, in Example~\ref{ex:jy} with Gaussian design, Corollary~\ref{cor:zeroone} gives a recovery probability tending to $0$ at $\nu=0$ ($c_0=1.125>1$) and to $1$ for every $\nu\in[0.279,1.356]$ ($c_\nu\le0.9$ and $\theta_\nu>0$).
\begin{remark}[Use of Theorem~\ref{thm:exact}]\label{r:exact}
Since $r_\delta$ is a lower bound on the probability of exact recovery at level $\delta$, it can replace the conservative conditions (i)--(iv) in the definition of the good noise levels below (the guarantee is then in terms of the exact recovery probability of a single fit); and, in simulations where $S$ is known, it can be compared quantitatively with observed recovery frequencies. To first order, as $n_1\to\infty$, $\sqrt{n_1}(\hat z-\hat c_\nu)$ is asymptotically $\mathcal N(0,V)$, where $\hat c_\nu\coloneq\Sigma_{S^cS}T^{-1}v_\nu$ and
\[
V\coloneq\mu\mu^\top\Big(\frac{\zeta}{s_q^2}+\frac2{\lambda^2}\Big)+\Big(\zeta+\frac{s_q^2}{\lambda^2}\Big)C_N ,
\]
which gives the following closed-form approximation (heuristic: first order in $n_1^{-1/2}$, with no error bound) $\mathbb P(\|\hat z\|_\infty<1)\approx\mathbb P\big(\|\hat c_\nu+n_1^{-1/2}V^{1/2}g\|_\infty<1\big)$ with $g\sim\mathcal N(0,I_{p-s})$ a standard Gaussian vector (matching the dimension of $\hat z\in\mathbb R^{p-s}$); it approximates the dual-feasibility event only, and $r_\delta\approx\mathbb P(\|\hat z\|_\infty<1)$ when $\theta_\nu>0$ and $n_1$ is large.
\end{remark}

\begin{remark}[Range of admissible noise levels]\label{r:range}
All constants in $t_1$, $t_2$ and $u_\Sigma$ are explicit, so sufficient conditions for $\omega\le\omega_0^*$ can be written down. It holds if each of the three terms of $\omega$ is at most $\omega_0^*/3$, which follows from
\[
\begin{gathered}
n_1\ge\frac{18\,s^2K^4\log(2ps/\alpha_I)}{(\omega_0^*)^2},\qquad
\nu\le\frac{\omega_0^*\sqrt{n_1}}{6\sqrt2\,K\,s\sqrt L},\\
\nu^2\le\frac{\omega_0^*\sqrt{n_1}}{12\,s\sqrt{L+1}}\quad\text{(the last one when }L+1\le n_1).
\end{gathered}
\]
The first is a condition on the sample size only. Because the blocks $\widetilde\Delta_{SS}$ and $\widetilde\Delta_{S^cS}$ have only $s$ columns, the row-sum factor is $s$ rather than $p$, and the union bound covers $ps$ rather than $p^2$ entries; the range of $\nu$ therefore depends on the model size $s$ instead of the number of variables $p$ (an improvement by a factor of order $p/s$ in the row-sum factor). Since $\nu_\delta^2=\delta_{\mathrm{obs}}^2+\delta^2$, the local regime is empty unless $\delta_{\mathrm{obs}}$ itself lies below the admissible range. The constants are conservative, so the guarantee is local; the link to the whole grid is provided by aggregation (Proposition~\ref{prop:bridge}{, below}). Note that $\omega_0^*$ itself depends on $s$ through $\bar\kappa$ and $V$, so the overall dependence of this fixed-design condition on $s$ is worse than the explicit factor $s^2\log(ps)$. Under Gaussian design, Theorem~\ref{thm:gauss_rate} {(below)} gives dual-feasibility conditions of order $s\log p$ when the noise-induced bias is small (Remark~\ref{r:tuning}).
\end{remark}

\begin{remark}\label{r2}
Let $X^{(\delta)} = X_0 + E_\delta$, where $E_\delta \in \mathbb{R}^{n\times p}$ has i.i.d. entries with mean $0$ and variance $\nu^2$, independent of $\varepsilon$. Consider
\[
Y = X_0\beta + \varepsilon,
\qquad \mathbb{E}(\varepsilon)=0,\ \mathrm{Cov}(\varepsilon)=\sigma^2 I_n .
\]
Since $X_0 = X^{(\delta)} - E_\delta$, we can rewrite
\[
Y = X^{(\delta)}\beta + \tilde{\varepsilon},
\qquad \tilde{\varepsilon} \coloneq \varepsilon - E_\delta \beta.
\]
Then
\[
\begin{aligned}
\mathrm{Cov}(\tilde{\varepsilon})
&= \mathrm{Cov}(\varepsilon) + \mathrm{Cov}(E_\delta \beta)\\
&= \sigma^2 I_n + \nu^2 \|\beta\|_2^2 I_n
= \big(\sigma^2 + \nu^2 \|\beta\|_2^2\big) I_n,
\end{aligned}
\]
using independence across columns of $E_\delta$ within each row and independence from $\varepsilon$. This suggests an effective increase in noise level (in a heuristic sense), degrading the signal-to-noise ratio and tending to suppress selections under fixed tuning (e.g., fixed $\lambda$). This is the classical errors-in-variables effect, where treating $X^{(\delta)}$ as noise-free induces attenuation and loss of identifiability in high-noise regimes. Note that $\tilde{\varepsilon}$ is not independent of $X^{(\delta)}$,
so this is not a standard homoscedastic regression model.
Proposition~\ref{prop:ridge} gives the complementary, mean-level view of the same effect: averaged over the noise, the perturbed objective carries an extra ridge term $\frac{\nu^2}{2}\|\beta\|_2^2$. Both effects (inflated noise variance and ridge shrinkage) reduce the selection of the relevant variables under a fixed $\lambda$ as the noise level grows; irrelevant variables need not follow (see the discussion after Theorem~\ref{thm:errctrl}).
\end{remark}

\begin{remark}[Choosing the noise grid in practice]\label{r:grid}
The grid $\mathcal G$ and its weights are inputs of 2D-SS, and Proposition~\ref{prop:bridge} says what to aim for: somewhat more than half of the total weight on levels at which the single-fit conditions hold. Those conditions involve latent quantities and cannot be verified from data, so we offer only the following practical guidance, which carries no guarantee. (i) \emph{Scale.} The variables are standardised, so $\delta$ is measured in units of their standard deviation and grids are comparable across data sets. (ii) \emph{Lower end.} The path cannot start below $\nu_0=\delta_{\mathrm{obs}}$, which is unknown; we start at $\delta=0$ or at a small positive level. (iii) \emph{Upper end.} A natural choice for the largest level is the one at which the average selection frequency of the base selector, at the fixed $\lambda$, has dropped well below its maximum, because levels beyond it contribute little to $\widehat{\bar f}_j$ (in Figure~\ref{fig:ic_violated} the frequencies of the strong relevant variable falls from $1$ at $\delta=0.75$ to $0.39$ at $\delta=3$). (iv) \emph{Spacing and weights.} We use a small number of levels with uniform weights (weights that decrease with $\delta$, see above, are an alternative; the Appendix subsection ``Sensitivity to the noise grid and the weights'' shows how the inclusion rate and $F_1$ depend on the upper end, the number of levels and the weights); our studies use $m=5$ levels in $[0,1]$ (Studies~2--4), $m=10$ levels in $[0.05,2.5]$ (the rat data) and $m=14$ levels in $[0,3]$, denser near $0$ (Figure~\ref{fig:ic_violated}).
\end{remark}

\begin{remark}[Probability in Proposition~\ref{prop:bridge}]\label{r:markov}
The factor $1/\kappa$ in Proposition~\ref{prop:bridge} comes from Markov's inequality applied to the data-dependent failure rate; in the reverse direction the loss is linear. For $\varphi_0\in[0,\min\{1-\pi,\pi\})$ and $\gamma_0\in(0,\pi-\varphi_0)$, Assumption~\ref{as:avg_sep} holds with margin at least $\gamma_0$ with probability at least
\[
1-\frac{2\pi\bar\alpha}{\pi-\varphi_0-\gamma_0}-\mathbb P(\bar\varphi>\varphi_0).
\]
Indeed, $\mathbb E(Q)\le\pi\bar\alpha$ (notation of the proof of Proposition~\ref{prop:bridge}), so $Q\le(\pi-\varphi_0-\gamma_0)/2$ except with probability at most $2\pi\bar\alpha/(\pi-\varphi_0-\gamma_0)$, and then $\pi-2Q-\bar\varphi\ge\gamma_0$ on $\{\bar\varphi\le\varphi_0\}$. The failure probability is thus linear in $\bar\alpha$. Moreover $\bar\alpha$ enters the sample-size conditions only logarithmically, through $L=\log(8ps/\alpha)$ and $\log(2ps/\alpha_I)$ in $u_\Sigma$ (or through $\ell=\log(9p/\alpha)$ under Gaussian design), plus the term $e^{-n_1/8}$, which is negligible for moderate $n_1$. Reducing the failure probability by a factor $10$ therefore increases these logarithms by $\log10\approx2.3$, a mild change in $\omega$.
\end{remark}

\begin{remark}[Sharper bounds under shape restrictions]\label{r:unimodal}
The bound of Theorem~\ref{thm:errctrl} is a worst-case bound obtained from Markov's inequality applied to {the $\delta$-averaged} simultaneous-selection frequency $T_k$, and can be loose. {If, at a single level $\delta$ (that is, for $\mathcal G=\{\delta\}$), the distribution of $B^{-1}\sum_{b=1}^Bt_{b,\delta}$ is unimodal, or more generally $r$-concave, the sharper tail bounds of \citet[Sections~3.2--3.3]{shah2013variable} apply in place of Markov's inequality; whether they extend to the $\delta$-averaged frequency is left open, since a mixture of unimodal distributions need not be unimodal. A different route to sharper bounds integrates, rather than maximises, a transformation of the selection probabilities along the regularisation path \citep{melikechi2026integrated}; a conditional analogue along the noise path is given in Remark~\ref{r:ipss}.}
\end{remark}

\begin{remark}[Tighter bounds under a multi-pair condition]\label{r:ipss}
Theorem~\ref{thm:errctrl} uses the simultaneous-selection frequency only through its first moment. A sharper bound is available for a modified selection rule if more is assumed, in the manner of integrated path stability selection \citep{melikechi2026integrated}, here along the noise path. Assume exchangeability as in Theorem~\ref{thm:errctrl}, with the $B$ pairs drawn independently of each other and of the data (from the same observed sample), write $\rho_\delta\coloneq\bar q_\delta/(p-s)$ for the common probability that an irrelevant variable is selected at level $\delta$, and let $\hat\pi_k(\delta)\coloneq\frac1{2B}\sum_{b=1}^B\big(1\{k\in\hat S(A_{2b-1};\delta)\}+1\{k\in\hat S(A_{2b};\delta)\}\big)$. For $M\in\mathbb N$ and $\tau\in(0,1]$ let
\[
\hat S^{(M)}_\tau\coloneq\Big\{k:\ \sum_{\delta\in\mathcal G}w_\delta\,\big(2\hat\pi_k(\delta)-1\big)_+^M\ \ge\ \tau\Big\}.
\]
{Given the data $\mathcal D$ the pairs are i.i.d.; let $g_k(\delta)\coloneq\mathbb P\big(k\in\hat S(A_1;\delta)\cap\hat S(A_2;\delta)\mid\mathcal D\big)$ and $m_j(\delta)\coloneq\max_{k\in S^c}\mathbb E\,[g_k(\delta)^j]$ for $j=1,\dots,M$.} Then
\[
\mathbb E\,\big|\hat S^{(M)}_\tau\cap S^c\big|\ \le\ \frac{p-s}{\tau}\sum_{\delta\in\mathcal G}w_\delta\,B^{-M}\sum_{j=1}^{M}\mathrm{St}(M,j)\,B(B-1)\cdots(B-j+1)\,m_j(\delta),
\]
where $\mathrm{St}(M,j)$ are the Stirling numbers of the second kind. To see this, write $t_{b,\delta}\coloneq1\{k\in\hat S(A_{2b-1};\delta)\cap\hat S(A_{2b};\delta)\}$ and $\hat T_k(\delta)\coloneq B^{-1}\sum_bt_{b,\delta}$. The inequality $a\le(1+t)/2$ used in the proof of Theorem~\ref{thm:errctrl} gives $2\hat\pi_k(\delta)-1\le\hat T_k(\delta)$, hence $(2\hat\pi_k(\delta)-1)_+^M\le\hat T_k(\delta)^M$, and Markov's inequality gives $\mathbb P(k\in\hat S^{(M)}_\tau)\le\tau^{-1}\sum_\delta w_\delta\,\mathbb E\,\hat T_k(\delta)^M$. Expanding $\hat T_k(\delta)^M=B^{-M}\sum_{(b_1,\dots,b_M)}\prod_it_{b_i,\delta}$ and using $t^2=t$, each product is the indicator that $k$ lies in the intersection over the $j$ distinct pairs involved; {given $\mathcal D$ the pairs are i.i.d., so its conditional probability is $g_k(\delta)^j$ and its expectation is at most $m_j(\delta)$,} and exactly $\mathrm{St}(M,j)B(B-1)\cdots(B-j+1)$ of the $B^M$ tuples involve $j$ distinct pairs. Summing over the $p-s$ irrelevant variables gives the claim.

{For $M=1$, $m_1(\delta)=\rho_\delta^2$ by the independence of the two halves used in the proof of Theorem~\ref{thm:errctrl}, and the bound is $\sum_\delta w_\delta\bar q_\delta^{\,2}/((p-s)\tau)$. For $j\ge2$, Jensen's inequality gives $m_j(\delta)\ge\rho_\delta^{2j}$, with equality only if $g_k(\delta)$ is almost surely constant, that is, only if the simultaneous-selection probability does not depend on the data. The values below, which replace $m_j(\delta)$ by $\rho_\delta^{2j}$, are therefore best-case values. For $M=2$ they equal $\frac{p-s}{\tau}\sum_\delta w_\delta\big(\rho_\delta^2/B+(B-1)\rho_\delta^4/B\big)$, which tends to $\tau^{-1}\sum_\delta w_\delta\bar q_\delta^{\,4}/(p-s)^3$ as $B\to\infty$ (for $B\gg\rho_\delta^{-2}$); for $B\ll\rho_\delta^{-2}$ the best-case $M=2$ bound is about $1/B$ times the $M=1$ bound, and it approaches that limit only for $B\gg\rho_\delta^{-2}$. For illustration, with $\bar q_\delta=q_\star=5$, $p-s=995$, $B=50$ and $\tau=0.2$, the best-case bound is $0.126$ for $M=1$ and $0.0025$ for $M=2$ (these bound different selection rules).

The gain has a price. First, independence of the observations and exchangeability fix only $m_1(\delta)$; for $M\ge2$ the bound needs a bound on the higher moments $m_j(\delta)$ of the data-conditional simultaneous-selection probability, which must be assumed or estimated and is at least $\rho_\delta^{2j}$. Second, the bound is for the selection rule $\hat S^{(M)}_\tau$, not for the threshold rule of Theorem~\ref{thm:errctrl} or the largest-gap rule of Algorithm~\ref{alg:delta_averaged_selection}.}
\end{remark}

\section*{Beyond the Lasso: other base selectors (Section~\ref{s2:other})}
Algorithm~\ref{alg:delta_averaged_selection} is defined for any base selector, and the mechanism of Section~\ref{s2:small} is not specific to the Lasso. The analysis above uses three ingredients only: the stationarity conditions of the penalised least-squares problem, the decomposition \eqref{eq:gram} of the Gram matrix, and the decomposition $r=\Sigma^{(\delta)}\beta-\nu^2\beta+\xi$ of the response remainder (proof of Theorem~\ref{thm:single_level}). A different penalty changes only which vector is evaluated at the shifted population solution. We use the notation of that proof, $r=\frac1{n_1}(X_I+W)^\top Y_I$ and $E=\widetilde\Delta_{SS}$, and the constants $\omega$ and $\rho_\nu$ defined there, which depend on the actual noise level $\nu$ only.
\begin{proposition}[Other penalised least-squares base selectors]\label{prop:other}
Assume the setting of Theorem~\ref{thm:single_level} and fix $\eta\in(0,1)$.
\begin{enumerate}
\item[(a)] \emph{Elastic net.} Let the base selector minimise $\frac1{2n_1}\|Y_I-(X_I+W)\beta\|_2^2+\lambda_1\|\beta\|_1+\frac{\lambda_2}2\|\beta\|_2^2$ and put $\vartheta^2\coloneq\nu^2+\lambda_2$. Then Lemma~\ref{l1} and Theorem~\ref{thm:single_level} hold with $\nu$ replaced by $\vartheta$ in $T$, $v$, $\kappa$, $\omega_0$, $\beta^\star$, $\theta$, $\Gamma$, $g$ and $\Psi$, and $\lambda$ replaced by $\lambda_1$, while $\omega$ and $\rho_\nu$ keep the actual noise level $\nu$. In particular the relevant condition is $\big\|\Sigma_{S^cS}(\Sigma_{SS}+\vartheta^2I_s)^{-1}\big(z+\tfrac{\vartheta^2}{\lambda_1}\beta_S\big)\big\|_\infty\le1-\eta$.
\item[(b)] \emph{Weighted $\ell_1$.} Let the penalty be $\lambda\sum_jw_j|\beta_j|$ with fixed weights $w_j>0$ (for instance adaptive-lasso weights \citep{zou2006adaptive} computed from an independent sample), and put $w_{\min}\coloneq\min_{k\in S^c}w_k$, $T\coloneq\Sigma_{SS}+\nu^2I_s$ and $u\coloneq w_S\circ z+\frac{\nu^2}\lambda\beta_S$. If $\max_{k\in S^c}\big|[\Sigma_{S^cS}T^{-1}u]_k\big|/w_k\le1-\eta$, then Theorem~\ref{thm:single_level} holds with $z$ replaced by $w_S\circ z$ in $v$, $\beta^\star$ and $g$, with $\|v\|_\infty$ replaced by $\|u\|_\infty$, and with the second term of $\omega_0$ and the right-hand side of (iii) multiplied by $w_{\min}$ (the first term of $\omega_0$, which ensures invertibility, is unchanged; the weights are treated as fixed, so the statement is conditional on $w$ and $w$ must be independent of $(I,\varepsilon,W_0,W)$; in the uniform constant $\omega_0^*$ the second term is multiplied by $w_{\min}$ and $V$ becomes $\sqrt s(\|w_S\|_\infty/\lambda_{\min}(\Sigma_{SS})+\|\beta_S\|_\infty/\lambda)$).
\item[(c)] \emph{Folded-concave penalties (SCAD, MCP).} Let the penalty be of the SCAD type \citep[Section~2.1]{Fan01122001} or of the MCP type \citep[Section~2.1]{MCP}, that is, with derivative $p_\lambda'(0+)=\lambda$ and $p_\lambda'(t)=0$ for $t\ge a\lambda$ (here $a=3.7$ is the value recommended by Fan and Li for SCAD, and $a$ is the parameter $\gamma$ of the MCP). Put $T\coloneq\Sigma_{SS}+\nu^2I_s$, $\beta^\circ_\nu\coloneq T^{-1}\Sigma_{SS}\beta_S$, $\theta^\circ_\nu\coloneq\min_{j\in S}z_j\beta^\circ_{\nu,j}$ and $\Psi^\circ_\nu\coloneq2\kappa_\nu\big(\kappa_\nu\|\Sigma_{SS}\beta_S\|_\infty\omega+\omega\|\beta_S\|_\infty+\rho_\nu\big)$, and suppose that (i$^\circ$) $\frac{\nu^2}{\lambda}\|\Sigma_{S^cS}T^{-1}\beta_S\|_\infty\le1-\eta$, (ii) $\omega\le\omega_0(\nu)$ with $v$ replaced by $\frac{\nu^2}\lambda\beta_S$ (at $\nu=0$ the second term of $\omega_0$ is read as $+\infty$), (iii) $\rho_\nu\Gamma_\nu<\eta\lambda/2$ and (iv$^\circ$) $\theta^\circ_\nu>a\lambda+\Psi^\circ_\nu$. Then, with probability at least $1-\bar\alpha$, the jittered problem has a first-order stationary point with support $S$ and sign vector $z$. {(The superscript $^\circ$ marks the quantities that change because the folded-concave penalty has zero derivative at $|\beta_j|\ge a\lambda$, so the sign term $z$ is absent and only the ridge-bias term $\frac{\nu^2}{\lambda}\beta_S$ remains.)}
\end{enumerate}
\end{proposition}
\begin{proof}
Each part is the proof of Theorem~\ref{thm:single_level} with the stationarity equations changed. (a) Since $\Sigma^{(\delta)}_{SS}+\lambda_2I_s=T_\vartheta+E$ with $T_\vartheta=\Sigma_{SS}+\vartheta^2I_s$ and $r_S=(\Sigma_{SS}+E)\beta_S+\xi_S$, the witness is $\hat\beta_S=(T_\vartheta+E)^{-1}[(\Sigma_{SS}+E)\beta_S+\xi_S-\lambda_1z]$ and $\hat z_{S^c}=\Sigma^{(\delta)}_{S^cS}(T_\vartheta+E)^{-1}v_\vartheta+\lambda_1^{-1}[\xi_{S^c}-\Sigma^{(\delta)}_{S^cS}(T_\vartheta+E)^{-1}\xi_S]$ with $v_\vartheta=z+\frac{\vartheta^2}{\lambda_1}\beta_S$; these are the expressions of that proof with $\nu$ replaced by $\vartheta$ in the deterministic part, while $\xi$ and $\widetilde\Delta$ are unchanged because they depend on the noise only. (b) Stationarity on $S$ gives $\hat\beta_S=(T+E)^{-1}[(\Sigma_{SS}+E)\beta_S+\xi_S-\lambda\,w_S\circ z]$, and dual feasibility is $|\hat z_k|<w_k$ for $k\in S^c$, where $\hat z_{S^c}=\Sigma^{(\delta)}_{S^cS}(T+E)^{-1}u+\lambda^{-1}[\xi_{S^c}-\Sigma^{(\delta)}_{S^cS}(T+E)^{-1}\xi_S]$; dividing by $w_k\ge w_{\min}$ in the bounds of Lemma~\ref{l1} (second term of $\omega_0$) and of Step~(II) (right-hand side of (iii)) gives the stated factors. (c) Let $\hat\beta_S\coloneq(\Sigma^{(\delta)}_{SS})^{-1}r_S$ and $\hat\beta_{S^c}=0$. Then $\hat\beta_S-\beta^\circ_\nu=[(T+E)^{-1}-T^{-1}]\Sigma_{SS}\beta_S+(T+E)^{-1}(E\beta_S+\xi_S)$, so $\|\hat\beta_S-\beta^\circ_\nu\|_\infty\le\Psi^\circ_\nu$ and, by (iv$^\circ$), $z_j\hat\beta_j>a\lambda$ for $j\in S$; there the penalty derivative vanishes and stationarity on $S$ holds. On $S^c$, $r_{S^c}-\Sigma^{(\delta)}_{S^cS}\hat\beta_S=\Sigma^{(\delta)}_{S^cS}(T+E)^{-1}(\nu^2\beta_S-\xi_S)+\xi_{S^c}$, whose sup-norm is below $\lambda$ by Lemma~\ref{l1} with $v=\frac{\nu^2}\lambda\beta_S$, (i$^\circ$) and (iii), exactly as in Step~(II); this is the subgradient condition at $0$.
\end{proof}
\begin{remark}[What changes with the penalty]\label{r:other}
In (a) the injected noise and the penalty's own ridge weight enter the reference matrix only through $\vartheta^2=\nu^2+\lambda_2$: the noise path cannot start below $\vartheta_0^2=\delta_{\mathrm{obs}}^2+\lambda_2$, and for $\Sigma_{SS}=I_s$ Proposition~\ref{prop:repair} applies with $\nu^2$ replaced by $\vartheta^2$, so the zero of the statistic lies at $\vartheta_k^{\star2}=-A_k\lambda_1/B_k$ and is reachable by jitter only if $\vartheta_k^{\star2}\ge\delta_{\mathrm{obs}}^2+\lambda_2$ ($\ge\lambda_2$): a large $\lambda_2$ can overshoot the repair point. In (b) the sign vector is weighted and the condition is scaled by $w_k$, which is why weights that are large on irrelevant variables relax it. In (c) the sign term $z$ is absent: no irrepresentable condition is needed at $\nu=0$, because the left-hand side of (i$^\circ$) vanishes there, but the attenuated coefficients $\beta^\circ_\nu$ must stay above $a\lambda$ and the leakage term grows with $\nu^2$; for $\Sigma_{SS}=I_s$ it equals $\frac{\nu^2}{1+\nu^2}|B_k|/\lambda$, the same marginal-correlation limit as for the Lasso. For the design of Example~\ref{ex:jy} with $\lambda=0.05$, $a=3.7$ and $\eta=0.1$, (i$^\circ$) reads $\frac{2\nu^2}{1+\nu^2}\le0.9$, that is $\nu\le0.90$, and $\min_j|\beta_j|/(1+\nu^2)>a\lambda=0.185$ holds for $\nu^2<1.70$; the population conditions (i$^\circ$) and (iv$^\circ$) therefore hold for $\nu\in[0,0.90]$, whereas the Lasso irrepresentable condition, with the same margin, holds for $\nu\in[0.279,1.356]$. {Theorem~\ref{thm:errctrl} requires a base selector that uses only the half-sample with fixed tuning (for (b), conditionally on the weights; its closed form also needs exchangeability, which heterogeneous weights break), and Proposition~\ref{prop:bridge} needs only the single-level guarantee at the good levels, so the aggregation results apply to (a) and (b) with these provisos, and to (c) only when the algorithm returns the stationary point of Proposition~\ref{prop:other}(c).}
\end{remark}
}

\section*{Proof of Proposition~\ref{prop:repair}}
\begin{proof}
(a) Since $\Sigma_{SS}^{-1}-T_\nu^{-1}=\Sigma_{SS}^{-1}(T_\nu-\Sigma_{SS})T_\nu^{-1}=\nu^2\Sigma_{SS}^{-1}T_\nu^{-1}$, we have $T_\nu^{-1}=\Sigma_{SS}^{-1}-\nu^2\Sigma_{SS}^{-1}T_\nu^{-1}$. Hence $a_kT_\nu^{-1}v_\nu=a_k\Sigma_{SS}^{-1}v_\nu-\nu^2a_k\Sigma_{SS}^{-1}T_\nu^{-1}v_\nu$, and $a_k\Sigma_{SS}^{-1}v_\nu=a_k\Sigma_{SS}^{-1}z+\frac{\nu^2}\lambda a_k\Sigma_{SS}^{-1}\beta_S=A_k+\frac{\nu^2}{\lambda}B_k$.

(b) If $\Sigma_{SS}=I_s$ then $T_\nu=(1+\nu^2)I_s$, $A_k=a_kz$ and $B_k=a_k\beta_S$, so $c_\nu^k=a_k(z+\frac{\nu^2}\lambda\beta_S)/(1+\nu^2)$, which is the stated convex combination. Write $\mathsf a=A_k$, $\mathsf b=B_k/\lambda$ and $t=\nu^2$, so that $c(t)=\mathsf b+(\mathsf a-\mathsf b)/(1+t)$ is monotone in $t$ and its range over $t\in[0,\infty)$ is the half-open segment between $\mathsf a$ (included) and $\mathsf b$ (excluded). This segment meets $[-m,m]$ iff $\mathsf a\in[-m,m]$, or $\mathsf a>m$ and $\mathsf b<m$, or $\mathsf a<-m$ and $\mathsf b>-m$. For $\mathsf a>m$ and $\mathsf b<m$, the inequality $c(t)\le m$ is equivalent to $t(m-\mathsf b)\ge\mathsf a-m$, and $c(t)\ge-m$ is equivalent to $t(\mathsf b+m)\ge-(\mathsf a+m)$, which holds for all $t$ if $\mathsf b\ge-m$ and is equivalent to $t\le(\mathsf a+m)/(-\mathsf b-m)$ otherwise.{

(c) If $A_kB_k<0$ then $\mathsf a\ne\mathsf b$, so $c(t)=\mathsf b+(\mathsf a-\mathsf b)/(1+t)$ is strictly monotone, and $c(t)=0$ is equivalent to $\mathsf a+t\mathsf b=0$, that is, $t=-\mathsf a/\mathsf b=-A_k\lambda/B_k>0$.}
\end{proof}
\section*{Proof of Lemma~\ref{l1}}
\begin{proof}
Let $T=\Sigma_{SS}+\nu^2I_s$, $E=\Delta_{SS}$, $D=\Delta_{S^cS}$ and $v=v_\nu$, so $\kappa_\nu=\|T^{-1}\|_\infty$, $\|E\|_\infty\le\omega$ and $\|D\|_\infty\le\omega$. The off-diagonal block of $\Sigma_\nu=\Sigma+\nu^2I_p$ equals $\Sigma_{S^cS}$ because $S\cap S^c=\emptyset$, hence $\Sigma'_{SS}=T+E$ and $\Sigma'_{S^cS}=\Sigma_{S^cS}+D$.

\textbf{Step 1.} If $\omega\le1/(2\kappa_\nu)$ then $\|T^{-1}\|_\infty\|E\|_\infty\le\frac12$, so $T+E$ is invertible and
$\|(T+E)^{-1}\|_\infty\le2\kappa_\nu$. From $(T+E)^{-1}-T^{-1}=-(T+E)^{-1}ET^{-1}$ we get $\|(T+E)^{-1}-T^{-1}\|_\infty\le2\kappa_\nu^2\omega$.

\textbf{Step 2.} Decompose
\[
(\Sigma_{S^cS}+D)(T+E)^{-1}v-\Sigma_{S^cS}T^{-1}v=D(T+E)^{-1}v+\Sigma_{S^cS}\big[(T+E)^{-1}-T^{-1}\big]v .
\]
Using $\|Mv\|_\infty\le\|M\|_\infty\|v\|_\infty$ and Step 1, the sup-norm of the right-hand side is at most
\[
\begin{aligned}
&\omega\cdot2\kappa_\nu\|v\|_\infty+\|\Sigma_{S^cS}\|_\infty\,2\kappa_\nu^2\omega\|v\|_\infty\\
&\qquad=2\kappa_\nu\|v\|_\infty\big(1+\kappa_\nu\|\Sigma_{S^cS}\|_\infty\big)\omega\le\frac\eta2,
\end{aligned}
\]
where the last inequality is the second term in the definition of $\omega_0(\nu)$.

\textbf{Step 3.} By EIC$_\nu$, $\|\Sigma_{S^cS}T^{-1}v\|_\infty\le1-\eta$. The triangle inequality and Step 2 give
$\|\Sigma'_{S^cS}(\Sigma'_{SS})^{-1}v_\nu\|_\infty\le(1-\eta)+\eta/2=1-\eta/2$.

\textbf{Step 4 (uniformity in $\nu$).} For a symmetric positive definite $s\times s$ matrix $T$, $\|T^{-1}\|_\infty\le\sqrt s\,\|T^{-1}\|_2$, and $\|T^{-1}\|_2=1/(\lambda_{\min}(\Sigma_{SS})+\nu^2)$. Hence $\kappa_\nu\le\sqrt s/(\lambda_{\min}(\Sigma_{SS})+\nu^2)\le\bar\kappa$. Moreover $\|v_\nu\|_\infty\le1+\nu^2\|\beta_S\|_\infty/\lambda$, so
\[
\begin{aligned}
\kappa_\nu\|v_\nu\|_\infty&\le\sqrt s\,\frac{1+\nu^2\|\beta_S\|_\infty/\lambda}{\lambda_{\min}(\Sigma_{SS})+\nu^2}\\
&\le\sqrt s\Big(\frac1{\lambda_{\min}(\Sigma_{SS})}+\frac{\|\beta_S\|_\infty}\lambda\Big)=V,
\end{aligned}
\]
using $\frac1{\lambda_{\min}+\nu^2}\le\frac1{\lambda_{\min}}$ and $\frac{\nu^2}{\lambda_{\min}+\nu^2}\le1$. Therefore $\omega_0(\nu)\ge\omega_0^*$. No restriction such as $\nu\le1$ is needed, since no quadratic term is absorbed into a linear one.
\end{proof}
\section*{Proof of Theorem~\ref{thr:hp_gauss_final}}
\begin{proof}
Let $L=\log(8ps/\alpha)$ and $t_1,t_2,u_\Sigma,\omega$ as in the statement, and write $\widetilde\Delta=\Delta_0+\Delta_1+\Delta_2$ with
\[
\Delta_0\coloneq\Sigma_I-\Sigma,\quad
\Delta_1\coloneq\frac{1}{n_1}(X_{0,I}^\top \widetilde W + \widetilde W^\top X_{0,I}),\quad
\Delta_2\coloneq\frac{1}{n_1}\widetilde W^\top \widetilde W - \nu^2 I_p .
\]
Since $I$ is independent of $(W_0,W)$, conditionally on $I$ the matrix $\widetilde W=W_{0,I}+W$ has i.i.d.\ $\mathcal N(0,\nu^2)$ entries. Lemma~\ref{l1} involves only the blocks $\widetilde\Delta_{SS}$ and $\widetilde\Delta_{S^cS}$, whose entries are the entries $(j,k)$ of $\widetilde\Delta$ with $j\in\{1,\dots,p\}$ and $k\in S$; there are $ps$ of them.

\textbf{Step 0: bound $\Delta_0$ (subsampling).} For $j\in[p]$, $k\in S$, $(\Sigma_I)_{jk}$ is the mean of a sample of size $n_1$ drawn without replacement from the finite population $\{X_{0,ij}X_{0,ik}\}_{i=1}^n\subset[-K^2,K^2]$, whose mean is $\Sigma_{jk}$. By Hoeffding's inequality, which also holds for sampling without replacement \citep[Theorem~2 and Section~6]{409cf137-dbb5-3eb1-8cfe-0743c3dc925f}, $\mathbb P(|(\Delta_0)_{jk}|>u)\le2\exp(-n_1u^2/(2K^4))$. A union bound over the $ps$ entries with $u=u_\Sigma$ gives $\max_{j\in[p],k\in S}|(\Delta_0)_{jk}|\le u_\Sigma$ with probability at least $1-\alpha_I$. Since each row of $(\Delta_0)_{SS}$ and $(\Delta_0)_{S^cS}$ has $s$ entries, both blocks have $\|\cdot\|_\infty\le s\,u_\Sigma$ on this event.

\textbf{Step 1: bound $\Delta_1$.}
Conditionally on $I$, for each $(j,k)$, $\big(\frac1{n_1}X_{0,I}^\top \widetilde W\big)_{jk}\sim\mathcal N\big(0,\frac{\nu^2}{n_1^2}\sum_{i\in I}X_{0,ij}^2\big)$, with variance at most $\nu^2K^2/n_1$. By the Gaussian tail bound,
\[
\mathbb P\Big(\Big|\Big(\tfrac1{n_1}X_{0,I}^\top \widetilde W\Big)_{jk}\Big|>\tfrac{t_1}{2}\Big)\le2\exp\Big(-\frac{n_1(t_1/2)^2}{2\nu^2K^2}\Big)=2e^{-L}.
\]
Each entry $(\Delta_1)_{jk}$ with $j\in[p]$, $k\in S$ is the sum of two such terms, $\frac1{n_1}X_{0,j}^\top \widetilde W_k$ and $\frac1{n_1}X_{0,k}^\top \widetilde W_j$, so there are at most $2ps$ distinct terms. A union bound shows that all of them are at most $t_1/2$ in absolute value (hence $|(\Delta_1)_{jk}|\le t_1$) with probability at least $1-4ps\,e^{-L}=1-\alpha/2$. Since each row of $(\Delta_1)_{SS}$ and $(\Delta_1)_{S^cS}$ has $s$ entries, both blocks have $\|\cdot\|_\infty\le s\,t_1$ on this event.

\textbf{Step 2: bound $\Delta_2$.}
The entries of $\Delta_2$ are $\nu^2$ times averages over $i\in I$ of independent variables: on the diagonal $\frac1{n_1}\sum_i(Z_i^2-1)$, and off the diagonal $\frac1{n_1}\sum_iZ_iZ_i'$, with $Z_i,Z_i'$ independent standard normal. For the off-diagonal case, $Z_iZ_i'=(U_i^2-V_i^2)/2$ with $U_i=(Z_i+Z_i')/\sqrt2$ and $V_i=(Z_i-Z_i')/\sqrt2$ i.i.d.\ standard normal, so the average equals $\frac12(A-B)$, where $A$ and $B$ are averages of $U_i^2-1$ and $V_i^2-1$. By the $\chi^2$ tail bounds of \citet[Lemma~1]{laurent2000adaptive}, for every $x>0$,
\[
\mathbb P\Big(\Big|\frac1{n_1}\sum_{i}(U_i^2-1)\Big|>2\sqrt{\frac{x}{n_1}}+\frac{2x}{n_1}\Big)\le2e^{-x}.
\]
Hence, with $x=L+1$, every diagonal entry of $\Delta_2$ exceeds $t_2$ in absolute value with probability at most $2e^{-x}\le2e^{-L}$, and every off-diagonal entry with probability at most $4e^{-x}<2e^{-L}$. A union bound over the $ps$ entries gives
\[
\mathbb{P}\!\left(\max_{j\in[p],k\in S}|(\Delta_2)_{jk}|>t_2\right)
\le 2ps\, e^{-L}
= \frac{\alpha}{4}.
\]
Thus both blocks of $\Delta_2$ have $\|\cdot\|_\infty\le s\,t_2$ with probability at least $1-\alpha/4$.

\textbf{Step 3: combine.}
By a union bound over Steps 0--2, with probability at least $1-\alpha_I-3\alpha/4\ge1-\alpha-\alpha_I$,
\[
\max\{\|\widetilde\Delta_{SS}\|_\infty,\|\widetilde\Delta_{S^cS}\|_\infty\}\le s(t_1+t_2+u_\Sigma)=\omega .
\]

\textbf{Step 4: deterministic step.}
On this event and with $\omega\le\omega_0(\nu)$, Lemma~\ref{l1} (with $\Sigma'=\Sigma^{(\delta)}$ and $\Delta=\widetilde\Delta$) gives
$\|\Sigma^{(\delta)}_{S^cS}(\Sigma^{(\delta)}_{SS})^{-1}v_\nu\|_\infty\le1-\eta/2$, so the stated bound holds with probability at least $1-\alpha-\alpha_I$.
\end{proof}
\section*{Proof of Theorem~\ref{thm:single_level}}
\begin{proof}
In this proof $E$ and $\xi$ are local symbols, with a different meaning from Section~\ref{s2:single} and the proof of Theorem~\ref{thm:gauss_rate}. Write $T=\Sigma_{SS}+\nu^2I_s$, $E=\widetilde\Delta_{SS}$, $z=\operatorname{sign}(\beta_S)$, $r\coloneq\frac1{n_1}(X_I+W)^\top Y_I$, and $Y_I=X_{0,I}\beta+\varepsilon_I$, so that $X_I+W=X_{0,I}+\widetilde W$.

\textbf{Step 1: decomposition of $r$.} A direct expansion gives
\[
\begin{gathered}
r=\Sigma^{(\delta)}\beta-\nu^2\beta+\xi,\\
\xi\coloneq\frac1{n_1}\Big[X_{0,I}^\top\varepsilon_I+\widetilde W^\top\varepsilon_I-X_{0,I}^\top \widetilde W\beta-\big(\widetilde W^\top \widetilde W-n_1\nu^2I_p\big)\beta\Big].
\end{gathered}
\]
Since $\beta_{S^c}=0$ and $\Sigma^{(\delta)}_{SS}=T+E$, $r_S=(\Sigma_{SS}+E)\beta_S+\xi_S$ and $r_{S^c}=\Sigma^{(\delta)}_{S^cS}\beta_S+\xi_{S^c}$.

\textbf{Step 2: control of $\xi$.} Consider the event $\mathcal E_1$ of Theorem~\ref{thr:hp_gauss_final} (probability at least $1-\alpha-\alpha_I$), which, by Steps 1--2 of its proof, includes $|(X_{0,I}^\top \widetilde W)_{jk}|/n_1\le t_1/2$ and $|(\Delta_2)_{jk}|\le t_2$ for all $j\in[p]$, $k\in S$. Because $\beta$ is supported on $S$, on $\mathcal E_1$
\[
\frac{|(X_{0,I}^\top \widetilde W\beta)_j|}{n_1}\le\|\beta_S\|_1\frac{t_1}2,\qquad \big|\big((\widetilde W^\top \widetilde W/n_1-\nu^2I)\beta\big)_j\big|\le\|\beta_S\|_1t_2 .
\]
Next, conditionally on $I$, $X_{0,j,I}^\top\varepsilon_I/n_1$ is Gaussian with variance at most $\sigma^2K^2/n_1$, so with probability at least $1-2p\,e^{-L_p}=1-\alpha/2$,
\[
\max_j\frac{|X_{0,j,I}^\top\varepsilon_I|}{n_1}\le\sigma K\sqrt{\frac{2L_p}{n_1}} .
\] Further, $\|\varepsilon_I\|_2^2/\sigma^2\sim\chi^2_{n_1}$ exceeds $2n_1$ with probability at most $e^{-n_1/8}$ by the standard $\chi^2$ tail bound. Given $\|\varepsilon_I\|_2\le\sigma\sqrt{2n_1}$, and since $\widetilde W$ is independent of $\varepsilon$, the entries $\widetilde W_j^\top\varepsilon_I$ are $\mathcal N(0,\nu^2\|\varepsilon_I\|_2^2)$, so with probability at least $1-2p\,e^{-L_p}=1-\alpha/2$,
\[
\max_j\frac{|\widetilde W_j^\top\varepsilon_I|}{n_1}\le\frac{\nu\|\varepsilon_I\|_2\sqrt{2L_p}}{n_1}\le2\sigma\nu\sqrt{\frac{L_p}{n_1}} .
\] Hence, with probability at least $1-\alpha_I-2\alpha-e^{-n_1/8}=1-\bar\alpha$ over $(I,\varepsilon,W_0,W)$, we have $\|\xi\|_\infty\le\rho_\nu$ and $\max\{\|\widetilde\Delta_{SS}\|_\infty,\|\widetilde\Delta_{S^cS}\|_\infty\}\le\omega$.

\textbf{Step 3: primal--dual witness.} On this event, set $\hat\beta_{S^c}=0$, $\hat z_S=z$ and $\hat\beta_S=(T+E)^{-1}(r_S-\lambda z)$, which is well defined because $\omega\le\omega_0(\nu)\le1/(2\kappa_\nu)$ (Lemma~\ref{l1}, Step 1). Define $\hat z_{S^c}=(r_{S^c}-\Sigma^{(\delta)}_{S^cS}\hat\beta_S)/\lambda$. By construction the stationarity conditions $r-\Sigma^{(\delta)}\hat\beta=\lambda\hat z$ of the Lasso hold; it remains to verify (I) $\operatorname{sign}(\hat\beta_S)=z$ and (II) $\|\hat z_{S^c}\|_\infty<1$. Together with invertibility of $\Sigma^{(\delta)}_{SS}$, they imply that $\hat\beta$ is the unique Lasso solution and has support $S$ \citep[Lemma~2]{wainwright2009sharp}; see also \citet[Lemma~1]{zhao2006model} for the KKT conditions. For the elastic net on a fixed design, the same KKT argument is given in the proof of Lemma~1 of \citet{jia2010model}; here it is applied to the perturbed design.

\textbf{(II) Dual feasibility.} Using Step 1 and $\Sigma^{(\delta)}_{SS}=T+E$,
\[
\begin{aligned}
\hat z_{S^c}&=\Sigma^{(\delta)}_{S^cS}(T+E)^{-1}v_\nu\\
&\quad+\frac1\lambda\Big[\xi_{S^c}-\Sigma^{(\delta)}_{S^cS}(T+E)^{-1}\xi_S\Big].
\end{aligned}
\]
The first term has sup-norm at most $1-\eta/2$ by Lemma~\ref{l1}. Since $\|\Sigma^{(\delta)}_{S^cS}(T+E)^{-1}\|_\infty\le(\|\Sigma_{S^cS}\|_\infty+\omega)\,2\kappa_\nu$, the bracket has norm at most $\Gamma_\nu\rho_\nu$, so $\|\hat z_{S^c}\|_\infty\le1-\eta/2+\Gamma_\nu\rho_\nu/\lambda<1$ by (iii).

\textbf{(I) Sign consistency.} Since $(\Sigma_{SS}+E)\beta_S$ is the first part of $r_S$ (Step 1), $\hat\beta_S=(T+E)^{-1}\big[(\Sigma_{SS}+E)\beta_S+\xi_S-\lambda z\big]$, and therefore
\[
\hat\beta_S-\beta^\star_\nu=\big[(T+E)^{-1}-T^{-1}\big](\Sigma_{SS}\beta_S-\lambda z)+(T+E)^{-1}(E\beta_S+\xi_S),
\]
so $\|\hat\beta_S-\beta^\star_\nu\|_\infty\le2\kappa_\nu^2\omega g_\nu+2\kappa_\nu(\omega\|\beta_S\|_\infty+\rho_\nu)=\Psi_\nu$. Hence $z_j\hat\beta_j\ge z_j\beta^\star_{\nu,j}-\Psi_\nu\ge\theta_\nu-\Psi_\nu>0$ for every $j\in S$ by (iv).
\end{proof}
\section*{Proofs of Lemma~\ref{lem:failure} and Corollary~\ref{cor:benefit}}
\begin{proof}
Work on the event of Step 2 of the proof of Theorem~\ref{thm:single_level}, which has probability at least $1-\bar\alpha$ and on which $\max\{\|\widetilde\Delta_{SS}\|_\infty,\|\widetilde\Delta_{S^cS}\|_\infty\}\le\omega$ and $\|\xi\|_\infty\le\rho_\nu$. Steps 1 and 2 of the proof of Lemma~\ref{l1} do not use EIC$_\nu$; with $\eta_0$ in place of $\eta$ they show that $T+E$ is invertible and that every component of $\Sigma^{(\delta)}_{S^cS}(T+E)^{-1}v_\nu$ differs from the corresponding component of $\Sigma_{S^cS}T^{-1}v_\nu$ by at most $\eta_0/2$. By the expression for $\hat z_{S^c}$ in Step 3 of the proof of Theorem~\ref{thm:single_level}, the remaining term has norm at most $\Gamma_\nu\rho_\nu/\lambda<\eta_0/4$. Choosing $k$ with $|[\Sigma_{S^cS}T^{-1}v_\nu]_k|=c_\nu\ge1+\eta_0$ gives $|\hat z_k|>1+\eta_0-\eta_0/2-\eta_0/4>1$. If a Lasso solution $\hat\beta$ with support $S$ and $\operatorname{sign}(\hat\beta_S)=z$ existed, its stationarity conditions on $S$ would give $\hat\beta_S=(T+E)^{-1}(r_S-\lambda z)$, and its dual vector on $S^c$ would coincide with $\hat z_{S^c}$ and satisfy $\|\hat z_{S^c}\|_\infty\le1$, a contradiction. This proves Lemma~\ref{lem:failure}. Corollary~\ref{cor:benefit} follows by applying Lemma~\ref{lem:failure} at $\nu_0$ and Theorem~\ref{thm:single_level} at $\nu_1$.
\end{proof}
\section*{Proofs of Lemmas~\ref{lem:decomp} and~\ref{lem:gstruct}}
\begin{proof}[Proof of Lemma~\ref{lem:decomp}]
Since $\widehat C$ is invertible, $r_S-\widehat C\hat\beta_S=\lambda z$. Put $\varsigma_k\coloneq\widetilde X_S^\top e_k/n_1$. Because $\widetilde X_k=e_k+\widetilde X_Sc_k$, we have $\widehat\Sigma_{kS}=c_k^\top\widehat C+\varsigma_k^\top$ and $r_k=e_k^\top Y_I/n_1+c_k^\top r_S$. Hence
\[
r_k-\widehat\Sigma_{kS}\hat\beta_S=\frac{e_k^\top Y_I}{n_1}+c_k^\top\big(r_S-\widehat C\hat\beta_S\big)-\varsigma_k^\top\hat\beta_S=\frac{e_k^\top Y_I}{n_1}+\lambda c_k^\top z-\varsigma_k^\top\hat\beta_S .
\]
By the definition of $q$, $Y_I=\widetilde X_SP_\nu\beta_S+q$, so that $e_k^\top Y_I/n_1=\varsigma_k^\top P_\nu\beta_S+e_k^\top q/n_1$ and $r_S=\widehat CP_\nu\beta_S+h$, which gives $\hat\beta_S=P_\nu\beta_S+\Omega(h-\lambda z)$. Substituting, the terms $\varsigma_k^\top P_\nu\beta_S$ cancel and
\[
\begin{aligned}
\hat z_k&=c_k^\top z+\varsigma_k^\top\Omega z+\frac1\lambda\Big(\frac{e_k^\top q}{n_1}-\varsigma_k^\top\Omega h\Big)\\
&=c_k^\top z+e_k^\top\psi .
\end{aligned}
\]
Since $c_k^\top v_\nu=c_k^\top z+\mu_k/\lambda$, and $c_k^\top v_\nu=[\Sigma_{S^cS}T^{-1}v_\nu]_k=c_\nu^k$, this is the first claim. For the second, $\beta^\star_\nu=T^{-1}\Sigma_{SS}\beta_S-\lambda T^{-1}z=P_\nu\beta_S-\lambda T^{-1}z$ ($\Sigma_{SS}$ and $T$ commute), so $\hat\beta_S-\beta^\star_\nu=\Omega h-\lambda\Omega z+\lambda T^{-1}z$, and $\Omega(\widehat C-T)T^{-1}=T^{-1}-\Omega$ gives the stated form.
\end{proof}
\begin{proof}[Proof of Lemma~\ref{lem:gstruct}]
The vectors $(x_{0,iS},\tilde x_{iS},\tilde x_{ik},\varepsilon_i)$ are i.i.d.\ and jointly Gaussian with $\operatorname{Cov}(\tilde x_{iS})=T$, $\operatorname{Cov}(\tilde x_{ik},\tilde x_{iS})=a_k$, $\operatorname{Cov}(\tilde x_{ik},x_{0,iS})=a_k$ and $\operatorname{Cov}(\tilde x_{iS},x_{0,iS})=\Sigma_{SS}$. The residual $e_{ik}=\tilde x_{ik}-c_k^\top\tilde x_{iS}$ has $\operatorname{Cov}(e_{ik},\tilde x_{iS})=a_k-c_k^\top T=0$, and $\operatorname{Var}(e_{ik})=\Sigma_{kk}+\nu^2-a_kT^{-1}a_k^\top$. Let $\varrho_i\coloneq x_{0,iS}-P_\nu\tilde x_{iS}$; then $\operatorname{Cov}(\varrho_i,\tilde x_{iS})=\Sigma_{SS}-P_\nu T=0$ and $\operatorname{Cov}(\varrho_i)=\Sigma_{SS}-\Sigma_{SS}T^{-1}\Sigma_{SS}=\nu^2\Sigma_{SS}T^{-1}$. As $P_\nu$ is symmetric, $q_i=y_i-\tilde x_{iS}^\top P_\nu\beta_S=\varrho_i^\top\beta_S+\varepsilon_i$, so $q_i\sim\mathcal N(0,s_q^2)$ and $q_i$ is uncorrelated with, hence independent of, $\tilde x_{iS}$. Finally
\[
\operatorname{Cov}(e_{ik},\varrho_i)=a_k-a_kP_\nu-c_k^\top\Sigma_{SS}+c_k^\top TP_\nu=a_k(I-P_\nu)=\nu^2a_kT^{-1},
\]
so $\operatorname{Cov}(e_{ik},q_i)=\nu^2a_kT^{-1}\beta_S=\mu_k$. Joint Gaussianity and the vanishing cross-covariances with $\tilde x_{iS}$ give the independence of $\tilde x_{iS}$ and $(e_{ik},q_i)$.
\end{proof}
\section*{Proof of Theorem~\ref{thm:gauss_rate}}
\begin{proof}
(Here $E$ and $\xi$ are defined afresh and differ from the symbols of the same name in the proof of Theorem~\ref{thm:single_level}.) Write $E\coloneq\widehat C-T$ and let $\xi\coloneq\sqrt{n_1}\,\widehat C^{-1/2}h/s_q$. We work on the intersection of the following events, each defined on the intersection of the preceding ones where indicated (the union bound is applied to $\mathcal A_1^c$, $\mathcal A_1\cap\mathcal A_2^c$, and so on); their probabilities are computed with Lemma~\ref{lem:gstruct}.

$\mathcal A_1$: $\|E\|_2\le\|T\|(2r_s+r_s^2)$. Writing $\widetilde X_S=ZT^{1/2}$ with $Z$ an $n_1\times s$ standard Gaussian matrix, $E=T^{1/2}(Z^\top Z/n_1-I_s)T^{1/2}$, and by the singular-value bounds for Gaussian matrices \citep[Cor.~5.35]{vershynin2012introduction} with $t=\sqrt{2\ell}$, $\|Z^\top Z/n_1-I_s\|_2\le2r_s+r_s^2$ with probability at least $1-2e^{-\ell}$. By (ii), $r_s\le1/6$, so $\|E\|_2\le3\|T\|r_s\le\lambda_T/2$; hence $\lambda_{\min}(\widehat C)\ge\lambda_T/2$, $\|\Omega\|_2\le2/\lambda_T$ and $\|\widehat C\|_2\le\tfrac32\|T\|$.

$\mathcal A_2$: $|\|q\|_2^2/n_1-s_q^2|\le s_q^2r_q$, of probability at least $1-2e^{-\ell}$, since $\|q\|_2^2/s_q^2\sim\chi^2_{n_1}$ \citep[Lemma~1]{laurent2000adaptive}.

$\mathcal A_3$: $\|\xi\|_2^2\le G_s$. Since $q$ is independent of $\widetilde X_S$, conditionally on $\widetilde X_S$ we have $h=(s_q/\sqrt{n_1})\widehat C^{1/2}\xi$ with $\xi\sim\mathcal N(0,I_s)$, so $\xi$ is standard Gaussian and independent of $\widetilde X_S$, and $\mathbb P(\|\xi\|_2^2>G_s)\le e^{-\ell}$ by \citet[Lemma~1]{laurent2000adaptive}.

$\mathcal A_4$: $|h^\top\Omega z|\le s_q\sqrt{4\ell s/(\lambda_Tn_1)}$. Given $\widetilde X_S$, $h^\top\Omega z\sim\mathcal N(0,s_q^2z^\top\Omega z/n_1)$ and $z^\top\Omega z\le s\|\Omega\|_2\le2s/\lambda_T$ on $\mathcal A_1$; the Gaussian tail bound gives probability at least $1-2e^{-\ell}$.

On $\mathcal A_1\cap\mathcal A_3$: $\|h\|_2^2=(s_q^2/n_1)\xi^\top\widehat C\xi\le\bar h^2$ and $h^\top\Omega h=(s_q^2/n_1)\|\xi\|_2^2\le s_q^2G_s/n_1$.

$\mathcal A_5$: the variables $N_k$ defined below satisfy $|N_k|\le\sqrt{2\ell}$ for all $k\in S^c$, with probability at least $1-2(p-s)e^{-\ell}$ (union bound).

By Lemma~\ref{lem:gstruct}, conditionally on $(\widetilde X_S,q)$ the vector $e_k$ has independent entries with $e_{ik}\sim\mathcal N(\mu_kq_i/s_q^2,\varpi_k^2)$, $\varpi_k^2\coloneq s_k^2-\mu_k^2/s_q^2\le s_k^2\le\sigma_x^2$; as $\psi$ is a function of $(\widetilde X_S,q)$, $e_k^\top\psi=(\mu_k/s_q^2)\,q^\top\psi+\varpi_k\|\psi\|_2N_k$ with $N_k$ standard normal given $(\widetilde X_S,q)$ (if $\varpi_k=0$, let $N_k$ be an independent standard normal variable; the identity then holds trivially). Since $q^\top\psi=h^\top\Omega(z-h/\lambda)+\|q\|_2^2/(\lambda n_1)$,
\[
F_k=\frac{\mu_k}{s_q^2}\Big(h^\top\Omega z-\frac{h^\top\Omega h}\lambda+\frac{\|q\|_2^2/n_1-s_q^2}{\lambda}\Big)+\varpi_k\|\psi\|_2N_k .
\]
On the events above the bracket is at most $\Lambda_\nu$ in absolute value. Moreover $\|\psi\|_2\le\|\widetilde X_S\Omega(z-h/\lambda)\|_2/n_1+\|q\|_2/(\lambda n_1)$, where $\|\widetilde X_S\Omega m\|_2^2/n_1^2=m^\top\Omega m/n_1\le2\|m\|_2^2/(\lambda_Tn_1)$ with $\|m\|_2\le\sqrt s+\bar h/\lambda$ for $m=z-h/\lambda$, and $\|q\|_2\le s_q\sqrt{n_1(1+r_q)}$; so $\|\psi\|_2\le\Xi_\nu$. Hence $\max_k|F_k|\le(\bar\mu/s_q^2)\Lambda_\nu+\sigma_x\sqrt{2\ell}\,\Xi_\nu=\Phi_\nu\le\eta/2$ by (iii). The probability of the union of the exceptional events is at most $(7+2(p-s))e^{-\ell}\le9pe^{-\ell}=\alpha$.

\emph{Dual feasibility.} By Lemma~\ref{lem:decomp} and (i), $\|\hat z_{S^c}\|_\infty\le\max_k|c_\nu^k|+\max_k|F_k|\le1-\eta+\eta/2<1$.

\emph{Sign consistency.} By Lemma~\ref{lem:decomp}, $\|\hat\beta_S-\beta^\star_\nu\|_\infty\le\|\hat\beta_S-\beta^\star_\nu\|_2\le\|\Omega\|_2\big(\|h\|_2+\lambda\|E\|_2\|T^{-1}z\|_2\big)\le\frac2{\lambda_T}\big(\bar h+3\lambda\kappa_Tr_s\sqrt s\big)$, using $\|E\|_2\le3\|T\|r_s$ and $\|T^{-1}z\|_2\le\sqrt s/\lambda_T$. By (iv), $z_j\hat\beta_j\ge\theta_\nu-\|\hat\beta_S-\beta^\star_\nu\|_\infty>0$ for all $j\in S$.

Together with the invertibility of $\widehat C$, strict dual feasibility and the correct signs imply that $\hat\beta=(\hat\beta_S,0)$ is the unique Lasso solution and has support $S$ \citep[Lemma~2]{wainwright2009sharp}; see also \citet[Lemma~1]{zhao2006model} for the KKT conditions.
\end{proof}
\section*{Proofs of Theorem~\ref{thm:exact} and Corollary~\ref{cor:zeroone}}
\begin{proof}[Proof of Theorem~\ref{thm:exact}]
Since $n_1\ge s$, $\widehat C$ is invertible almost surely ({$n_1\ge s+1$ keeps $\chi^2_{n_1-s}$ non-degenerate}). Let $H=\widetilde X_S\Omega\widetilde X_S^\top/n_1=\widetilde X_S(\widetilde X_S^\top\widetilde X_S)^{-1}\widetilde X_S^\top$, the orthogonal projection onto the column space of $\widetilde X_S$. As $h=\widetilde X_S^\top q/n_1$, we have $\widetilde X_S\Omega h/n_1=Hq/n_1$, so that $\psi=\widetilde X_S\Omega z/n_1+(I-H)q/(\lambda n_1)$. The two summands are orthogonal, and $q^\top\widetilde X_S\Omega z/n_1=h^\top\Omega z$; hence
\[
\|\psi\|_2^2=\frac{z^\top\Omega z}{n_1}+\frac{\|(I-H)q\|_2^2}{\lambda^2n_1^2},\qquad q^\top\psi=h^\top\Omega z+\frac{\|(I-H)q\|_2^2}{\lambda n_1}.
\]
By Lemma~\ref{lem:gstruct}, the rows of $\widetilde X_S$ are i.i.d.\ $\mathcal N(0,T)$, so $\widehat C$ has the law in the statement, and $q$ is independent of $\widetilde X_S$ with i.i.d.\ $\mathcal N(0,s_q^2)$ entries. Conditionally on $\widetilde X_S$, $Hq$ and $(I-H)q$ are independent Gaussian vectors, $h=\widetilde X_S^\top Hq/n_1\sim\mathcal N(0,s_q^2\widehat C/n_1)$, and $\|(I-H)q\|_2^2/s_q^2\sim\chi^2_{n_1-s}$. Thus $(\widehat C,h,\|(I-H)q\|_2^2/s_q^2)$ has the law of $(\widehat C,h,\Upsilon)$ in the statement (with $\xi=\sqrt{n_1}\,s_q^{-1}\widehat C^{-1/2}h$).

Let $e_i=(e_{ik})_{k\in S^c}$. By Lemma~\ref{lem:gstruct}, $(e_i,q_i)$ is Gaussian and independent of $\tilde x_{iS}$, with $\operatorname{Cov}(e_i)=\Sigma_{S^cS^c}+\nu^2I-\Sigma_{S^cS}T^{-1}\Sigma_{SS^c}$ (a direct computation with $e_{ik}=\tilde x_{ik}-c_k^\top\tilde x_{iS}$), $\operatorname{Cov}(e_i,q_i)=\mu$ and $\operatorname{Var}(q_i)=s_q^2$. Conditionally on $q_i$, $e_i\sim\mathcal N(\mu q_i/s_q^2,C_N)$. Since the rows are independent, conditionally on $(\widetilde X_S,q)$ the vector $(e_k^\top\psi)_{k\in S^c}$ is Gaussian with mean $(\mu/s_q^2)\,q^\top\psi$ and covariance $\|\psi\|_2^2C_N$, so $(e_k^\top\psi)_k=(\mu/s_q^2)q^\top\psi+\|\psi\|_2N$ with $N\sim\mathcal N(0,C_N)$ independent of $(\widetilde X_S,q)$. By Lemma~\ref{lem:decomp}, $\hat z_k=c_k^\top z+e_k^\top\psi$, that is $\hat z_{S^c}=\Sigma_{S^cS}T^{-1}z+e^\top\psi$, and $\hat\beta_S=P_\nu\beta_S+\Omega(h-\lambda z)$. Substituting the expressions above for $\|\psi\|_2$ and $q^\top\psi$ gives the stated representation, and the independence of $(\widehat C,\xi,\Upsilon,N)$ follows from the conditional structure.

For the bounds: if $\operatorname{sign}(\hat\beta_S)=z$ and $\|\hat z_{S^c}\|_\infty<1$, then, $\widehat C$ being invertible, $\hat\beta=(\hat\beta_S,0)$ is the unique Lasso solution and has support $S$ and sign $z$ \citep{wainwright2009sharp}. Conversely, a solution with support $S$ and sign $z$ must satisfy the stationarity conditions on $S$, which force $\hat\beta_S=\Omega(r_S-\lambda z)$, and its dual vector on $S^c$ equals $\hat z_{S^c}$ and has sup-norm at most $1$. Finally, given $(\widehat C,\xi,\Upsilon)$ each $\hat z_k$ is Gaussian with variance $\|\psi\|_2^2\varpi_k^2>0$ when $\varpi_k>0$, so $\mathbb P(|\hat z_k|=1)=0$ and the two bounds coincide.
\end{proof}
\begin{proof}[Proof of Corollary~\ref{cor:zeroone}]
As $n_1\to\infty$ with $(p,s,\Sigma,\beta,\sigma,\lambda,\nu)$ fixed, $\widehat C\to T$, $h\to0$, $\Upsilon/n_1\to1$ and $\|\psi\|_2\to0$ in probability. Hence $\hat z\to\Sigma_{S^cS}T^{-1}z+\mu/\lambda=\hat c_\nu$, since $\mu_k/\lambda=\frac{\nu^2}{\lambda}a_kT^{-1}\beta_S$, and $\hat\beta\to P_\nu\beta_S-\lambda T^{-1}z=\beta^\star_\nu$. If $c_\nu=\|\hat c_\nu\|_\infty<1$ and $\theta_\nu>0$, the event $\{\operatorname{sign}(\hat\beta)=z,\|\hat z\|_\infty<1\}$ has probability tending to $1$, and Theorem~\ref{thm:exact} gives the claim. If $c_\nu>1$ or $\theta_\nu<0$, the event $\{\operatorname{sign}(\hat\beta)=z,\|\hat z\|_\infty\le1\}$ has probability tending to $0$.
\end{proof}
\section*{Proof of Theorem~\ref{thm:delta_avg_recovery}}
\begin{proof}
Fix $\delta\in\mathcal G$ and $j\in\{1,\dots,p\}$. Conditional on the data $\mathcal D$, the indicators
$ 1\{\hat\beta^{(b,\delta)}_j\neq 0\}$ are i.i.d.\ Bernoulli with mean $f_j^{(\delta)}$, because the pairs $(I_b,W^{(b)})$ are drawn independently across $b$ and the population frequency $f_j^{(\delta)}$ is defined over both the subsample and the noise. By Hoeffding's
inequality \citep[Theorem~2]{409cf137-dbb5-3eb1-8cfe-0743c3dc925f}, for any $t>0$,
\[
\mathbb{P}\!\left(\left|\hat f_j^{(\delta)}-f_j^{(\delta)}\right|>t\right)\le 2e^{-2Bt^2}.
\]
Applying a union bound over all $mp$ pairs $(\delta,j)$ gives
\[
\mathbb{P}\!\left(\max_{\delta\in\mathcal G}\max_{1\le j\le p}\left|\hat f_j^{(\delta)}-f_j^{(\delta)}\right|>\varepsilon_{\mathrm{MC}}\right)
\le 2mp\,e^{-2B\varepsilon_{\mathrm{MC}}^2}=\alpha_{\mathrm{MC}},
\]
with $\varepsilon_{\mathrm{MC}}$ as defined. Hence, with probability at least $1-\alpha_{\mathrm{MC}}$, $|\hat f_j^{(\delta)}-f_j^{(\delta)}|\le \varepsilon_{\mathrm{MC}}$ for all $\delta\in\mathcal G$ and all $j$.
On this event, averaging with weights that sum to one preserves the same uniform deviation bound:
\[
\begin{aligned}
\left|\widehat{\bar f}_j-\bar f_j\right|
&=
\left|
\sum_{\delta\in\mathcal G}w_\delta\big(\hat f_j^{(\delta)}-f_j^{(\delta)}\big)
\right|\\
&\le
\sum_{\delta\in\mathcal G}w_\delta\left|\hat f_j^{(\delta)}-f_j^{(\delta)}\right|
\le \varepsilon_{\mathrm{MC}}.
\end{aligned}
\]
Therefore,
\[
\min_{j\in S}\widehat{\bar f}_j
\ge \min_{j\in S}\bar f_j-\varepsilon_{\mathrm{MC}},
\qquad
\max_{k\in S^c}\widehat{\bar f}_k
\le \max_{k\in S^c}\bar f_k+\varepsilon_{\mathrm{MC}}.
\]
By Assumption~\ref{as:avg_sep}, $\min_{j\in S}\bar f_j\ge \max_{k\in S^c}\bar f_k+\gamma$, hence
\[
\min_{j\in S}\widehat{\bar f}_j-\max_{k\in S^c}\widehat{\bar f}_k
\ge \gamma-2\varepsilon_{\mathrm{MC}}.
\]
If $\varepsilon_{\mathrm{MC}}<\gamma/2$, then
$\min_{j\in S}\widehat{\bar f}_j>\max_{k\in S^c}\widehat{\bar f}_k$, so any
$\tau$ satisfying $\max_{k\in S^c}\widehat{\bar f}_k<\tau\le \min_{j\in S}\widehat{\bar f}_j$ yields
$\hat S(\tau)=S$.
\end{proof}
\section*{Proof of Theorem~\ref{thm:gap}}
\begin{proof}
Work on the event of Theorem~\ref{thm:delta_avg_recovery}, on which $\max_j|\widehat{\bar f}_j-\bar f_j|\le\varepsilon_{\mathrm{MC}}$. Sorting is $1$-Lipschitz in the sup-norm, so $|\widehat{\bar f}_{(k)}-\bar f_{(k)}|\le\varepsilon_{\mathrm{MC}}$ for every $k$, and therefore each adjacent gap changes by at most $2\varepsilon_{\mathrm{MC}}$. For $k\neq s$, $\widehat{\bar f}_{(k)}-\widehat{\bar f}_{(k+1)}\le g_{\mathrm{in}}+2\varepsilon_{\mathrm{MC}}$. For $k=s$, since $\bar f_{(s)}=\min_{j\in S}\bar f_j$ and $\bar f_{(s+1)}=\max_{k\in S^c}\bar f_k$ under Assumption~\ref{as:avg_sep}, $\widehat{\bar f}_{(s)}-\widehat{\bar f}_{(s+1)}\ge\gamma-2\varepsilon_{\mathrm{MC}}$. The hypothesis $\gamma>g_{\mathrm{in}}+4\varepsilon_{\mathrm{MC}}$ gives $\gamma-2\varepsilon_{\mathrm{MC}}>g_{\mathrm{in}}+2\varepsilon_{\mathrm{MC}}$, so the unique maximiser is $\hat s=s$. Finally $\gamma>2\varepsilon_{\mathrm{MC}}$, so $\min_{j\in S}\widehat{\bar f}_j>\max_{k\in S^c}\widehat{\bar f}_k$; then $\widehat{\bar f}_{(s)}$ and $\widehat{\bar f}_{(s+1)}$ are these two extreme values, the midpoint $\hat\tau$ lies strictly between them, and $\hat S=S$.
\end{proof}
\section*{Proof of Theorem~\ref{thm:errctrl}}
\begin{proof}
Fix $k\in L$, a level $\delta$ and $b\in\{1,\dots,B\}$. Let $a_{b,\delta}\coloneq\frac12\big(1\{k\in\hat S(A_{2b-1};\delta)\}+1\{k\in\hat S(A_{2b};\delta)\}\big)\in\{0,\tfrac12,1\}$ and $t_{b,\delta}\coloneq1\{k\in\hat S(A_{2b-1};\delta)\}\,1\{k\in\hat S(A_{2b};\delta)\}$. Then $a_{b,\delta}\le(1+t_{b,\delta})/2$, hence
\[
\widehat{\bar f}^{\,\mathrm{cp}}_k=\sum_{\delta}w_\delta\frac1B\sum_{b=1}^Ba_{b,\delta}\ \le\ \frac12+\frac12T_k,\qquad T_k\coloneq\sum_\delta w_\delta\frac1B\sum_{b=1}^Bt_{b,\delta}.
\]
If $\widehat{\bar f}^{\,\mathrm{cp}}_k\ge\tau$ then $T_k\ge2\tau-1>0$, so Markov's inequality gives $\mathbb P(k\in\hat S^{\mathrm{cp}}_\tau)\le\mathbb E(T_k)/(2\tau-1)$. The two halves of a pair are disjoint sets of i.i.d.\ observations and receive independent jitter noise, so given the pair the two selections are independent; since the observations are identically distributed, each of them has probability $p_k(\delta)$, whatever the half-sample. Hence $\mathbb E(t_{b,\delta})=p_k(\delta)^2$ and $\mathbb E(T_k)=\sum_\delta w_\delta\,p_k(\delta)^2$. Summing $\mathbb P(k\in\hat S^{\mathrm{cp}}_\tau)\le\mathbb E(T_k)/(2\tau-1)$ over $k\in L$ gives the first inequality of the first claim. The second follows from $p_k(\delta)^2\le\theta_L\,p_k(\delta)$ for $k\in L$ and $\sum_{k\in L}p_k(\delta)=\mathbb E|\hat S(A;\delta)\cap L|$. For the second claim, exchangeability gives $p_k(\delta)=\bar q_\delta/(p-s)$ for every $k\in S^c$, so that with $L=S^c$ we have $\sum_{k\in S^c}p_k(\delta)^2=\bar q_\delta^{\,2}/(p-s)$, which gives the first bound; the second follows from $\bar q_\delta^{\,2}\le\max_\delta\bar q_\delta\cdot\bar q_\delta$.
\end{proof}
\end{appendices}
\end{document}